%% file: trace-poly.tex
\documentclass[preprint,9pt,nonatbib]{sigplanconf}

\usepackage[compact]{titlesec}

\usepackage[utf8]{inputenc}
\usepackage[T1]{fontenc}
\usepackage{textcomp}
\usepackage{latexsym,amssymb,amsmath,ae,amscd,amsthm,stmaryrd,MnSymbol}
\usepackage{mathtools}
\usepackage{mathpartir}
\usepackage{xspace}   
\usepackage{paralist}
\usepackage{color}

\newcommand\cutout[1]{}
\DeclareMathAlphabet{\mathpzc}{OT1}{pzc}{m}{it}

\definecolor{dark-blue}{rgb}{0.2,0,1}

\input{macros}
\newcommand\mypar[1]{\paragraph{\normalfont\bf #1}}

\newtheorem{theorem}{Theorem}

\newtheorem{lemma}[theorem]{Lemma}

\theoremstyle{definition}
\newtheorem{remark}[theorem]{Remark}
\newtheorem{example}[theorem]{Example}
\newtheorem{definition}[theorem]{Definition}

\cutout{\let\oldendproof\endproof
\renewcommand{\endproof}{\qed\oldendproof}}

\newcommand\EE{\mathcal{E}}

\usepackage{times}
\usepackage{setspace}
\usepackage{hyperref}
\usepackage{paralist}
\usepackage[numbers,square]{natbib}
\def\bibfont{\small}

\newcommand\gj[1]{}
\newcommand\nt[1]{}
\newcommand\ntt[1]{{#1}}
\newcommand\gjt[1]{{#1}}

\begin{document}

\setlength{\pdfpageheight}{\paperheight}
\setlength{\pdfpagewidth}{\paperwidth}

\title{Trace semantics for polymorphic references\thanks{%
Research supported by the Engineering and Physical Sciences Research Council (EP/L022478/1) and the Royal Academy of Engineering. 
We thank T.~Cuvillier, J.~Rathke and the reviewers for comments and suggestions.
}}
\authorinfo{Guilhem Jaber}{Université Paris Diderot}{}
\authorinfo{Nikos Tzevelekos}{Queen Mary University of London}{}

 \maketitle

\begin{abstract}
We introduce a trace semantics for a call-by-value language with full polymorphism
and higher-order references. 
This is an operational game semantics model 
based on a nominal interpretation of
parametricity whereby polymorphic values are abstracted with special kinds of names.
The use of polymorphic references leads to violations of parametricity which we
counter by closely recoding the disclosure of typing information in the semantics.
We prove the model sound for the full language and
strengthen our result to full abstraction for a large fragment where polymorphic references obey specific inhabitation conditions.
\end{abstract}

\input{intro}
\input{syntax}
\input{simple}
\input{fulltrace}
\input{sound}
\input{fulla}

\bibliographystyle{abbrv}
\begin{spacing}{0.5}
\input{biblio}
\end{spacing}

\clearpage
\onecolumn
\appendix
\input{app-syntax}
\input{appendix-typecast}

\input{appendix-validconf-bis}
\input{app-validir}
\input{app-equiv}
\input{app-compconf}
\input{appendix-soundness}

\input{app-fulla}

\end{document}

%% file: macros.tex
\newcommand\inbrax[1]{\langle#1\rangle}

\newcommand\boldemph[1]{\emph{\textbf{#1}}}
\newcommand \SystemReF {{\rm System ReF}}

\newcommand\ee\varepsilon

\newcommand \Loc {\mathrm{Loc}}
\newcommand \Var {\mathrm{Var}}

\newcommand \Val {\mathrm{Val}}
\newcommand \Types {\mathrm{Types}}
\newcommand \TVar {\mathrm{TVar}}

\newcommand \AValues {\mathrm{AValues}}

\newcommand \Univ {\mathcal{U}}

\newcommand \red {\rightarrow}

\newcommand \Int {\mathrm{Int}}
\newcommand \Unit {\mathrm{Unit}}

\newcommand \unit {()}

\newcommand \nb [1] {#1}

\newcommand \proj [1] {\pi_{#1}}

\newcommand \pack [1] {\mathsf{pack} #1}
\newcommand \unpack [3] {\mathsf{unpack}\, #1\, \mathsf{as}\, #2\, \mathsf{in}\, #3}
\newcommand \hole {\bullet}
\newcommand\assert{\mathsf{assert}}

\newcommand \emptyseq \varepsilon
\newcommand \comp [1] {\mathbf{comp}(#1)}

\newcommand \ifte [3] {\mathsf{if}\ #1\ #2\ #3}
\newcommand\nuref[1]{\mathsf{ref}\,#1}

\newcommand \ctxeqf[6] {#2; #1 ; #3 \vdash #4 \ctxeq #5 : #6 }
\newcommand \ciueqf[6] {#2; #1 ; #3 \vdash #4 \ciueq #5 : #6 }

\newcommand \ctxeq {\simeq}
\newcommand \ciueq {\simeq_{ciu}}

\newcommand \typingCtx[5] {#2 ; #1 \vdash #3 : #4 \rightsquigarrow #5 }
\newcommand \typingCtxTD[5] {#2 ; #1 \vdash_{e} #3 : #4 \rightsquigarrow #5 }
\newcommand \typingTerm[5] {#2; #1 ; #3 \vdash #4 : #5 }
\newcommand \typingTermTD[5] {#2 ; #1; #3 \vdash_{e} #4 : #5 }

\newcommand \typingSubst[3] {\typingTerm{#1}{\cdot}{\cdot}{#2}{#3}}
\newcommand \typingHeap[2] {#1 : #2 }
\newcommand \typingHeapTD[2] { #1 :_{e} #2 }

\newcommand \empHeap {\varepsilon}

\newcommand \empStore {\varepsilon}

\newcommand \ClosedHeap [1] {\mathsf{Cl}(#1)}

\newcommand \pair [2]{\langle #1,#2\rangle}

\newcommand \substfun [2] {#1:#2}

\newcommand \argT [1] {\mathsf{arg}(#1)}
\newcommand \retT [2] {\mathsf{ret}_{#2}(#1)}

\newcommand\Aloc{\Loc}
\newcommand\Atvar{\TVar}
\newcommand\Afun{\mathrm{Fun}}
\newcommand\Apol{\mathrm{Pol}}
\newcommand\AfunT [1] {\Afun_{#1}}
\newcommand\ApolT [1] {\Apol_{#1}}

\newcommand \FuncNamePointers \Afun
\newcommand \PolyNamePointers \Apol

\newcommand \AVal [2] {\mathsf{AVal}(#1,#2)}
\newcommand \AValD [3] {\mathsf{AVal}(#1,#2)_{#3}}
\newcommand \AStore [2] {\mathsf{AStore}(#1,#2)}
\newcommand \APEnv [2] {\mathsf{AEnv}(#1)_{#2}}

\newcommand \SemStore [1] {\mathsf{S}\!\sem{#1}}
\newcommand \SemPEnv [2] {\mathsf{E}\!\sem{#1}_{#2}}

\newcommand {\TDclos} [1] {\xbar{#1}}

\newcommand \sem [1] {\left\llbracket #1 \right\rrbracket}

\newcommand \configuration[5] {
    \langle #1,#2,#3,#4,#5
  \rangle}
  
\newcommand \configurationF[5] {
    \langle #1,#2,#3,#4,#5
  \rangle}

\newcommand \action [2] {#1,#2}
\newcommand \faction [3] {#1,#2,#3}

\newcommand \ansP [1] {\langle \bar{#1} \rangle}

\newcommand \ansO [1] {\langle #1 \rangle}
\newcommand \questP [2] {\bar{#1} \! \langle #2 \rangle}
\newcommand \questO [2] {#1 \! \langle #2 \rangle}

\newcommand \Tr {\mathrm{Tr}}
\newcommand \trace[1] {\mathsf{Tr}#1}

\newcommand \pref {\rightsquigarrow}

\newcommand \prefix {\sqsubseteq}

\newcommand \extend {\sqsupseteq}

\newcommand \restrictTr [2] {{#1}_{|#2}}

\newcommand \xredint [1] {\xrightarrow{#1}}
\newcommand \xRedint [1] {\xRightarrow{#1}}

\newcommand \emptyStack {\lozenge}

\newcommand \nomeq [1] {\sim_{#1}}

\newcommand \actionPerm {\ast}

\newcommand \Anl {\mathbb{A}}

\newcommand \support [2] {\nu_{#1}(#2)}

\newcommand\rmL{{\rm L}}
\newcommand\rmP{{\rm P}}
\newcommand\rmT{{\rm T}}
\newcommand\rmF{{\rm F}}

      \definecolor{light-gray}{gray}{0.9}

\newcommand \ie {\emph{i.e.}\xspace}

\newcommand \st {such that }

\newcommand \Power[1] {\mathcal{P}(#1)}

\newcommand \dom[1] {\mathrm{dom}(#1)}
\newcommand \codom[1] {\mathrm{cod}(#1)}

\newcommand \N {\mathbb{N}}
\newcommand \Z {\mathbb{Z}}

\newcommand \sep {~|~}

\newcommand \defeq {\overset{de\hspace{-0.2ex}f}{=}}

\newcommand \subst [2]{\{#2 / #1 \}}
\newcommand \refer {\mathrm{ref} \, }

\newcommand \parrow  {\rightharpoonup}

\newcommand*\xbar[1]{%
  \hbox{%
    \vbox{%
      \hrule height 0.5pt 
      \kern0.5ex
      \hbox{%
        \kern-0.1em
        \ensuremath{#1}%
        \kern-0.1em
      }%
    }%
  }%
}

\newcommand{\redsat}  {\rightsquigarrow}

\newcommand{\cast} [3]{\mathbf{cast}^{#1}_{#2\rightarrow#3}}

\newcommand{\RCast} [1] {\mathsf{Cast}(#1)}
\newcommand{\RCastt} [1] {\mathsf{Cast}^\circ\!(#1)}

\newcommand{\emptyPMap} {\epsilon}

\newcommand{\mergeConf} [2] {#1 \doublewedge #2}
\newcommand{\mergeStack} [2] {#1 || #2}

\newcommand{\TRel} [1] {\leq_{#1}}
\newcommand{\MGI} [1] {\mathsf{mgi}(#1)}

\newcommand{\equivT}{\sim}

\newcommand{\true}{\mathbf{true}}

\newcommand{\TConstr} [1] {\mathsf{TConstr}(#1)}

\newcommand{\Freshen} [3] {\mathcal{F}_{}(#3)} 
\newcommand{\Refresh} [2] {\mathrm{Refresh}_{#1}(#2)} 


%% file: intro.tex
\section{Introduction}
\newcommand \lang \SystemReF
\newcommand \letin [2] {\mathsf{let}\,#1 \,  \mathsf{in}\, #2}
\newcommand \Ifte [3] {\mathsf{if} \ #1 \ \mathsf{then} \ #2 \ \mathsf{else} \ #3}

Polymorphism is a prevalent feature of modern programming languages, 
allowing one to use generic data structures and 
powerful code abstractions. 
Reasoning with polymorphism is  both challenging and rewarding: polymorphic code is 
bound to have uniform behaviour under different instantiations, \ntt{a property known as \emph{Strachey parametricity}~\cite{Strachey67}} and
formalized by Reynolds as relational parametricity\cite{Reynolds83},
which in turn provides 
``theorems for free''~\cite{wadler1989theorems}.

Understanding the formal semantics of polymorphism amounts to capturing the parametric behaviour 
of code under different instantiations. 
This has traditionally been hard, effectively due to the 
requirement for a model where instantiations 
from within the same model are possible. 
As far as the full abstraction problem is concerned, 
the construction of fully abstract models has so far had successes in the game semantics  framework.
%
The problem has been addressed by use of \emph{hypergames} by Hughes~\cite{Hug_phd}, 
whereby game arenas can be seen as moves which can be opened inside enclosing arenas during a play. 
The model of Abramsky and Jagadeesan~\cite{AJ05} followed a different approach, 
namely that of fixing a universe of moves with holes, 
the latter representing type variables awaiting instantiation, 
and constructing arenas from that given pool of moves, which is effectively closed under instantiation. 
While these models addressed purely functional languages,
in recent years a remarkable research programme by Laird~\cite{Lai13,laird2010game} 
has extended the reach of polymorphic games to languages with higher-order state.

An important aspect of previous models~\cite{Hug_phd,AJ05,Lai13,laird2010game}, 
and of the modelled languages, is the uniformity of polymorphic behaviour. 
However, when we move to languages with mutable references that can extrude their scope, 
this property can be easily broken as we see below. 
Thus, the modelling of languages with ML- or Java-like references presents additional complications and, 
as far as we are aware, is still open. 
Our paper addresses precisely this problem. 

The language we analyse, \lang, includes a typed lambda calculus with products, references and polymorphism. 
For instance, we can examine the following type.
\[
\forall\alpha.\,(\refer\alpha\times\refer\Int)\to\alpha
\]
One may be tempted to think that any term inhabiting this type is bound to return, 
given input $(x,y)$, the value stored in $x$. 
Of course, this is not necessarily the case if, for example, 
$\alpha$ is instantiated with $\Int$ and $x$ and $y$ happen to represent the same location. The following term would take advantage of such a coincidence,
\[
\Lambda\alpha.\lambda\pair{x}{y}^{\refer\alpha\times\refer\Int}\!.\,{y:=42; !x}
\]
and in that case return 42 regardless of what the initial value stored in $x$ was. 
Thus, in this example, the given coincidence leads to an accidental interference with the returned result. 
More interestingly, we can instrument our example in a way that it can \emph{discover} such coincidences and effectively deduce that $\alpha=\Int$.
Let us write $y{+}{+}$ below for $y:={!y}+1$.
\begin{align*}
\Lambda\alpha.\lambda\pair{x}{y}^{\refer\alpha\times\refer\Int}\!.\;
&\letin{\;x'=\nuref{!x},\,y'=\nuref{!y}\;}{}\\
&y{+}{+};\,x:=x';\\
&\Ifte{{!y'\!={!y}}}{({y:=42; !x})}{!x}
\end{align*}
The term above increases the value of $y$ and then 
restores $x$ to its initial value $x'$. It then
compares the value of $y$ with its initial one $y'$. If these are not the same, then $x$ and $y$ are different locations, so the value of $x$ is returned. 
If, however, the value of $y$ has not changed then the term has successfully discovered that $x$ and $y$ refer to the same location, whence 42 is returned.

The above example demonstrates that uniform polymorphic behaviour can be violated through references, as differently typed variables can be instantiated with a 
common reference. 
More than that, references can disclose type instantiation information which can then be taken advantage of by a polymorphic function. 
In our example above the result of this disclosure was a non-parametric return value of 42, 
but we can imagine scenarios where a term records the references $x$ and $y$ that allowed it to escape uniform behaviour, and uses them as a general-use ``bridge'' 
between values of type $\alpha$ and $\Int$. In fact, such devices, called \emph{casting functions}, shall play a central role in our semantics. 
More generally, our modelling approach is crafted around carefully keeping track of the type information 
that has been leaked from the program to its environment, and viceversa, and allowing moves to be played in accordance with that information assuming 
that the context (the Opponent) has the epistemic power to exploit all such leaked information.

\paragraph{Related work}

Operational techniques have been designed to study languages
with both polymorphism and references.
Realizability models~\cite{ahmed04,vmm,birkedal2010realisability}, later refined into Kripke logical relations~\cite{amal:statedep,dreyer:JFP2012},
use a notion of ``world as heap-invariant'' to model references. 
Environmental bisimulations have also been designed to deal with equivalence of programs
in such languages~\cite{sumii2009complete}.
{While complete, these approaches 
partially rely on context quantifications and in particular
do not directly account for the interaction between polymorphism and references, and the kind of type disclosure that the latter brings in.}

Our approach follows the line of research on trace semantics for higher-order languages~\cite{JR99,JR05,laird2007fully,GT12},
which in turn can be seen as an operational reformulation of game semantics~\cite{tst,Jab15}, on one hand; and of open bisimulation techniques~\cite{Lassen05,JagadeesanPR07,LassenL08}, on the other.
%
In this area, Jeffrey and Rathke 
proposed a fully abstract trace semantics for a polymorphic variant of the pi-calculus~\cite{JR08}, which refined a previous sound model of Pierce and Sangiorgi~\cite{pierce2000behavioral}.
That work is related to ours in spirit, and it already raises the intricacies involved in combining polymorphism with name equality testing.
However, 
the apparatus of {\em loc.\,cit.} does not lend itself to ML-like languages like \lang,
as in 
the latter we need stronger
semantic abstractions to cater for
the less expressive syntactic contexts. 
Overall, there seems to be a greater picture behind this work and~\cite{JR08,LassenL08} which remains to be exposed.

\paragraph{Future directions}
In this work we addressed Church-style polymorphism.
It would be interesting to examine whether our ideas could be adapted 
to deal with {Curry style}.
In doing so, we would give a semantic reading
of the \emph{value restriction}, which ensures type safety by enforcing terms of polymorphic types
to be values.
This, along with the study of ML-specific restrictions like {rank-1 polymorphism},
would bring us closer to modelling a large fragment of ML, which can be seen as a broader goal behind this work.

Moreover, our current model sets the foundation for a sound, and complete for a large collection of types, proof methods for program equivalence. Similarly to our previous work on monomorphic languages\,\cite{proof1,proof2}, we aim to explore such methods and accompany\!\! them with automated, or semi-automated, equivalence checkers.


\cutout{
In another line of work, we would like to give a \emph{compositional} presentation of our model.
First, we could simply relate it to Laird's model~\cite{laird2010game}, inferring copycat links
using freshness conditions on polymorphic names. Then, the main difficulty
would be to incorporate \emph{nominal} references to this model, and the disclosure process of types
that follows on from.
}


%% file: syntax.tex
\section{System ReF}
\label{sec:language}

\begin{figure}[t]
\begin{center}
$\begin{array}{@{}r@{\ }c@{\ }l}
    v,u &::=& \; \unit \sep \nb{n} \sep x
    \sep l \sep \lambda x^\theta.M \sep \Lambda \alpha.M \sep \pack{\pair{\theta}{v}}
\sep \pair{v}{u}\\[2mm]
    M,N & ::= & 
\; v 
\sep M N \sep M \theta \sep M\oplus N \sep \ifte{M_1}{M_2}{M_3}  \sep
 \pair{M}{N}\\
& &\sep \proj{1}(M) \sep \proj{2}(M) \sep \Omega_\theta 
 \sep\nuref M \sep !M \sep M := N \\
&& \sep    M = N
\sep \pack{\pair{\theta}{M}} \sep \unpack{M}{\pair{\alpha}{x}}{N} \\[2mm]
%
    E &::=& \;
\hole \sep E M \sep E \theta \sep v E \sep E \,{\sf op}\, M \sep v \,{\sf op}\, E \sep \ifte{E}{M}{M'}\\
&&   \sep \nuref E 
 \sep !E  \sep \pair{E}{M}
 \sep \pair{v}{E} \sep \proj{1}(E)\sep \proj{2}(E) \\
&&\sep \pack{\pair{\theta}{E}} \sep \unpack{E}{\pair{\alpha}{x}}{M}
\end{array}$\vspace{-.75em}
\end{center}
\caption{\SystemReF \
($n \in \Z$, $l \in \Loc$ and ${\sf op}\in\{\oplus,=,:=\}$).} 
\label{fig-def}\vspace{-1em}
\end{figure}

We introduce \SystemReF, a polymorphic call-by-value $\lambda$-calculus with {higher-order references}. The types of \SystemReF\ are:
\[
    \theta, \theta'  ::=  \;\alpha\sep \Unit \sep \Int 
    \sep \refer{\theta} \sep 
    \theta \times \theta' \sep \theta \rightarrow \theta' \sep \forall \alpha.\theta \sep \exists \alpha.\theta
\]
where $\alpha \in \TVar$, and  $\TVar$ a countably infinite set of type variables.
As usual,
a type is {closed} if all its type variables $\alpha$ are bound.
We shall call arrow and universal types \boldemph{function types}.
The syntax of values $v$, terms $M$ and evaluation contexts $E$ is
given in Figure~\ref{fig-def}. 
We assume a countably infinite set $\Loc$ of \emph{locations} and some standard collection of binary integer operators, which we generally denote by $\oplus$. 
We use the following macros: ${\sf let}\ x = N\ {\sf in}\ M$ stand for $(\lambda x.M)N$;
and $N; M$ means $(\lambda x.M) N$ with $x$ fresh in $M$.

The typing rules for \SystemReF\ include standard rules for functions and projections, rules for integers, and rules for polymorphism and references  
given in Figure~\ref{fig-typingrules}.
Typing judgments are of the form
$\typingTerm{\Sigma}{\Delta}{\Gamma}{M}{\theta}$, where
$\Sigma$ is a location context, i.e.\ a finite partial function from locations to \emph{closed} types;
$\Gamma$ a variable context;
and $\Delta$ a set of type variables {containing all free type variables of $\Gamma$}.
Given a closed evaluation context $E$, we write $\typingCtx{\Sigma}{\Delta}{E}{\theta}{\theta'}$
when $\typingTerm{\Sigma}{\Delta}{x:\theta}{E[x]}{\theta'}$.
Compared to the ML type-system, we work with \emph{Church-style} polymorphism,
where type abstractions and applications are explicit. This explains why we do \ntt{not} need
the so-called \emph{value restriction}~\cite{wright1995simple} to accommodate references.

\begin{figure*}[t]
\begin{gather*}
  \inferrule*{(l:\theta) \in \Sigma}{\typingTerm{\Sigma}{\Delta}{\Gamma}{l}{\refer\theta}}
\qquad             
  \inferrule*{\typingTerm{\Sigma}{\Delta}{\Gamma}{M}{\refer \theta}}{\typingTerm{\Sigma}{\Delta}{\Gamma}{!M}{\theta}}
\qquad
  \inferrule*{\typingTerm{\Sigma}{\Delta}{\Gamma}{M}{\refer \theta} \quad
              \typingTerm{\Sigma}{\Delta}{\Gamma}{N}{\theta}}
             {\typingTerm{\Sigma}{\Delta}{\Gamma}{M := N}{\Unit}}
\qquad
  \inferrule*{\typingTerm{\Sigma}{\Delta}{\Gamma}{M}{\refer \theta} \quad
              \typingTerm{\Sigma}{\Delta}{\Gamma}{N}{\refer \theta'}}
             {\typingTerm{\Sigma}{\Delta}{\Gamma}{M = N}{\Int}}
\\
  \inferrule*
  {\typingTerm{\Sigma}{\Delta,\alpha}{\Gamma}{M}{\theta}}
              {\typingTerm{\Sigma}{\Delta}{\Gamma}{\Lambda\alpha.M}{\forall\alpha.\theta}}
\quad
 \inferrule*{\typingTerm{\Sigma}{\Delta}{\Gamma}{M }{\forall\alpha.\theta}}
             {\typingTerm{\Sigma}{\Delta}{\Gamma}{M\theta'}{\theta\subst{\alpha}{\theta'}}}
\quad
\inferrule*{\typingTerm{\Sigma}{\Delta}{\Gamma}{M}{\theta\subst{\alpha}{\theta'}}}
             {\typingTerm{\Sigma}{\Delta}{\Gamma}{\pack{\pair{\theta'}{M}}}{\exists \alpha.\theta}}
\quad
\inferrule*
          {\typingTerm{\Sigma}{\Delta}{\Gamma}{M}{\exists \alpha.\theta}\quad  
           \typingTerm{\Sigma}{\Delta,\alpha}{\Gamma,x:\theta}{N}{\theta'}}
          {\typingTerm{\Sigma}{\Delta}{\Gamma}{\unpack{M}{\pair{\alpha}{x}}{N}}{\theta'}}
\\
\begin{array}{@{}rll@{\quad}rll@{\quad}rll@{}}
\hline\\[-2mm]
     (E[(\lambda x.M) v],S) &  \red  &  (E[M\subst{x}{v}],S) 
& (E[\Omega],S) & \red & (E[\Omega],S) 
&       (E[l = l],S)   &  \red  &  (E[1],S) 
\\
      (E[(\Lambda\alpha.M) \theta],S) &  \red  &  (E[M\subst{\alpha}{\theta}],S) 
&     (E[\pi_i \pair{v_1}{v_2}],S) & \red & (E[v_i],S) 
&        (E[l = l'],S)  &  \red  &  (E[0],S) 
\\
     (E[\ifte{\nb{n}}{M_1}{M_2}],S) &  \red  & (E[M_i],S)  
&
      (E[\nuref{v}],S)  &  \red  &   (E[l],S \cdot [l \mapsto v])
&      (E[!l],S)  &  \red  &  (E[S(l)],S) 
\\
      (E[l := v],S)  &  \red  &  (E[\unit],S[l \mapsto v]) 
&
\multicolumn{6}{c}{E[\unpack{(\pack{\pair{\theta}{ v}})}{\pair{\alpha}{x}}{M}] \red  E[M\subst{\alpha}{\theta}\subst{x}{v}]}
\\[-3mm]
\end{array}
\end{gather*}
\caption{{\sc Up}: Typing rules of \SystemReF\ (excerpt).
{\sc Down}:
Operational semantics (for $\mathtt{if}$: $i=2$ if $n=0$, otherwise $i=1$). 
}\label{fig-typingrules}\label{fig-red}\vspace{-1mm}
\end{figure*}

We next proceed with the operational semantics. Closed terms are reduced using stores containing their locations. More precisely, a
\boldemph{store} is a finite partial map $S:\Loc \parrow \Val$ from locations to values.
We define the following notation for stores, which we shall also be using for general partial maps:
\smallskip

\begin{compactitem}[$\bullet$]
\item The empty store is written $\empStore$. 
Adding a new element $(l,v)$ to a store $S$
is written $S \cdot [l \mapsto v]$, and is defined only if $l \notin
\dom{S}$.  
\item We also define $S[l \mapsto v]$, for $l \in
\dom{S}$, as the partial function $S'$ which satisfies $S'(l') = S(l')$
when $l'\neq l$, and $S'(l) = v$.
\item
The restriction of a store $S$ to a set of locations $L$ is written $S_{|L}$.
\end{compactitem}\smallskip

We write $S:\Sigma$ just if $\typingTerm{\Sigma}{\cdot}{\cdot}{S(l)}{\theta}$ for 
all  $l \in \dom{S}$. 
%
Given a set $L$ of locations and a store $S$, we define the image of $L$ by $S$, written $S^{*}(L)$,
as $S^*(L) = \bigcup_{j\in\omega} S^j(L)$ with 
$S^{j+1}(L) = S^j(L)\cup\{l\in\Loc\mid l\text{ contained in }S(S^j(L))\}$ 
and $S^0(L) = L$.
$S$ is called \boldemph{closed} just if $\dom{S}=S^*(\dom{S})$.
\cutout{We define 
\[
\ClosedHeap{\Sigma} = \{\ S \sep S:\Sigma \wedge \dom{S} = S^{*}(\dom{\Sigma})\ \}
\]
to be the set of minimal
closed stores whose domain contains all locations in $\Sigma$.
\nt{what is the purpose of typing stores and what are closed stores used for? Is it the case that in the operational semantics every pair $(M,S)$ must have $S$ closed and well-typed in the context of $M$?}
}

\begin{definition}
The operational semantics of \SystemReF\ involves pairs $(M,S)$ consisting of a closed term
$\typingTerm{\Sigma}{\Delta}{\cdot}{M}{\theta}$ and a closed store $S:\Sigma$.
Its small-step rules are given in Figure~\ref{fig-red}.
We write $(M,S) \Downarrow$ when $(M,S)\to^*(v,S')$ for some value $v$.
\end{definition}

\begin{remark}
We have equipped our language with a construct performing reference equality tests. 
This is in accordance with, and has the same operational semantics as, reference equality tests in ML, 
albeit extended to arbitrary reference types. Depending on type and type inhabitation, 
such tests can be encoded \ntt{in ML}
via appropriately crafted sequences 
of writes and reads in examined references.
\end{remark}


We finally introduce the notion of term equivalence we examine. 

\begin{definition}
Let $\Sigma$ be closed.
Two terms  $\typingTerm{\Sigma}{\Delta}{\Gamma}{M_1,M_2}{\theta}$ are \boldemph{contextually equivalent}, written $\ctxeqf{\Sigma}{\Delta}{\Gamma}{M_1}{M_2}{\theta}$, 
if for all contexts $C$, all  $\Sigma'\supseteq\Sigma$
and all closed $S:\Sigma'$
\st 
 $\typingTerm{\Sigma'}{\cdot}{\cdot}{C[M_i]}{\Unit}$,
we have
$(C[M_1],S) \Downarrow$ iff $(C[{M_2}],S)\Downarrow$. 
\end{definition}



%% file: simple.tex
\section{The Semantic Model}

Our trace model is constructed within nominal sets, that is, a universe embedded with atomic objects for representing locations, type variables, functions and polymorphic values. We introduce the semantic universe next and then proceed to the operational rules defining the semantics.

\subsection{Semantic Universe}

We define the set of \boldemph{names} to be:\vspace{-1mm}
\[
\Anl = \Loc \uplus \TVar
\uplus
\Afun
\uplus
\Apol\vspace{-1.25mm}
\]
where
each of the components is a countable set.
We range over elements of 
$\Loc$ by $l$ and variants; over $\TVar$ by $\alpha$, {\em etc}; over
$\Afun$ by $f,g$, {\em etc}; and over $\Apol$ by $p$, {\em etc}.
For each type $\theta$ and 
 type variable $\alpha$ we define:
$
\AfunT{\theta}=\Afun\times\{\theta\}$ and 
$\ApolT{\alpha}=\Apol\times\{\alpha\}
$.
By abuse of notation, we may write elements $(f,\theta)$ of $\AfunT{\theta}$ simply as $f$, and similarly for $\ApolT{\alpha}$.

Semantic objects feature elements of $\Anl$ as atomic entities which, moreover, can be acted upon by finite permutations of $\Anl$.
A \boldemph{nominal set}~\cite{gabbay2002new} is a pair $(X,\actionPerm)$ of a set $X$ along with an action ($\actionPerm$) from the set of finite 
component-preserving computations of $\Anl$ 
on the set $X$.\footnote{A finite permutation $\pi:\Anl\to\Anl$ is component-preserving 
if it preserves the partitioning of $\Anl$, e.g.\
if $d\in\Aloc$ then $\pi( d)\in\Aloc$.
} 
Given some $x\in X$, the set of names featuring in $x$ form its \emph{support}, written $\nu(x)$, which we stipulate to be finite. Formally, $\nu(x)$ is the smallest subset of $\Anl$ such that all permutations which elementwise fix $\nu(x)$ also fix $x$.
We shall sometimes write $\nu_{\rm C}(x)$, for ${\rm C}\in\{\rm L,T,F,P\}$ in order to select a specific kind of names from the support of $x$. For instance, $\support{\rm L}{x}=\nu(x)\cap\Loc$. 
\ntt{Using the same notation, we also write $\support{\rm T}{\theta}$ for the free type variables of $\theta$.}
We usually write $(X,\actionPerm)$ simply as $X$, for economy. 

We next introduce our basic semantic objects, which constitute
the semantic representations of syntactic values. 

\begin{definition}
We define \boldemph{abstract values} as:
\[
\AValues \ni\ v,u\ ::=\ ()\ |\ i\ |\ l\ |\ (f,\theta)\ |\ (p,\alpha)\ |\ \alpha\ |\ \pair{u}{v}
\]
where $i\in\Z$, $l\in\Aloc$, $f\in\Afun$, $p\in\Apol$ and $\alpha\in\TVar$. 
Note we still range over abstract values by $u,v$ (and hope no confusion arises).
We similarly set \boldemph{abstract stores} to be finite partial maps $\Loc \parrow \AValues$.
\end{definition}

Thus,
ground values (integers, $\unit$ and locations) are represented by their concrete values,
and for all other types but products we employ name abstractions. 
This abstraction is in order either because of polymorphism in the values, or simply because function code can only be examined by querying the given function.
Functions are represented by functional names, and polymorphic values by polymorphic names.

The semantics of a type $\theta$,
written $\sem{\theta}$, consists of pairs  $(v,\phi)$ of an abstract value $v$ along with a function $\phi : \support{\rmL}{v} \to \Power{\Types}$, and is given as:
\[
\begin{array}{l@{\;\ }l@{\;\;\;}l}
\sem{\Unit} & =& \{((),\emptyseq)\} 
\\
\sem{\Int} & =& \{(n,\emptyseq)\ |\ n\in\Z\} 
\\
\sem{\refer\theta} & =& \{(l,\{(l,\theta)\})\ |\ l\in\Aloc\} 
\\
\sem{\alpha} & =& \{(p,\emptyseq) \sep p \in \ApolT{\alpha}\}
\\ 
\sem{\theta\to\theta'} & =& \{(f,\emptyseq) \sep f \in \AfunT{\theta\to\theta'}\} 
\\
\sem{\forall\alpha.\theta} & =& \{(f,\emptyseq) \sep f \in \AfunT{\forall\alpha.\theta}\}
\\
\sem{\exists\alpha.\theta} & =& \{(\pair{\alpha'}{v},\phi) \sep (v,\phi) \in \sem{\theta\subst{\alpha}{\alpha'}}\}
\\ 
\sem{\theta_1\times\theta_2} & =& \{(\pair{v_1}{v_2},\phi_1\cup\phi_2)\ |\ (v_i,\phi_i)\in\sem{\theta_i}\}
\end{array}\]
%
%
The role of $\phi$ is to assign types to all the locations of an abstract value. 
As discussed in the Introduction, though, the same location can appear with several types in the execution of a given term phrase. 
Hence, $\phi$ assigns sets of types to each location instead of a unique type.
More generally, a \boldemph{typing function} is a finite map $\phi:\Aloc\rightharpoonup\Power{\Types}$.
%
The type translation is extended to typing environments by mapping each 
$\Delta=\{\alpha_1,\cdots,\alpha_k\}$, $\Sigma=\{l_1:\theta_1,\cdots,l_n:\theta_n\}$ and $\Gamma=\{x_1:\theta'_1,\cdots,x_k:\theta'_k\}$ to:
\[
\sem{\Delta,\Sigma,\Gamma} = \{\, ((\vec\alpha,\vec l,\vec v),\bigcup_{i=1}^n[l_i\mapsto\theta_i]\cup\bigcup_{j=1}^k\phi_j)\
|\ (v_j,\phi_j)\in\sem{\theta_j'}\, \}.
\]

\paragraph{Extending the syntax for $\Afun\cup\Apol$}
While functional and polymorphic names are not part of the syntax of \SystemReF, their involvement in its semantics makes it useful to introduce them as syntax as well. 
We  hence extend the set of values of \SystemReF\ to include 
elements
$(f,\theta)$ and $(p,\alpha)$ as typed constants.

\subsection{Interaction Reduction}

Traces will consist of sequences of \boldemph{moves} enriched with abstract stores and value disclosures.
Moves represent the interaction between the modelled program and its enclosing context and consist of function calls and returns.
Each move comes with a polarity: P for \emph{Player} (i.e.\ the program produces the move), and O for \emph{Opponent} (the context/environment).
There are four kinds of moves:
\begin{enumerate}
\item[\sc PQ.] \emph{Player Questions} are moves of the form $\questP{f}{u}$, representing a call to a functional name $f\in\Afun$ with argument $u\in\AValues$.
\item[\sc OQ.] \emph{Opponent Questions} are of the form $\questO{f}{u}$, with 
 $f\in\Afun$ and $u\in\AValues$; moreover, there are \emph{initial} opponent questions of the form $\questO{?}{u}$ ($u\in\AValues$). 
\item[\sc PA.] \emph{Player Answers} are moves of the form $\ansP{u}$, with  $u\in\AValues$.
\item[\sc OA.] \emph{Opponent Answers}, which are of the form $\ansO{u}$  ($u\in\AValues$). 
\end{enumerate}
On the other hand, value disclosures are partial functions $\rho$
representing the values of  polymorphic names revealed in a move. Their role will be explained in the next section.

\begin{definition}
A \boldemph{full move} is a triple $(\action{m}{S,\rho})$ of a move $m$, a {closed} abstract store $S$ and a finite map $\rho:\Apol\rightharpoonup\rm AValues$.
A sequence of full moves is called a \boldemph{trace}.
\end{definition}

\label{subsec:warmup-ir}

\cutout{
For a typing context $\Gamma$ and  a function $\gamma : \Var \parrow
\Val$, we say that $\gamma$ is an \emph{environment} on
$\Gamma$, written $\substfun{\gamma}{\Gamma}$, if $\gamma$ is defined
exactly on all the variables occurring in $\Gamma$, and $\gamma(x)$ is
a value of type $\theta$ whenever $(x,\theta) \in \Gamma$.
Then, the action of the environment $\gamma$ on a term $M$, defined as
$M\overrightarrow{\subst{x_i}{\gamma(x_i)}}$ with $x_i$ ranging over $\Gamma$,
is written $M\{\gamma\}$.
}

The trace semantics is produced via 
a reduction relation for open terms which only reveals the steps in the computation where there is {interaction}: a call or return between the term and its context.
More precisely, this relation is a bipartite labelled transition system between Player  
  and Opponent configurations, where labels are full moves, and whose main components are \boldemph{evaluation stacks} $\EE$,
defined as either:
\begin{compactitem}[$\bullet$]
 \item \emph{passive}, 
which are related to {Opponent configurations} and are of the shape $(E^n,\theta_n \rightsquigarrow \theta'_n)::\cdots::(E^1,\theta_1 \rightsquigarrow\theta'_1)$,
where each $E^i$ is an evaluation context of type $\theta_i \rightsquigarrow \theta'_i$;
 \item or \emph{active}, which are related to {Player configurations} and
are of the form $(M,\theta)::\EE'$, i.e.\ they consist of a term $M$ of type $\theta$
and a passive stack $\EE'$.
\end{compactitem}
The empty stack is written $\emptyStack$.

\begin{definition}
A \boldemph{configuration} is a tuple $\configuration{\EE}{\gamma}{\phi}{S}{\lambda}$ with:
\begin{compactitem}[$\bullet$]
 \item an evaluation stack $\EE$,
a typing function $\phi$ for locations,
and a {closed} store $S$,
 \item an \emph{environment} $\gamma$ mapping names to values,
 \item an \emph{ownership function} $\lambda\in (\Anl \times \{O,P\})^*$ ordering played names and mapping them to the party who has introduced them;
\end{compactitem}
and which satisfies the following conditions:
\begin{compactitem}[$\bullet$]
 \item the relation $\{(a,X)\mid \lambda =\lambda_1\cdot(a,X)\cdot\lambda_2\}$ is a partial function and $\lambda$ has no repetition of names
 \item $\dom{\gamma} = \{a \in \Apol \cup \Afun\cup\Atvar \sep \lambda(a)=P\}$
 \item $\dom{\phi} = \{l \in \Aloc \cap \dom{\lambda}\}\subseteq \dom{S}$
 \item for all $a \in \support{}{\EE,\codom{S},\codom{\gamma}} \backslash \Aloc$, $\lambda(a)=O$
\end{compactitem}
where, because of the first condition above, we write $\lambda(a)=X$ if 
$\lambda=\lambda_1\cdot(a,X)\cdot\lambda_2$ for some $\lambda_1,\lambda_2$.

In addition, we include special configurations of the form  $\inbrax{\typingTerm{\Sigma}{\Delta}{\Gamma}{M}{\theta}}$, one
for each typed term $\typingTerm{\Sigma}{\Delta}{\Gamma}{M}{\theta}$.
\end{definition}

Thus, a configuration registers syntactic and semantic information on the execution of a term necessary to produce its traces. $\cal E$ and $S$ are syntactic objects directly connected to the operational semantics.
The other components either are of semantic nature ($\phi,\lambda$) or bridge the semantics and the syntax ($\gamma$). In $\gamma$ we record the actual values that correspond to the functional and polymorphic names and type variables that the term (i.e.\ P) has produced. 
On the other hand, $\lambda$ is a name-polarity function which also keeps track of the order in which names were introduced. 
The last condition on $\lambda$ in the above definition is especially important: it stipulates that, except for location names, all the free names that appear in the term, either directly or indirectly via $\gamma$ or $S$, must belong to O. In other words, P cannot see the abstract values that he has provided to O during the interaction.

When the evaluation of a term $E[M]$ reaches, for example, some $E[fv]$ where $f$ is a function name provided by the context, a move asking the context to evaluate $f(v)$ will be produced. However, since $v$ is a syntactic value and in moves we only allow semantic entities, we need a way to pass from syntactic values to abstract ones. This is achieved as follows.
%
To each value $u$ of type $\theta$,
we associate the set $\AVal{u}{\theta}$ of triples $(v,\gamma,\phi)$, where each of them represents:
\begin{inparaitem}[$\bullet$]
\item a corresponding abstract value $v$;
\item an environment $\gamma$ instructing  the related mapping of names to values;
\item and a typing function $\phi$ recording the types used for each location in the translation. 
\end{inparaitem}
It is defined as:
\renewcommand\defeq{=}
\begin{align*}
&  \AVal{u}{\iota}  \defeq   \{(u,\emptyseq,\varnothing)\} 
  \quad \text{for } \iota=\Unit \text{ or } \Int \text{ and } u\in \sem{\iota}\\
&  \AVal{l}{\refer\theta}    \defeq   \{(l,\emptyseq,\{(l,\refer \theta)\} \sep l \in \Loc\} \\
  %
&  \AVal{u}{\alpha}  \defeq  \{(p,[p \mapsto u],\varnothing) \sep p \in \ApolT{\alpha}\} \!\cup\! 
  \{(u,\emptyseq,\varnothing) \sep u \in \ApolT{\alpha}\!\} \\
&  \AVal{u}{\theta}  \defeq  \{(f,[f \mapsto u],\varnothing) \sep f \in \AfunT{\theta}\} \quad \text{for } \theta \text{ functional}\\
%
  &\AVal{\pair{u_1}{u_2}}{\theta_1 \times \theta_2}   \\
&\quad \defeq    \{(\pair{v_1}{v_2},\gamma_1 \cdot \gamma_2,\phi_1\cup \phi_2) 
    \sep 
 (v_i,\gamma_i,\phi_i) \in \AVal{u_i}{\theta_i}\}  \\
  &\AVal{\pair{\theta'}{u}}{\exists\alpha.\theta} \\
&\quad \defeq  
    \{(\pair{\alpha'}{v},\gamma \cdot [\alpha' \mapsto \theta'],\phi) 
    \sep 
 (v,\gamma,\phi) \in \AVal{u}{\theta\subst{\alpha}{\alpha'}}\} 
\end{align*}
For uniformity, it makes sense to 
view types as values of special ``universe'' type $\Univ$ and set
 $\AVal{\theta}{\Univ} = \{(\alpha,[\alpha \mapsto \theta],\varnothing) \sep \alpha \in \TVar\}$.
By abuse of notation, we shall use $u$ and variants to range over values, abstract values and types when utilising the notation presented next.
Given a functional type $\theta$ and some $u$, we let the \emph{argument} and \emph{return type} of $\theta$ be:
\[
\begin{array}{rlrl}
\argT{\theta' \rightarrow \theta} &= \theta' &
\argT{\forall\alpha.\theta} &= \Univ 
\\
\retT{\theta' \rightarrow \theta}{ u} &= \theta&
 \retT{\forall\alpha.\theta}{ u} &= \theta\subst{\alpha}{u} 
\end{array}
\]
with the last expression above being well-defined only if $u$ is a type.

Finally, in a similar fashion that $\sf AVal$ allows us to move from concrete values to abstract ones, the operator $\sf AStore$ takes us from stores to abstractions thereof. That is, for each store $S$ and typing function $\phi$, the set $\AStore{S,\phi}$ consists of triples of the form $(S',\gamma',\phi')$ where: 
\begin{inparaitem}[$\bullet$]
\item $S'$ is an abstraction of $S$ according to the type information in $\phi$; 
\item $\gamma'$ is the mapping of the fresh abstract names of $S'$ to their concrete values;
\item  and $\phi'$ is the type information for any locations in the codomain of $S'$.
\end{inparaitem}
The formal definition in the case where $\phi$ is single-valued is given as follows. We postpone the definition for general $\phi$ to Section~\ref{sec:trace-sem}.
\[
  \AStore{S}{\phi}  \defeq \!\!\!\!\! \bigodot_{l \in \dom{S}}\!\!\!\! \{([l \mapsto v],\gamma',\phi') \sep (v,\gamma',\phi') \in \AVal{S(l)}{\phi(l)}\} \vspace{-1mm}
\]
\cutout{
  \AStore{S}{\phi}  \defeq &\!\!\! \bigodot_{l \in \dom{S}}\!\!\! \{
([l \mapsto v],\gamma_v,\phi_v) \sep 
   \phi_v = \bigcup\nolimits_{\theta \in X} \phi_{\theta} \\[-2mm]
&\qquad\qquad \land
   \forall \theta \in X.\,(v,\gamma_v,\phi_{\theta}) \in \AVal{S(l)}{\theta}
   \}
with $X_l = \min{\RCast{\phi}(\theta)}$ for any $\theta \in \phi(l)$.}%
Here $\odot$ is the pointwise concatenation of sets of triples $(S,\gamma,\phi)$, defined as
$X_1 \odot X_2 \defeq \{(S_1 \cdot S_2,\gamma_1 \cdot \gamma_2,\phi_2 \cup \phi_2) \sep (S_i,\gamma_i,\phi_i) \in X_i, i \in \{1,2\}\}$,
and $\bigodot_{i \in \emptyset} X_i \defeq \{(\empHeap,\emptyseq,\varnothing)\}$.
A similar notion is used for producing abstract stores where only typing information (and no concrete store) is defined as follows. 
\begin{align*}
  \SemStore{\phi}  \defeq &\!\!\! \bigodot_{l \in \dom{\phi}}\! \{([l \mapsto v],\phi') \sep (v,\phi') \in\sem{\phi(l)}\} 
\end{align*}
This is used for determining what stores can O play.

\cutout{
  \SemStore{\phi}  \defeq &\!\!\! \bigodot_{l \in \dom{\phi}}\!\!\! \{
([l \mapsto v],\phi_v) \sep 
   \phi_v = \!\!\!\bigcup\limits_{\theta \in X_l}\!\!\! \phi_{\theta} 
 \land
\forall \theta \in X_l.\,   (v,\phi_{\theta}) \in \sem{\theta}
   \}

with $X_l = \min{\RCast{\phi}(\theta)}$ for any $\theta \in \phi(l)$,}

We now give the definition of our trace semantics. 
Note that, for syntactic objects $Z$ and (e.g.\ type) environments $\delta$,
we write $Z\{\delta\}$ for the result of recursively applying $\delta$ in $Z$ as a substitution.

\begin{definition}[Trace Semantics]
We call \boldemph{Interaction Reduction} the system generated by the rules in Figure~\ref{fig:ir}. Given a configuration $C$, we let $\Tr(C)$ be the set of all traces produced from $C$. 
Terms are translated by setting
\[
\sem{\typingTerm{\Sigma}{\Delta}{\Gamma}{M}{\theta}} =  \comp{\Tr\inbrax{\typingTerm{\Sigma}{\Delta}{\Gamma}{M}{\theta}}}
\]
for each typed term $\typingTerm{\Sigma}{\Delta}{\Gamma}{M}{\theta}$, where ${\bf comp}$ selects the \boldemph{complete traces},
that is those traces where the number of answers is {greater or equal} to the number of questions.
\end{definition}

\cutout{
A trace $t$ is said to be \emph{generated} by a configuration $C$
if it can be written as a sequence $(\action{m_1}{S_1}{\rho_1}) \cdots (\action{m_n}{S_n}{\rho_n})$ of full moves
\st $C \xRedint{(\faction{m_1}{S_1}{\rho_1})} C_1 \xRedint{(\faction{m_2}{S_2}{\rho_2})} \ldots \xRedint{(\faction{m_n}{S_n}{\rho_n})}
C_n$, and we write $C \xRedint{t} C_n$. 
The set of traces generated by $C$ is written $\trace{(C)}$.
}

\begin{figure*}
\begin{tabular}{@{\ }l@{\;\;\ }l}
\textsc{(Int)} & 
$\configuration{(M,\theta)::\EE}{\gamma}{\phi}{S}{\lambda}
    \xredint{\phantom{\action{\ansP{v}}{S'}}} 
    \configuration{(M',\theta)::\EE}{\gamma}{\phi}{S'}{\lambda}$, \
given $(M,S) \red (M',S')$.
\\

\textsc{(PA)} & 
 $\configuration{(u,\theta)::\EE}{\gamma}{\phi}{S}{\lambda}
    \xredint{\faction{\ansP{v}}{S'\!\!}{\,\rho}}
  \configuration{\EE}{\gamma \cdot \gamma'}{\phi \cup \phi'}{S}{\lambda\cdot\lambda'}$, \ given $(v,\gamma_v,\phi_v) \in \AValD{u}{\theta}{\kappa}$.  \\

\textsc{(PQ)} & 
$\configuration{(E[f\, u],\theta)::\EE}{\gamma}{\phi}{S}{\lambda}
    \xredint{\faction{\questP{f}{v}}{S'\!\!}{\,\rho}}
    \configuration{(E,\theta'\rightsquigarrow\theta)::\EE}{\gamma \cdot \gamma'}{\phi \cup \phi'}{S}{\lambda\cdot\lambda'}$, \\
& given 
$f \in \AfunT{\theta_f}$ with $\lambda(f) = O$ and
$(v,\gamma_v,\phi_v) \in \AValD{u}{\argT{\theta_f}}{\kappa}$, $\theta'=\retT{\theta_f}{v}$.\\

\textsc{(OA)} &
$\configuration{(E,\theta'\! \rightsquigarrow\theta):: \EE}{\gamma}{\phi}{S}{\lambda}
    \xredint{\faction{\ansO{v}}{S'\!\!}{\,\rho}}
   \configuration{(\widetilde{E}[\widetilde{v}],\theta)::\widetilde{\EE}}{\widetilde{\gamma}}{\phi \cup \phi'}
    {\widetilde{S}[\widetilde{S'}]}{\lambda\cdot\lambda'}$, \
given  $(v,\phi_v) \in \sem{\theta'}_{\kappa}$. \\ 

\textsc{(OQ)} & $\configuration{\EE}{\gamma}{\phi}{S}{\lambda}
     \xredint{\faction{\questO{f}{v}}{S'\!\!}{\,\rho}}
     \configuration{(\widetilde{u}\, \widetilde{v},\theta)::\widetilde{\EE}}{\widetilde{\gamma}}{\phi \cup \phi'}
     {\widetilde{S}[\widetilde{S'}]}{\lambda\cdot\lambda'}$ \\
& given $f \in \AfunT{\theta'}$ with $\lambda(f) = P$ and 
$(v,\phi_v) \in \sem{\argT{\theta'}}_{\kappa}$, $\theta = \retT{\theta'}{v}$ and
  $\gamma(f) = u$. 
\\

\textsc{(Ini)} & $\inbrax{\typingTerm{\Sigma}{\Delta}{\Gamma}{M}{\theta}}
     \xredint{\faction{\questO{?\,}{v}}{S'\!\!}{\,\rho}}
     \configuration{(M\overrightarrow{\subst{x}{\widetilde{u}}},\theta)}{\varepsilon}{\phi'}
     {\widetilde{S'}}{\lambda'}$, \ given $\dom{\Gamma}=\{x_1,\cdots,x_n\}$, 
$(v,\phi_v)\in\sem{\Delta,\Sigma,\Gamma}$ and $v = (\vec\alpha,\vec l,\vec u)$.
\\[2mm]
$\widetilde{Z}$
&
Above, $\widetilde{Z} = Z\{\rho\}\{\gamma\}$, if $Z$ a term, context or stack, and 
$\widetilde{Z} = \{(z,\widetilde{Z(z)})\mid z\in\dom{Z}\}$ if $Z$ a map into terms.
\\
%
\hline\\[-3mm]
P1 & $\lambda' = \{(a,P) \sep a \in \support{}{v,S',\rho} \land a \notin \support{}{\lambda}\}$,
$\phi'=\phi_v \cup \phi_S \cup \phi_\rho$ and
$\gamma'=\gamma_v \cdot \gamma_S \cdot \gamma_\rho$
\\ P2 
& for all $f \in \support{\rmF}{S',v,\rho}, f \notin \support{}{\lambda}$
 and, for all  $\alpha \in \support{\rmT}{S',v,\rho}, \alpha \notin \support{}{\lambda}$
\\
P3 & $\support{\rmL}{\lambda'} = S^*(\support{\rmL}{v,\rho,\lambda})$ 
   and $(S',\gamma_S,\phi_S) \in \AStore{S_{|\support{\rmL}{\lambda'}}}{\phi}$
\\
P4 &  for all $p \in \support{\rmP}{S',v,\rho}$ with $p \in \ApolT{\alpha}$,
 $\lambda(\alpha) = P$ iff $p \notin \support{}{\lambda}$
\\[1mm]
P* & $(\rho,\gamma_\rho,\phi_\rho) \in \APEnv{\restrictTr{(\gamma\cdot\gamma_v\cdot\gamma_s)}{\Apol}}{\kappa,\kappa'}$ where
  $\kappa = \RCast{\phi}$ and $\kappa' = \RCast{\phi\cup\phi'}$, with $\phi\cup\phi'$ valid.
\\[1mm]
\hline\\[-3mm]
O1 &  $\lambda' = \{(a,O) \sep a \in \support{}{v,S',\rho} \land a \notin \support{}{\lambda}\}$
$\phi'=\phi_v \cup \phi_S \cup \phi_\rho$ and each $a\in\support{}{\lambda'}\backslash\Aloc$ is single-played in $(\faction{v}{S'}{\rho})$
\\
O2 & for all $f \in \support{\rmF}{S',v,\rho}, f \notin \support{}{\lambda}$
  and for all  $\alpha \in \support{\rmT}{S',v,\rho}, \alpha \notin \support{}{\lambda}$
\\
O3 & $S'$ closed, $\support{\rmL}{v,\rho} \subseteq \dom{S'}= S'^*(\support{\rmL}{v,\rho,\lambda})$,
  $\dom{S'} \cap \dom{S} = \support{\rmL}{\lambda}$ and $(S',\phi_S) \in \SemStore{\phi'}$
\\
O4 &  for all $p \in \support{\rmP}{S',v,\rho}$ with $p \in \ApolT{\alpha}$,
  $\lambda(\alpha) = O$ iff $p \notin \support{\rmP}{\lambda}$
\\[1mm]
O* & $(\rho,\phi_\rho) \in \SemPEnv{\xi}{\kappa,\kappa'}$ where
  $\xi = \{p \in \ApolT{\alpha} \sep \lambda{\cdot}\lambda'(p) = O\}$,
  $\kappa = \RCast{\phi}$ and $\kappa' = \RCast{\phi\cup\phi'}$ with $\phi\cup\phi'$ valid.
\end{tabular}
\caption{Interaction Reduction. Rules {\sc(PQ),(PA)} satisfy conditions P1-P4 and P*, while 
{\sc(OQ),(OA)} satisfy O1-O4 and O*. Rule {(\sc Ini)} satisfies O1, O3 and O* (taking $S=\varepsilon$, $\phi=\varnothing$ and
$\lambda=\varepsilon$).}\vspace{-2.95mm}
\label{fig:ir}
\end{figure*} 

In the rest of this section we explain the reduction rules and their conditions, apart from conditions P* and O* which concern {type disclosure} and are relegated to the next section. For the same reason, we also assume that typing functions $\phi$ are always single-valued and disregard any indexing with $\kappa$ used in the rules ($\kappa$'s are cast functions).

\mypar{Internal \sc(Int)}
This rule dictates that the interaction reduction includes the operational semantics of \SystemReF\ as long as internal computation steps are concerned, i.e.\ ones that do not involve external functions.

\mypar{P-Question \sc(PQ)} 
This rule describes the move occurring when an external function call is reached.
Thus, in order for P to provide the value (say) $u$ and store $S$, he first needs to abstract it to $v$ by hiding away all private code under fresh names. 
These will be the names put in $\lambda'$, along with any new location names revealed in the store $S'$ to be played. 
Since this is a P-move then, all names in $\lambda'$ are owned by P (P1).
In turn, $S'$ is the restriction of $S$ to public locations, again elevated to its abstraction. These abstractions result in new $\gamma'=\gamma\cdot\gamma_v\cdot\gamma_S$ and $\phi'=\phi\cup\phi_v\cup\phi_S$ (P1).
Note that the $\lambda$ component of a configuration enlists the \emph{public} names of a trace, i.e.\ those explicitly played in moves. Hence,
P3 stipulates that the locations included in the store $S'$ are precisely the ones reachable in $S$ from the names in $\lambda$ and any names in $v$ (put otherwise, name privacy is imposed). 
Finally, P2 dictates that any functional or type variable names played in the move must be fresh (as they represent abstractions of concrete values). 
Similarly, every polymorphic name played of type $\alpha$, with $\alpha$ of own polarity, must be fresh. If, on the other hand, $\alpha$ belongs to O, then P can only play old polymorphic names of that type (P4).

\mypar{P-Answer \sc(PA)}
In this case, a final value is reached and returns, with similar conditions applied.

\mypar{O-Question \sc(OQ)}
When it is the context's turn to play, one option is for O to call one of the functions provided by P. 
The rule looks very similar to the P-Question, yet it differs in one important point:
while O plays $v$ and $S'$, what is fed instead to the configuration is $v$ where all its P polymorphic and functional names have been replaced by their actual values 
(i.e.\ $v\{\gamma\}$)\footnote{we also substitute via $\rho$, but this we discuss in the next section.}  and the same goes for the abstract store $S'$. 
This is enforced by the use of $\widetilde{v}$ instead of $v$ and is due to the fact that P knows the actual values of these names, and therefore they should not remain abstract to him.
Another difference is the freedom to build $S'$, which nonetheless stipulates that O cannot guess any locations from $S$ unless the latter were already public.
Finally, observe in O1 the single-played restriction on fresh polymorphic, type or function names: 
as each such introduced name has the purpose of abstracting some concrete value or type played, every such name should be distinct (and fresh).\footnote{Formally, a move $(\faction{m}{S}{\rho})$ is said to \emph{single-play} a name $a\in\Anl\backslash\Aloc$ if $m$ is equal to 
$\questP{f}{v}$, $\questO{f}{v}$, $\ansP{v}$ or $\ansO{v}$ (for some $f$)
with $a \in \support{}{v,\ntt{S},\codom{\rho}}$ and there is only one occurrence of $a$ in $(\faction{v}{S}{\rho})$.
}
This condition is implicitly imposed in P1 as well, via the domain disjointness requirements in the definition of $\gamma'$.

\mypar{O-Answer \sc(OA)} 
On the other hand, a context can also return with a value, with similar conditions applied.

\mypar{Initial move \sc(Ini)} 
Initial moves are special O-Questions. In order for the 
interaction to commence, O needs to provide the context, that is, the values corresponding to the typing environment $\Delta,\Sigma,\Gamma$.

\smallskip
Let us look at a couple of examples.

\begin{example}
Consider the term $v\equiv \Lambda \alpha.\lambda x:\alpha \times \alpha.~\proj{1}(x)$ of type $\theta=\forall\alpha.\,\alpha{\times}\alpha\to\alpha$.
A characteristic trace of ${v}$ is 
$\questO{?}{}\cdot\ansP{g} \cdot \questO{g\,}{\alpha'} \cdot \ansP{f}\cdot
\questO{f}{p_1,p_2} \cdot\ansP{p_1}$,
produced as follows (we omit empty stores and $\rho$'s).%
{\small
\[\begin{array}{l@{\qquad}r}\\[-4.5mm]
\!\!\!\inbrax{\typingTerm{\cdot}{\cdot}{\cdot}{v}{\theta}} \xrightarrow{\questO{?\,}{}}
\configuration{(v,\theta)}{\ee}{\emptyset}{\ee}{\ee}
&(\theta=\forall\alpha.\,\alpha{\times}\alpha\to\alpha)
\\ 
\xrightarrow{\ansP{g}} \configuration{\emptyStack}{\gamma_1}{\emptyset}{\ee}{\lambda_1}
&\hspace{-2mm}(\gamma_1=[g\mapsto v],\lambda_1=(g,P))
\\
\xrightarrow{\questO{g\,}{\alpha'}}
\configuration{(v\,\alpha',\theta')}{\gamma_1}{\emptyset}{\ee}{\lambda_2} &\hspace{-55mm}
(\theta'=\alpha'\!{\times}\alpha'\!\!\to\!\alpha',\lambda_2\!=\lambda_1\!\cdot\!(\alpha'\!,O))
\\[2mm]
\to
\configuration{(v',\theta')}{\gamma_1}{\emptyset}{\ee}{\lambda_2} &\hspace{-55mm}
(v'\!\equiv \lambda x\!:\!\alpha' {\times} \alpha'\!.\,\proj{1}(x))
\\
\xrightarrow{\ansP{f}}
\configuration{\emptyStack}{\gamma_2}{\emptyset}{\ee}{\lambda_3} &
\hspace{-54mm}(\gamma_2=\gamma_1\!\cdot\![f\mapsto v'],\lambda_3=\lambda_2\!\cdot\!(f,P))
\\
\xrightarrow{\questO{f\,}{p_1,p_2}}
\configuration{(v'\!\pair{p_1}{p_2},\alpha')}{\gamma_2}{\emptyset}{\ee}{\lambda_4} &
(\lambda_4\!=\lambda_3\!\cdot\!(p_1,O)(p_2,O))
\\
\multicolumn{2}{l}{%
\to^*\configuration{(p_1,\alpha')}{\gamma_2}{\emptyset}{\ee}{\lambda_4} 
\xrightarrow{\ansP{p_1}}
\configuration{\emptyStack}{\gamma_2}{\emptyset}{\ee}{\lambda_4} }
\end{array}\]}%
Informally, after the initial move is played, the term is already evaluated to a function of type $\forall \alpha.\,\alpha{\times}\alpha\to\alpha$ and so P plays the move $\ansP{g}$ with $g\in\Afun_{\forall \alpha.\,\alpha{\times}\alpha\to\alpha}$.
At that point, the environment (O) may wish to interrogate $g$, supplying a type variable $\alpha'$
which is an abstraction of any type instantiation the environment may have chosen. Such a question would be of the form $\questO{g\,}{\alpha'}$. 
To the latter, P replies with a functional name $f$, via the move $\ansP{f}$,
of type $(\alpha'\times\alpha')\to\alpha'$.
Next, O decides to also interrogate $f$, say on input $\pair{4}{2}$. 
This translates to the move $\questO{f}{p_1,p_2}$, where now $p_1\mapsto 4$ and $p_2\mapsto 2$ for O. 
The trace concludes with P replying $\ansP{p_1}$, which is the return value of the first projection on $\pair{p_1}{p_2}$.
\end{example}

\begin{example}
Let us take
$v\equiv \lambda x\!:\!(\forall \alpha.\, \alpha {\rightarrow} \alpha).\,x \, \Int \, 3 + x \, \Int \,5$ 
of
type 
$\theta=(\forall \alpha.\alpha {\rightarrow} \alpha) \rightarrow \Int$.
A characteristic trace of $v$ is 
$\questO{?}{}\ansP{f} \cdot \questO{f}{g} \cdot \questP{g\,}{\alpha_1} \cdot \ansO{g_1} \cdot \questP{g_1}{p_1} \cdot \ansO{p_1} \cdot
\questP{g\,}{\alpha_2} \cdot \ansO{g_2} \cdot \questP{g_2}{p_2} \cdot \ansO{p_2} \cdot 
\ansP{8}$
and can be
produced by the following interaction.
{\small
\[\begin{array}{l@{}r}\\[-4mm]
\!\!\!\inbrax{\typingTerm{\cdot}{\cdot}{\cdot}{v}{\theta}} \xrightarrow{\questO{?\,}{}}
\configuration{(v,\theta)}{\ee}{\emptyset}{\ee}{\ee}
\\ 
\xrightarrow{\ansP{f}} \configuration{\emptyStack}{\gamma_1}{\emptyset}{\ee}{\lambda_1}
&\hspace{-40mm}(\gamma_1\!=[f\mapsto v],\lambda_1\!=(f,P))
\\
\xrightarrow{\questO{f\,}{g}}
\configuration{(vg,\Int)}{\gamma_1}{\emptyset}{\ee}{\lambda_2} &
(\lambda_2\!=\lambda_1\!\cdot\!(g,O))
\\[2mm]
\to
\configuration{(g\,\Int\,3+g\,\Int\,5,\Int)}{\gamma_1}{\emptyset}{\ee}{\lambda_2} &
\!\!\!\!(\gamma_2=\gamma_1\!\cdot\![\alpha_1\mapsto \Int])
\\[1mm]
\xrightarrow{\questP{g\,}{\alpha_1}}
\configuration{(\bullet\,3+g\,\Int\,5,\alpha_1{\to}\alpha_1\redsat\Int)}{\gamma_2}{\emptyset}{\ee}{\lambda_3} & 
(\lambda_3\!=\lambda_2\!\cdot\!(\alpha_1,P))
\\
\xrightarrow{\ansO{g_1}}
\configuration{(g_1\,3+g\,\Int\,5,\Int)}{\gamma_2}{\emptyset}{\ee}{\lambda_4} &
\hspace{-54mm}(\lambda_4\!=\lambda_3\!\cdot\!(g_1,O))
\\[1mm]
\xrightarrow{\questP{g_1}{p_1}}
\configuration{(\bullet+g\,\Int\,5,\alpha_1\redsat\Int)}{\gamma_3}{\emptyset}{\ee}{\lambda_5} & (\lambda_5\!=\lambda_4\!\cdot\!(p_1,P))
\\
\xrightarrow{\ansO{p_1}}
\configuration{(3+g\,\Int\,5,\Int)}{\gamma_3}{\emptyset}{\ee}{\lambda_5} &
\hspace{-54mm}(\gamma_3\!=\gamma_2\!\cdot\![p_1\mapsto 3])
\\[1mm]
\xrightarrow{\questP{g\,}{\alpha_2}}
\configuration{(3+\bullet\,5,\alpha_2{\to}\alpha_2\redsat\Int)}{\gamma_4}{\emptyset}{\ee}{\lambda_6} & (\lambda_6\!=\lambda_5\!\cdot\!(\alpha_2,P))
\\
\xrightarrow{\ansO{g_2}}
\configuration{(3+g_2\,5,\Int)}{\gamma_4}{\emptyset}{\ee}{\lambda_7} &
\hspace{-54mm}(\gamma_4=\gamma_3\!\cdot\![\alpha_2\mapsto \Int]) 
\\[1mm]
\xrightarrow{\questP{g_2}{p_2}}
\configuration{(3+\bullet,\alpha_2\redsat\Int)}{\gamma_5}{\emptyset}{\ee}{\lambda_8} & 
(\lambda_7\!=\lambda_6\!\cdot\!(g_2,O))
\\
\xrightarrow{\ansO{p_2}}
\configuration{(3+5,\Int)}{\gamma_5}{\emptyset}{\ee}{\lambda_8} &
\hspace{-54mm}(\gamma_5\!=\gamma_4\!\cdot\![p_2\mapsto 5], \lambda_8\!=\lambda_7\!\cdot\!(p_2,P))
\\
\to\configuration{(8,\Int)}{\gamma_5}{\emptyset}{\ee}{\lambda_8} 
\xrightarrow{\ansP{8}}
\configuration{\emptyStack}{\gamma_5}{\emptyset}{\ee}{\lambda_8} 
&\hspace{-40mm} 
\end{array}\]}%
Notice that $p_1,p_2$ are of different type, respectively $\alpha_1$ and $\alpha_2$.
As an exercise, we invite the reader to verify that the term
$v'\equiv \lambda x\!:\!(\forall \alpha.\, \alpha {\rightarrow} \alpha).\,
\letin{\,h=x \, \Int\,}{\,h \, 3 + h \,5}$ 
of the same
type 
$\theta$ produces the trace
$\questO{?}{}\ansP{f} \cdot \questO{f}{g} \cdot \questP{g\,}{\alpha'} \cdot \ansO{g'} \cdot \questP{g'\,}{p_1} \cdot \ansO{p_1} \cdot
\questP{g'\,}{p_2} \cdot \ansO{p_1} \cdot 
\ansP{6}$.
The latter behaviour can be triggered by a context which uses local state to record  polymorphic values of older calls:
\begin{align*}
&C\equiv\, \bullet\, \big(\Lambda \alpha.\,\letin{y = \nuref{(\lambda\_.\Omega_\alpha)}}{\letin {z =  \nuref{0}}{}}\\[-1mm]
&\qquad\quad\;\;\;\lambda x:\alpha.\, \ifte{(!z)}{(!y())}{(z:={!z}+1; y:=(\lambda\_.x); x)}\big)
\end{align*}
\end{example}

\cutout{
\mypar{Types}\nt{is this relevant?}
While we do not look into soundness of the above rules in this section, it is instructive to give a type system for our extended syntax as a sanity check. Typing judgements are now indexed by an environment $\delta$, and apart from:
\begin{gather*}
  \inferrule*{f \in \AfunT{\theta}}{\Delta;\Sigma;\Gamma\vdash_\delta f:\theta\{\delta\}} \qquad
  \inferrule*{p \in \ApolT{\alpha}}{\Delta;\Sigma;\Gamma\vdash_\delta p:\alpha}
\end{gather*}  
the rest of the typing rules remain the same.
}


%% file: fulltrace.tex
\section{Type Disclosure, Casts and *-Conditions}
\label{sec:trace-sem}


As already discussed in the Introduction, the existence of references can be used to the advantage of a program in order to break parametricity. 
This is done by discovering variables of different reference types which, upon execution, end up with the same concrete location. 
Once such an \gjt{\emph{aliased pair}} has been identified, of type say $\refer\theta_1,\refer\theta_2$, then a casting function between $\theta_1$ and $\theta_2$ is readily available. 
For instance, if the two variables are $x_i:\refer\theta_i$, here is a casting function from $\theta_1$ to $\theta_2$:
\[
{\sf cast}_1\ \equiv\ \lambda z_1\,{:}\,\theta_1.\ x_1:=z_1;\ {!x_2}: \theta_1\to\theta_2
\]
Clearly, if the same location $l$ flows in $x_1$ and $x_2$ then we obtain ${\sf cast}_1\{l/x_1,l/x_2\}$  which casts indeed as designed.
The reader may wonder under what circumstances can the same location be passed to variables of different types. This can be achieved, for instance, by a context:
\[
C \equiv\, \letin{\,x=\nuref{0}\,}{(\Lambda\alpha.\,\lambda y_1\,{:}\,\refer\alpha.\,\lambda y_2\,{:}\,\refer\Int.\,\bullet)\,\Int\,x\,x}
\]
whereby $\theta_1=\alpha$ and $\theta_2=\Int$.

These considerations bring about \emph{type disclosure}, which we examine next in detail. We conclude the prelude to this section with some interesting equivalence examples/non-examples, left as a quiz for the reader.

\begin{example}\label{ex:quiz}
Suppose $f\!:(\refer\Int\times\refer\Int) \rightarrow \Unit$, $g:\forall\alpha.$ $\refer\alpha\rightarrow\! \refer\alpha$ and
$h\!:\!\forall\alpha,\alpha'\!.(\refer(\alpha' {\rightarrow}\, \alpha) \times \refer(\alpha' {\rightarrow}\, \Int) \times \alpha) \rightarrow \alpha$.
\begin{enumerate}
\item
$\letin{x,y = \refer 0}{f (x,y);\letin{u = g\, \Int \, x \ }{\,\ifte{(u = y)}{1}{2}}}
$
\item[$\;{\cong}\text{?}$]
$\letin{x,y = \refer 0} {f (x,y);\letin{u = g \, \Int \, x \ }{\,\ifte{(u = y)}{3}{2}}}
$\\[-1.5mm]
\item
$\letin{ x = \refer (\lambda y.1)}{\letin{ u = h \, \Int \, \Int \, (x,x,0)}{\, \ifte{u}{1}{2}}}$
\item[$\;{\cong}\text{?}$]
$\letin{ x = \refer (\lambda y.1) }{\letin{ u = h \, \Int \, \Int \, (x,x,0)}{\,\ifte{u}{3}{2}}}$
\end{enumerate}
\end{example}

\subsection{Type disclosure and casts}

Type disclosure is the result of the same location appearing in several positions in the code, each expecting some different type.
In such cases, we need to associate in our semantics a set of types to each location, employing the non-unicity of typing functions $\phi$.
In order to restrict the behaviour of O in the interaction to plausible computations, we shall impose some validity conditions to $\phi$: 
after all, not all types can be instantiations of the same type variable (for instance, $\phi(l)=\{\refer\Int,\refer\Unit\}$ is not allowed).

Validity is also dependent on precedence of type variables in the trace: a recent type variable cannot be instantiating one which has appeared before it in the trace.
We define a partial relation $\TRel{\Phi}$ on types, indexed by an \emph{ordered} set $\Phi$ of type variables, as:
 \begin{gather*}
  \frac{}{\theta\TRel{\Phi} \theta}\quad
  \frac{\theta_1 \TRel{\Phi} \theta_2\TRel{\Phi} \theta_3}{\theta_1\TRel{\Phi} \theta_3}\quad
  \frac{\support{\rmT}{\theta}<_{\Phi}\alpha}{\theta \TRel{\Phi} \alpha} 
  \frac{\theta \TRel{\Phi} \theta'}{\refer\theta \TRel{\Phi} \refer\theta'}
\\
  \frac{\theta_1 \TRel{\Phi} \theta'_1 \;\; \theta_2 \TRel{\Phi} \theta'_2}{\theta_1 \times \theta_2 \TRel{\Phi} \theta'_1 \times \theta'_2}\quad
  \frac{\theta \TRel{\Phi} \theta' \;\; \alpha \notin \Phi}{Q\alpha.\theta \TRel{\Phi} Q\alpha.\theta'} \quad    
  \frac{\theta_1 \TRel{\Phi} \theta'_1 \;\; \theta_2 \TRel{\Phi} \theta'_2}{\theta_1 \rightarrow \theta_2 \TRel{\Phi} \theta'_1 \rightarrow \theta'_2}
 \end{gather*}
for $Q=\exists,\forall$ and with $\support{\rmT}{\theta}<_{\Phi}\alpha$ meaning that {all} $\alpha' \in \support{\rmT}{\theta}$ {are before} $\alpha$ in $\Phi$.
Let us fix some $\Phi$ for the next definition.

\begin{definition}
%
A typing function $\phi$ is said to be \boldemph{valid} if for all $l \in \dom{\phi}$ there exists a type $\theta_0$ such that $\theta_0\TRel{\Phi}\theta$ for all $\theta\in\phi(l)$.
\end{definition}

In the sequel we will be using a very specific set $\Phi$, which we shall be leaving implicit.
For any configuration $C$ with components $\lambda$ and $\phi$, we say that $\phi$ is valid if it is so with respect to the ordered set $\Phi_{\lambda}$ of type variables obtained from $\lambda$: 
$\Phi_{\lambda}= \pi_1(\lambda)\upharpoonright\Atvar$.

As type instantiations are noticed during an interaction, the two parties can start forming cast functions to move between types. 
We introduce the notion of \boldemph{cast relations} $\kappa$, which are simply relations over types.
The fact that $(\theta,\theta') \in \kappa$ means that we can cast values of type $\theta$ to $\theta'$.

Casts yield other casts. For example, a cast from $\theta_1\times\theta_2$ to $\theta_1'\times\theta_2'$ yields subcasts 
from $\theta_1$ to $\theta_1'$, and from $\theta_2$ to $\theta_2'$.\footnote{Assuming $\theta_1$ and $\theta_2$ are inhabited types.} We formalise this as follows.
Given a cast relation $\kappa$, we define its closure $\TDclos{\kappa}$ by:
{
\begin{gather*}
  \inferrule*{(\theta,\theta') \in \kappa}{(\theta,\theta') \in \TDclos{\kappa}}
 \quad
  \inferrule*{ }{(\theta,\theta) \in \TDclos{\kappa}}
 \quad
  \inferrule*{(\theta,\theta'') \in \TDclos{\kappa} \quad (\theta'',\theta') \in \TDclos{\kappa}}
             {(\theta,\theta') \in \TDclos{\kappa}}
 \quad
  \inferrule*{(\refer\theta,\refer\theta') \in \TDclos{\kappa}}{(\theta,\theta') \in \TDclos{\kappa}} 
 \\           
  \inferrule*{(\theta_1,\theta'_1) \in \TDclos{\kappa} \quad (\theta_2,\theta'_2) \in \TDclos{\kappa}}
             {(\theta_1 \times \theta_2,\theta'_1 \times \theta'_2) \in \TDclos{\kappa}}
 \;\;           
  \inferrule*{(\theta_1 \times \theta_2,\theta'_1 \times \theta'_2) \in \TDclos{\kappa}}
             {(\theta_1,\theta'_1) \in \TDclos{\kappa}}
 \;\;           
  \inferrule*{(\theta_1 \times \theta_2,\theta'_1 \times \theta'_2) \in \TDclos{\kappa}}
             {(\theta_2,\theta'_2) \in \TDclos{\kappa}}
 \\           
  \inferrule*{(\theta'_1,\theta_1) \in \TDclos{\kappa} \quad (\theta_2,\theta'_2) \in \TDclos{\kappa}}
             {(\theta_1\! \rightarrow \theta_2,\theta'_1\! \rightarrow \theta'_2) \in \TDclos{\kappa}}
 \;\;            
  \inferrule*{(\theta'_1\! \rightarrow \theta_2,\theta_1\! \rightarrow \theta'_2) \in \TDclos{\kappa}}
             {(\theta_1,\theta'_1) \in \TDclos{\kappa}}
 \;\;
  \inferrule*{(\theta_1\! \rightarrow \theta_2,\theta'_1\! \rightarrow \theta'_2) \in \TDclos{\kappa}}
             {(\theta_2,\theta'_2) \in \TDclos{\kappa}}
 \\
  \inferrule*{(\theta,\theta') \in \TDclos{\kappa} \quad \alpha \notin \support{}{\kappa}}
             {(Q\alpha.\theta,Q\alpha.\theta') \in \TDclos{\kappa}}
 \quad
  \inferrule*[Right=($*$)]{(Q\alpha.\theta,Q\alpha.\theta') \in \TDclos{\kappa} \quad \chi(\alpha,\theta,\theta')}
             {(\theta\subst{\alpha}{\theta_0},\theta'\subst{\alpha}{\theta_0}) \in \TDclos{\kappa}}
\end{gather*}}%
for $Q=\exists,\forall$, where $\chi(\alpha,\theta,\theta')$ means that $\alpha$ does not appear in the scope of a $\rm ref$ constructor in $\theta,\theta'$.
{Notice that all the rules are going in both directions, but the one on {\rm ref} types. Indeed, being able to cast from $\theta$ to $\theta'$
does not imply we can cast from $\refer\theta$ to $\refer\theta'$. This observation allows us to see that the terms of Example~\ref{ex:quiz}\,(1) are equivalent despite the type disclosure (cf.\ Section~\ref{sec:examples}).}

We can now define the cast relation $\RCast{\phi}$ related to a typing function $\phi$.
We can show that, for any valid typing function $\phi$, $\RCast{\phi}$ is a valid cast relation. 

\begin{definition}\label{def:Cast}
Given a typing function $\phi$, its associated cast relation $\RCast{\phi}$  
is the closure of $\{(\theta,\theta') \sep \exists l.\, \theta,\theta' \in \phi(l)\}$.
\end{definition}

Given a cast relation $\kappa$ and a type $\theta$, we let
\[
\min(\kappa(\theta)) = \{\theta' \in \kappa(\theta) \sep \forall \theta'' \in \kappa(\theta).\,
\theta'' \TRel{} \theta'\implies \theta'' =\theta' \}
\]
be the {set}
of minimal types of $\kappa(\theta)$. Because the closure rules
above are not reversible on $\rm ref$ types, this set is not in general a singleton (e.g.\ $\min(X)\!=\!X$ for $X\!=\{\refer(\alpha\times\Int),\refer(\Int\times\alpha')\}$).
This means that a type $\theta$ can have several minimal types in its cast class, and each of them needs to be taken in to account when computing abstract values to be played in a move. 
Hence, minimal types are central to the (full) definitions of $\sf AVal$, $\sf AStore$, etc.

\subsection{The starred conditions}

We next look at the use of environments $\rho$ and the conditions O* and P* which govern type disclosure in the interaction reduction.

Each move $(m,S,\rho)$ played in an interaction has the potential to reveal type information. 
Looking at the reduction rules, in particular, we see that such a move can enlarge 
the current typing function $\phi$ to a (valid) superset $\phi\cup\phi'$: this is due to the fact that locations $l$ 
which up until now had types $\phi(l)$ are put in positions which expect types $\theta\notin\phi(l)$ (e.g.\ in return position of some $f\in\Afun_{\theta'\to\theta}$). 
This leads to a corresponding increase in the cast capabilities to $\kappa'=\RCast{\phi\cup\phi'}$. 
Cast capabilities, though, may reveal the values behind polymorphic names: for instance, if we are able to form a cast from $\alpha$ to $\Int$, 
we can go back to an old $p\in\Anl_\alpha$, cast it as an integer and read its value. 
This decoding capability is the reason behind the presence of $\rho$ in the move: $\rho$ contains all those polymorphic names $p$ whose value is being revealed (indirectly, via casts) through the current move, along with the revealed values. 

The way polymorphic values are revealed is governed by conditions P* and O*.
The former stipulates that, given the old cast relation $\kappa$, the new casting $\kappa'$ is the one we obtain via the updated typing function $\phi\cup\phi'$. 
Moreover, as explained above, each concrete value $\gamma(p)$ of a polymorphic name $p$ needs to be partially revealed.
The degree to which the codomain of $\restrictTr\gamma\Apol$ will be revealed is determined by the function $\sf AEnv$. That is, $\APEnv{\restrictTr\gamma\Apol}{\kappa,\kappa'}$
comprises a new abstract environment $(\rho,\gamma_\rho,\phi_\rho)$ for these newly revealed values, that is moreover unique up to permutation of fresh names. 
The first component ($\rho$) is the map from polymorphic names to their revealed values. The other two components record the locations types ($\phi_\rho$) and value abstraction ($\gamma_\rho$) occurring via this disclosure. 
In the case of O*, a similar disclosure occurs, only that this time there is no $\gamma$ to guide the revealed values; rather, O supplies the disclosure in a non-deterministic fashion.

The definition of $\sf AEnv$ and its O-counterpart are given below,
\begin{align*}
  \APEnv{\gamma}{\kappa,\kappa'} \defeq &\hspace{-1.3em} 
\bigodot_{\begin{array}{l}\scriptstyle p \in \dom{\gamma}\\[-.8mm] \scriptstyle \text{s.t. }X_p \neq \varnothing\end{array}} \hspace{-1em}
  \{
   ([p \mapsto v],\gamma_v,\phi_v) \sep  \phi_v = \bigcup\nolimits_{\theta \in X_p} \phi_{\theta} \\[-6mm]
&\qquad\qquad\;\,\land
  \forall\theta\in X_p.\,(v,\gamma_v,\phi_\theta) \in \AVal{\gamma(p)}{\theta} 
\} \\[2mm]
  \SemPEnv{\xi}{\kappa,\kappa'} \defeq &\hspace{-1.6em} \bigodot_{\quad p \in\xi\text{ s.t. }X_p \neq \varnothing} \hspace{-1.5em}
  \{
   ([p \mapsto v],\phi_v) \sep  \phi_v = \!\bigcup\nolimits_{\theta \in X_p}  \! \phi_{\theta}\\[-2mm] 
&\qquad\qquad\qquad\qquad\qquad\;\,\land \forall\theta\in X_p.\,(v,\phi_\theta) \in \sem{\theta} \}
\end{align*}
with $\dom{\gamma},\xi\subseteq\Apol$ and $X_p = \min{\kappa'(\alpha)} \backslash \min{\kappa(\alpha)}$ for $p \in \ApolT{\alpha}$.
Thus, for each $p\in\ApolT{\alpha}$ in the domain of $\gamma$ such that, going from $\kappa$ to $\kappa'$, there is a new type disclosure on the type of $p$ (i.e.\ such that $X_p\not=\emptyset$), 
to compute the disclosure happening on $\gamma(p)$ we look at all the newly disclosed types $\theta\in X_p$ and for each of them select an abstract environment
from $\AVal{\gamma(p)}{\theta}$. If we can pick these environments so that they all agree in their value component $v$, we can reveal that $p$ maps to $v$.
\ntt{Note that $X_p$ determines how much of $\gamma(p)$ is revealed: for instance, $X_p = \{\alpha'\}$
with $\alpha'$ another type variable, then $v$ will simply be another polymorphic name $p'$.}
On the other hand, $\SemPEnv{\xi}{\kappa,\kappa'}$ is more liberal in choosing the common revealed value $p$, as it scans through each $\sem{\theta}$ instead of  
$\AVal{\gamma(p)}{\theta}$.
In a similar vein, we get:
\[
\begin{array}{r@{\;}l}
  \AValD{u}{\theta}{\kappa} \defeq & \{ (v,\gamma,\phi) \sep  \phi = \!\bigcup\nolimits_{\theta' \in X}  \! \phi_{\theta'} \\
& \qquad\quad\qquad\land 
    \forall\theta'\in X.\, 
    (v,\gamma,\phi_{\theta'}) \in \AVal{u}{\theta'} \} 
\\[2mm]
    \AStore{S}{\phi}  \defeq & \!\!\bigodot\limits_{l \in \dom{S}}\!\!\! \{
([l \mapsto v],\gamma_v,\phi_v) \sep 
   \phi_v = \bigcup\nolimits_{\theta \in X_l} \phi_{\theta} \\[-2mm]
&\qquad\qquad \land
   \forall \theta \in X_l.\,(v,\gamma_v,\phi_{\theta}) \in \AVal{S(l)}{\theta}\}
\end{array}
\]\vspace{-1mm}
\[
\begin{array}{r@{\;}l}
  \sem{v}_{\kappa} &\defeq  \{(v,\phi) \sep  \phi = \!\bigcup\nolimits_{\theta' \in X}  \! \phi_{\theta'} \land 
  \forall\theta'\in X.\,(v,\phi_{\theta'}) \in \sem{\theta'} \} 
\\[2mm]
  \SemStore{\phi}  &\defeq\!\!\!\! \bigodot\limits_{l \in \dom{\phi}}\!\!\! \{
([l \mapsto v],\phi_v) \sep 
   \phi_v\! =\!\!\! \bigcup\limits_{\theta \in X_l}\!\!\phi_{\theta} 

 \land \forall \theta \in X_l.\,   (v,\phi_{\theta}) \in \sem{\theta}\}
\end{array}
\]
with $X = \min(\kappa(\theta))$ and 
$X_l = \bigcup_{\theta \in \phi(l)}\min(\RCast{\phi}(\theta))$.

While there is some circularity between the different new components in condition P*, we can always pick them in a nominally deterministic way. 
We conclude this section with a couple of examples demonstrating type disclosure.

\cutout{
\begin{figure}[t]
\[\begin{array}{l}
  \APEnv{\gamma}{\kappa,\kappa'} \defeq \hspace{-1.2em} \displaystyle\bigodot_{\begin{array}{l}p \in \dom{\gamma}\\ X_p \neq \varnothing\end{array}} \hspace{-0.8em}
  \left\{\begin{array}{l}
   ([p \mapsto v],\gamma_v,\phi_v) \sep \forall \theta \in X_p. \exists \phi_\theta. \\ 
   \quad (v,\gamma_v,\phi_\theta) \in \AVal{\gamma(p)}{\theta} \\  
   \quad \land \phi_v = \bigcup\limits_{\theta \in X_p} \phi_{\theta} \end{array}\right\} \\
  \SemPEnv{\xi}{\kappa,\kappa'} \defeq \hspace{-1.2em} \displaystyle\bigodot_{\begin{array}{l}p \in \dom{\xi}\\ X_p \neq \varnothing\end{array}} \hspace{-0.8em}
  \left\{\begin{array}{l}
   ([p \mapsto v],\phi_v) \sep \forall \theta \in X_p. \exists \phi_\theta. \\ \quad (v,\phi_\theta) \in \sem{\theta} \land \phi_v = \bigcup\limits_{\theta \in X_p} \phi_{\theta} \end{array}\right\} \\    
   \text{ with } X_p = \min{\kappa'(\alpha)} \backslash \min{\kappa(\alpha)} \text{ for } p \in \ApolT{\alpha}
\end{array}\]
\caption{Definition of disclosure environments}
\label{f:AEnv}
\end{figure} 
}


\subsection{Examples}\label{sec:examples}

We first look at a term that uses type disclosure to cast between two of its inputs, similarly to the initial examples of the paper.
Let us set $\theta =\refer\alpha\times\refer\Int\times\alpha$ and $v\equiv\Lambda\alpha.\lambda \langle x,y,z\rangle^{\theta}\!.M$ 
with  $M \equiv \Ifte{{x\!=y}}{({y:=42; !x})}{z}$.
A characteristic trace of $v$ is the following (e.g.\ for $S\!=\![l \mapsto 9],\rho = [p \mapsto 7]$),
{\small
\[\begin{array}{lr}\\[-4mm]
\!\!\!\inbrax{\typingTerm{\cdot}{\cdot}{\cdot}{v}{\theta}} \xrightarrow{\questO{?\,}{}}
\configuration{(v,\theta)}{\ee}{\emptyset}{\ee}{\ee}
\\ 
\xrightarrow{\ansP{f}} \configuration{\emptyStack}{\gamma_1}{\emptyset}{\ee}{\lambda_1}
&(\gamma_1\!=[f\mapsto v],\lambda_1\!=(f,P))
\\
\xrightarrow{\questO{f\,}{\alpha}}
\configuration{(v\alpha,\theta \rightarrow \alpha)}{\gamma_1}{\emptyset}{\ee}{\lambda_2} 
& (\lambda_2\!=\lambda_1\!\cdot\!(\alpha,O))
\\
\xrightarrow{\ansP{g}}
\configuration{\diamond}{\gamma_2}{\emptyset}{\ee}{\lambda_2} 
& (\gamma_2\!=\gamma_1\!\cdot\![g \mapsto \lambda \langle x,y,z\rangle^{\theta}.M])
\\
\xrightarrow{\questO{g\,}{ l,l,p},S,\rho}
\configuration{(M',\alpha)}{\gamma_2}{\phi_1}{S}{\lambda_3} 
& (\phi_1\!=(l,\Int),\!(l,\alpha))
\\
\xrightarrow{\ansP{42},S}
\configuration{\emptyStack}{\gamma_2}{\phi_1}{S}{\lambda_3} 
& (\lambda_3\!=\lambda_2\!\cdot\!(l,O)\!\cdot\!(p,O))
\end{array}\]}%
where $M'\equiv M\subst{x,y}{l}\subst{z}{p}\{\rho\}\equiv \Ifte{{l\!=l}}{({l:=42; !l})}{7}$.

Now, going back to Example~\ref{ex:quiz},
let $f\!:(\refer\Int\times\refer\Int)\, {\rightarrow}\, \Unit$, $g:\forall\alpha.$ $\refer\alpha\rightarrow\! \refer\alpha$ and
$M\equiv\letin{x,y = \refer 0}{f \pair{x}{y};\letin{u = g\, \Int \, x \, }{\ifte{(u = y)}{1}{2}}}$ and
$N\equiv \ifte{(u = l')}{1}{2}$. Then, taking $\gamma\!=[\alpha \mapsto \Int]$,  $M$ can produce characteristic traces of two kinds:
{\small
\[\begin{array}{@{}l@{}r}\\[-4mm]
\!\!\!\inbrax{\typingTerm{\cdot}{\cdot}{\Gamma}{M}{\Int}} \xrightarrow{\questO{?\,}{f,g}}
\configuration{(M,\Int)}{\ee}{\emptyset}{\ee}{\lambda_1}
&(\lambda_1\!=(f,O)\cdot(g,O))
\\[1mm]
\to^* \configuration{(f (l,l');\letin{u = g\, \Int \, l \ }{N},\Int)}{\ee}{\emptyset}{S_1}{\lambda_1}
&(S_1\!=[l\mapsto 0,l'\mapsto 0])
\\ 
\xrightarrow{\questP{f}{l,l'},S_1} \configuration{\bullet;\letin{u = g\, \Int \, l \ }{N}}{\ee}{\phi_1}{S_1}{\lambda_2}
&(\lambda_2\!=\lambda_1\!\cdot\!(l,P)\!\cdot\!(l',P))
\end{array}
\]}
{\small
\[\begin{array}{@{}l@{}r}\\[-7mm]
\xrightarrow{\ansO{()},S_2}
\configuration{(();\letin{u = g\, \Int \, l \, }{N},\Int)}{\ee}{\phi_1}{S_2}{\lambda_2} 
& (\phi_1\!=(l,\Int),\!(l',\Int))
\\
\to\xrightarrow{\questP{g\,}{\alpha},S_2}
\configuration{(\letin{u = \bullet \, l \, }{N},\Int)}{\gamma}{\phi_1}{S_2}{\lambda_3} 
& (\lambda_3\!=\lambda_2\!\cdot\!(\alpha,P))
\\
\xrightarrow{\ansO{h},S_3}
\configuration{(\letin{u = h \, l \, }{N},\Int)}{\gamma}{\phi_1}{S_3}{\lambda_4} 
& (\lambda_4\!=\lambda_3\!\cdot\!(h,O))
\\
\xrightarrow{\questP{h\,}{l},S_3}
\configuration{(\letin{u = \bullet\, }{N},\Int)}{\gamma}{\phi_2}{S_3}{\lambda_4} 
& (\phi_2\!=\phi_1,\!(l,\alpha))
\\
\multicolumn{2}{l}{%
\!\!\!\![1]\xrightarrow{\ansO{l},S_4}
\configuration{(\letin{u = l }{N},\Int)}{\gamma}{\phi_2}{S_4}{\lambda_4} 
\to^*\xrightarrow{\ansP{2},S_4}
\configuration{\emptyStack}{\gamma}{\phi_2}{S_4}{\lambda_4} }
\\
\multicolumn{2}{l}{
\!\!\!\![2]\xrightarrow{\ansO{l''},S_4}
\configuration{(\letin{u = l'' }{N},\Int)}{\gamma}{\phi_3}{S_4}{\lambda_5} 
\to^*\xrightarrow{\ansP{2},S_4}
\configuration{\emptyStack}{\gamma}{\phi_3}{S_4}{\lambda_5} }
\end{array}\]}
\smallskip

\noindent
according to choices [1] and [2] for O's last move. In particular, O can either return the $l:\refer\alpha$ he received, or create a new 
$l'':\refer\alpha$ and return it. Due to $\phi_2$, O can cast from $\Int$ to $\alpha$ and put arbitrary values in $l,l''$. However, 
as $\RCast{\phi_2}(\refer\alpha)=\{\refer\alpha\}$,
he has no cast from $\refer\Int$ to $\refer\alpha$ and hence cannot return $l'$.


%% file: sound.tex
%

%
%
%

\section{Soundness}

We show that our model is sound, i.e.\ equality of term denotations implies contextual equivalence. 
In fact, we prove a stronger result (Theorem~\ref{thm:sound}), whereby equality is replaced by a larger equivalence relation which rules out some over-distinguishing O behaviours.

\subsection{Valid configurations}
\label{subsec:validconf}

To reason on the interaction reduction, we prove it preserves some invariants
which we collect in the notion of \emph{valid configuration}.

An obvious invariant we want to preserve is that elements of the evaluation stack are well-typed.
However, due to the fact that locations do not always have a unique type, and the ensuing casting capabilities that arise,
we cannot use the standard typing system defined in Section~\ref{sec:language}.
We thus need to generalise it by allowing location contexts to be multi-valued, i.e.\ use valid typing functions $\phi$ (instead of $\Sigma$), together with the new typing rule:
\[\inferrule*{\typingTermTD{\phi}{\Delta}{\Gamma}{M}{\theta} \quad (\theta,\theta') \in \RCast{\phi}}
             {\typingTermTD{\phi}{\Delta}{\Gamma}{M}{\theta'}}
\]
\ntt{We write $\typingHeapTD{S}{\!\phi}$ if 
$\forall l\in\dom{S}.\exists\theta \in \phi(l).~\typingTermTD{\phi}{\support{\rm T}{\phi}}{\cdot}{S(l)}{\theta}$.}

The extended type system still satisfies a safety property, which is crucial in the soundness proof of our model (cf.\ Appendix~\ref{app-rts}).

\begin{lemma}
\label{thm:red-typedterm}
Given $\typingTermTD{\phi}{\Delta}{\cdot}{M}{\theta}$ and $\typingHeapTD{S}{\phi}$
such that for all $p \in \support{}{M,S} \cap \ApolT{\alpha}, \min{(\RCast{\phi}(\alpha))} = \{\alpha\}$
 either 
$(M,S)$ diverges
or
there exists $(M',S')$ irreducible such that:
 \begin{compactitem}[$\bullet$]
  \item $(M,S) \red^*  (M',S')$,
  \item $M'$ is either equal to a value $v$ or to a callback $E[f \ v]$,
  \item there exists $\phi'$ disjoint from $\phi$ such that 
 $\typingTermTD{\phi \cup \phi'}{\Delta}{\cdot}{M'}{\theta}$ and $\typingHeapTD{S'}{\ntt{\phi}\cup \phi'}$.
 \end{compactitem}
\end{lemma}

Using this extended system, we can type evaluation stacks of configurations.
 A passive evaluation stack $(E_n,\theta_n \rightsquigarrow \theta'_n)::\ldots::(E_1,\theta_1 \rightsquigarrow \theta'_1)$
 is said to be \emph{well-typed} w.r.t.\ a typing function $\phi$ and 
 a type environment $\delta : \Atvar \parrow \Types$ 
 if,  for all $1\leq i\leq n$,\linebreak
 $\typingCtxTD{\phi}{\Delta}{E_i}{\theta_i\{\delta\}}{\theta'_i\{\delta\}}$.
 An active evaluation stack $(M,\theta,\Phi)::\EE$ is well-typed 
 for $\phi,\delta$ if $\typingTermTD{\phi}{\Delta}{\cdot}{M}{\theta\{\delta\}}$
 and $\EE$ is well-typed for $\phi,\delta$.
We can now specify which configurations are {valid}.

\begin{definition}
We call $\configuration{\EE}{\gamma}{\phi}{S}{\lambda}$ a \boldemph{valid configuration} if:
\begin{compactitem}[$\bullet$]
 \item $\dom{\gamma} = \{a \in \Apol \cup \Afun \cup \Atvar \sep \lambda(a) = P\}$,
 \item $\dom{\phi} = \dom{\lambda} \cap \Loc \subseteq \dom{S}$,
 \item for all $a \in \support{}{\EE,\codom{S},\codom{\gamma}} \backslash \Loc$, $\lambda(a) = O$,
 \item there exists $\phi'$ disjoint of $\phi$ 
 s.t.  $\typingHeapTD{S}{\phi \cup \phi'}$,
 \item $\EE$ is well-typed for $\phi \cup \phi',\gamma_{|\Atvar}$
 \item for all $p \in \support{}{\EE,\ntt{S},\codom{\gamma}} \cap \ApolT{\alpha}$, $\min{(\RCast{\phi}(\alpha))} = \{\alpha\}$.
\end{compactitem}
\end{definition}

\ntt{We write $C \xRedint{\faction{m}{S}{\rho}} C'$ 
when $C \xredint{}^{\!*}\! C'' \xredint{\faction{m}{S}{\rho}} C'$ for some configuration $C''$.}
Validity of configurations is preserved as follows.

\begin{lemma}
\label{thm:validconf}
If $C \xRedint{\faction{m}{S}{\rho}} C'$ and $C$ is valid then so is $C'$.
\end{lemma}


\subsection{Composite reduction}

The main ingredient in the soundness argument is a refinement of the LTS introduced previously which will eventually allow us to compose term denotations, 
in a way akin to composition in game semantics: each term in the composition becomes the Opponent for the other term.
More concretely, in the composite LTS
the behaviour of Opponent is fully specified by expanding the configurations with an extra
evaluation stack, environment and store. 

The new LTS is called \boldemph{composite interaction reduction}\hspace{-1pt}.\hspace{-1pt}
It works
on \emph{composite configurations}
$\configurationF{\EE_P}{\EE_O}{\gamma_P,\gamma_O}{\phi}{S_P,S_O}$, where:
 \begin{compactitem}
  \item $\EE_P,\EE_O$ are evaluation stacks \gjt{(one passive and one active)}; $\gamma_P,\gamma_O$ are environments; and $S_P,S_O$ are stores;
  \item $\phi$ is a common typing function for locations.
 \end{compactitem}
The rules of the composite reduction, given in Appendix~\ref{c:app-comp}, are in effect the P-rules of the ordinary interaction reduction, plus dual forms thereof fleshing out the O-rules.

A trace $t$ is said to be \emph{generated} by a composite configuration $C$
if it can be written as a sequence $(\faction{m_1}{S_1}{\rho_1}) \cdots (\faction{m_n}{S_n}{\rho_n})$ of full moves
\st $C \xRedint{\faction{m_1}{S_1}{\rho_1}} C_1 \xRedint{\faction{m_2}{S_2}{\rho_2}} \ldots \xRedint{\faction{m_n}{S_n}{\rho_n}}
C_n$, in which case we write $C \xRedint{t} C_n$. 
 We say that a composite configuration $C$ \emph{terminates} with the trace $t$, written $C \Downarrow_t$, if
 there exists a store $S$ such that $C \xRedint{t \cdot (\ansO{\unit},S,\emptyPMap)} 
 \configurationF{\emptyStack}{\emptyStack}{\gamma'_P,\gamma'_O}{\phi'_P}{S'_P,S'_O}$.
 
\ntt{We now define how to merge configurations $C_P,C_O$ into a composite one.}
\ntt{For each $X\in\{O,P\}$ we write $X^\bot$ for its dual ($\{X,X^\bot\}=\{O,P\}$), and extend this to $\lambda^\bot=(\_^\bot)\circ\lambda$.} 

\begin{definition}
Given a pair of environments $(\gamma_P,\gamma_O)$ from $\Anl \backslash \Loc$ to values, 
 we say these are \emph{compatible} when:
 \begin{compactitem}[$\bullet$]
  \item $\dom{\gamma_P} \cap \dom{\gamma_O} = \varnothing$,
  \item for all $a \in \dom{\!\gamma_X\!}$ ($X\! \in\! \{\!P,O\}$), $\support{}{\!\gamma_X(a)} \backslash \Loc \subseteq \dom{\!\gamma_{X^\bot}\!}$, 
  \item setting $\gamma^0=\gamma_P \cdot \gamma_O$, and $\gamma^i = \{(a,v\{\gamma\}) \sep (a,v) \in \gamma^{i-1}\}$ $(i > 0)$,
  there is an integer $n$ such that $\support{}{\codom{\gamma^n}} \backslash \Loc = \varnothing$;
 \end{compactitem}
and write $(\gamma_P \cdot \gamma_O)^{*}$ for the environment from $\Anl \backslash\Loc$ to $\Val$
 defined as $\gamma^n$, for the least $n$ satisfying the latter condition above.

A pair of valid configurations $(C_P,C_O)$ are called \boldemph{compatible} if, given 
$C_X=\configuration{\EE_X}{\gamma_X}{\phi_X}{S_X}{\lambda_X}$ (for $X\in\{P,O\}$):
 \begin{compactitem}[$\bullet$]
  \item $\phi_P = \phi_O$ and $\lambda_P = \lambda_O^{\bot}$, 
  \item $(\gamma_P,\gamma_O)$ are compatible and $\dom{\ntt{\gamma_P\cdot\gamma_O}} = \dom{\lambda_P} \backslash \Loc$,
  \item $\dom{S_P} \cap \dom{S_O} = \dom{\lambda_P} \cap \Loc$,
  \item the merge $\configurationF{\EE_P}{\EE_O}{\gamma_P,\gamma_O}{\phi_P}{S_P,S_O}$ of $C_P$ and $C_O$ is valid (cf.\ Appendix~\ref{c:app-comp}).
 \end{compactitem}
We write $\mergeConf{C_P}{C_O}$ for $\configurationF{\EE_P}{\EE_O}{\gamma_P,\gamma_O}{\phi_P}{S_P,S_O}$.
\end{definition}

\ntt{We can merge the (well-typed) evaluation stacks $(\EE_P,\EE_O)$
of compatible configurations by the following operation:}
\[
 \begin{array}{@{}r@{\;\;}@{}l@{\;\;\;}l@{}}
\mergeStack{\emptyStack}{(E,\theta \rightsquigarrow \theta')} \defeq E\qquad
  \mergeStack{\left((M,\theta)::\EE_P\right)}{\EE_O}
  & \defeq &  
     \left(\mergeStack{\EE_P}{\EE_O}\right)[M]\\
  \mergeStack{\left((E,\theta \rightsquigarrow \theta')::\EE_P\right)}{\left((M,\theta)::\EE_O\right)} & \defeq & 
     \left(\mergeStack{\EE_P}{\EE_O}\right)[E[M]]\\
  \mergeStack{\left((E,\theta \rightsquigarrow \theta')::\EE_P\right)}{\left((E',\theta' \rightsquigarrow \theta'')::\EE_O\right)} & \defeq & 
     \left(\mergeStack{\EE_P}{\EE_O}\right)[E'[E]]
 \end{array}
\]
and obtain a correspondence with the operational semantics.



\begin{lemma}
\label{thm:redint-iff-red}
Given $C = \configurationF{\EE_P}{\EE_O}{\gamma_P,\gamma_O}{\phi}{S_P,S_O}$ a valid composite configuration and $\gamma = \gamma_P \cdot \gamma_O$,
there exists a complete trace $t$\linebreak \st $C \Downarrow_t$ iff
$(\mergeStack{\EE_P}{\EE_O}\{\gamma^{*}\},S_P\{\gamma^{*}\}) \red^*\!\! (\unit,S')$
for some $S'$.
\end{lemma}


On the other hand, there is a semantic way to compare Player and Opponent configurations, by 
checking that the traces they generate are compatible.
{Given a trace $t$, let us write $t^\bot$ for its dual obtained by switching the polarity of each move in $t$ (e.g.\ each $\questP{f}{v}$ is changed to $\questO{f}{v}$, and so on).}

\begin{definition}
 Let $C_P$ and $C_O$ be two  configurations.
We write $C_P|C_O \downarrow_t$ when there exists a complete trace $t$
and a store $S$ 
such that $t \in \sem{C_P}$ and $t^\bot\cdot(\faction{\ansO{\unit}}{S}{\emptyPMap}) \in \sem{C_O}$.
\end{definition}

We therefore have the following correspondence between semantic and syntactic composition.

\begin{theorem}
\label{thm:semEq-iff-mergeIRed}
For all pairs of compatible configurations  $C_P$ and $C_O$, $C_P|C_O \downarrow_T$ iff
$\mergeConf{C_P}{C_O} \Downarrow_T$.
\end{theorem}

\subsection{Soundness result}

We need two final pieces of machinery for soundness.
The first one is so-called \emph{ciu-equivalence}, 
which allows one to characterise contextual equivalence by restricting focus to
evaluation contexts.

\begin{definition}\label{def:ciu}
Let $\Sigma$ be a location context.
 Two terms
 $\typingTerm{\Sigma}{\Delta}{\Gamma}{M_1,M_2}{\theta}$
 are ciu-equivalent, written $\ciueqf{\Sigma}{\Delta}{\Gamma}{M_1}{M_2}{\theta}$, when 
 for all typing substitutions $\typingSubst{\cdot}{\delta}{\Delta}$, location contexts $\Sigma' \extend \Sigma$,
 closed stores ${S}:{\Sigma'}$, value substitutions $\typingSubst{\Sigma'}{\gamma}{\Gamma\{\delta\}}$
 and evaluation contexts $\typingCtx{\Sigma'}{\cdot}{E}{\theta\{\delta\}}{\theta'}$, we have
 $(E[{M_1\{\gamma\}\{\delta\}}],S) \Downarrow$ iff $(E[{M_2\{\gamma\}\{\delta\}}],S)\Downarrow$.
\end{definition}

\begin{theorem}
\label{thm:equiv-ciu-ctx}
$\ctxeqf{\Sigma}{\Delta}{\Gamma\!}{\!M_1}{\!M_2}{\theta}$\, iff $\ciueqf{\Sigma}{\Delta}{\Gamma\!}{\!M_1}{\!M_2}{\theta}$.
\end{theorem}

As mentioned at the beginning of this section, we introduce an equivalence on term denotations which includes equality. 
The motivation for this is so as to prune out some distinctions that the model makes between behaviours that are in fact indistinguishable.
More precisely, our model abstracts away any actual values provided by Opponent for polymorphic inputs by names in $\Apol$. 
Moreover, when P plays back one of those names, O is in position to determine precisely which actual value is P returning in reality (as all polymorphic names introduced by O must be distinct). 
This discipline is based on the assumption that O can always instrument the values he provides to P so that he can later distinguish between them. 
It is a valid assumption, apart from the case when later in the trace there is some value disclosure for those polymorphic names which forbids O to implement such instrumentations. 

To remove this extra intensionality from the model, we introduce
an equivalence of traces which blurs out such distinctions:
\begin{compactitem}[$\bullet$]
  \item we first substitute in every \ntt{P-}move all the \emph{O polymorphic names} whose value have been disclosed by their disclosed value;
  \item we then enforce the freshness of \emph{P polymorphic names} played in P moves, which may be broken because of these
  substitutions.
\end{compactitem}
The latter step is implemented via a name-refreshing procedure, defined as follows. Given traces $t,t'$, we say that $t'$ is a \emph{P-refreshing} of $t$, written $t\pref t'$, if $t=t_1\cdot (m,S,\rho)\cdot t_2$, $t'=t_1\cdot t_2'$, with $m$ a P-move, and
there are polymorphic names $p,p'$ such that:
\begin{compactitem}[$\bullet$]
\item  $p\in\support{}{t_1}\cap\support{}{m,S,\codom{\rho}}$ is introduced in a P-move of $t_1$,
\item $p'\notin\support{}{t}$ and $t_2'$ is $(m,S,\rho)\cdot t_2$ 
where we first replace a single occurrence of $p$ in $(m,S,\codom{\rho})$ by $p'$,
and then replace any $[p\mapsto v]$ in the resulting subtrace by $[p\mapsto v]\cdot[p'\mapsto v]$.
\end{compactitem}
P-refreshing is bound to terminate in the traces we examine.
We write $\Freshen{}{}{t}$ for the set of all $t'$ such that $t\pref^* t'$ and $t'\not\pref$.

\cutout{
We first introduce a predicate $\Refresh{X}{c}$, which takes as a parameter a set of names $Y$ and a element $c$ of a nominal set (for example, a move, a store or a partial map),
and return the set of new element $c'$ where all the occurrences of the names in $X$ have been refreshed, so that each occurence of a name in $\support{c} \cap X$
is replaced by a fresh occurrence, together with a function $\psi$ mapping each name $a \in \support{c} \cap X$ to the set of name used to replace $a$ in $c'$.

When $m$ is a Player move, $\Freshen{X}{\psi}{(m,S,\rho) \cdot t}$ is defined as

\[ \left\{ \begin{array}{l}
 (m',S',\rho' \cdot \rho'') \cdot t' \sep ((m',S',\rho'),\psi') \in \Refresh{X}{m,S,\rho},\\ 
 \quad \rho'' = \{(p,v) \sep \exists p' \in \dom{\rho'}. \rho'(p') = v \land  p' \in \psi(p) \},\\
 \quad t' \in \Freshen{X}{\psi \cdot \psi'}{t} 
\end{array} \right\} \] 
Otherwise  it is defined as $(m,S,\rho) \cdot \Freshen{X}{\psi}{t}$

We write $\Freshen{}{}{t}$ for $\Freshen{X}{\emptyseq}{t}$, with $X$ the set of Player polymorphic names of $t$ (i.e.
the polymorphic names introduced in a Player move).
}                                     
                                     
\begin{definition}
Two traces $t_1,t_2$ are said to be equivalent, written $t_1 \equivT t_2$, if $\Freshen{}{}{\xbar{t_1}^\varepsilon} = \Freshen{}{}{\xbar{t_2}^\varepsilon}$, where $\xbar{\,t\,}^{\rho_1\cdots\rho_n}$ is defined as:
\[\xbar{t \cdot (m,S,\rho)}^{\rho_1\cdots\rho_n}\!\!\!\!\! \defeq \begin{cases} 
                                      \xbar{\,t\,}^{\rho_1\cdots\rho_n}\!\! \cdot ((m,S,\rho)\{\rho_1\}\cdots\{\rho_n\}) &\!\! \text{if } m \text{ a P-move} 
\\
                                      \xbar{\,t\,}^{\rho_1\cdots\rho_n\rho}\!\! \cdot (m,S,\rho) 
&\!\! \text{otherwise}
                                     \end{cases}\]
We extend equivalence to sets of traces in an elementwise fashion.
\end{definition}

%
%

\begin{lemma}
\label{thm:traceeq-conf}
 Let $t_1$ be a trace such that $t_1^\bot \in \Tr{(C)}$ with $C$ a valid configuration.
 Then for all $t_2 \equivT t_1$ we have $t_2^\bot \in \Tr{(C)}$.
\end{lemma}

We can now prove the main theorem of this section.

\begin{theorem}[Soundness]
\label{thm:sound}
For all terms $\typingTerm{\Sigma}{\Delta}{\Gamma}{M_1,M_2}{\theta}$, 
 $\sem{{M_1}} \equivT \sem{{M_2}}$
implies $M_1\cong M_2$.
\end{theorem}

\begin{proof}
Suppose $\sem{{M_1}} \equivT \sem{{M_2}}$. 
Using Theorem~\ref{thm:equiv-ciu-ctx}, we prove that ${M_1}\ciueq{M_2}$.
Let us take $\delta,\Sigma'\supseteq\Sigma,S,\gamma$ and $E$ as in Definition~\ref{def:ciu},
and suppose that $(E[M_1\{\gamma\}\{\delta\}],S) \Downarrow$.
\\
Take $(\vec\alpha,\vec l,\vec u) \in \sem{\Delta,\Sigma,\Gamma}$ and write
 $C_{P,1}$ for the P-configuration $\configuration{(M_1\overrightarrow{\subst{x}{\widetilde{u}}},\!\theta)}{\emptyPMap}{\phi}{S}{\lambda}$,
 so  $\inbrax{\typingTerm{\Sigma}{\Delta}{\Gamma}{M_1}{\theta}} \xredint{\faction{\!\!\!\questO{?\,}{\vec\alpha,\vec l,\vec u)}}{S'\!\!}{\,\rho}\!\!\!} C_{P\!,1}$.\!\!\!\!\!
 Let $C_O=\configuration{(E,\theta \rightsquigarrow\theta')}{\gamma' \cdot \delta}{\phi}{S}{\lambda^{\bot}}$
 where $\gamma' = \{(u_i,v_i) \sep \gamma(x_i) = v_i\}$. 
From Lemma~\ref{thm:redint-iff-red}, there exists a complete trace $t$ such that $\mergeConf{C_{P,1}}{C_O} \Downarrow_t$.
 Then, from Theorem~\ref{thm:semEq-iff-mergeIRed}, $C_{P,1} | C_O \downarrow_t$, so that $t \in \Tr{(C_{P,1})}$ and
 $t^\bot \in \Tr{(C_O)}$.
 Writing $C_{P,2}$ for the Player configuration $\configuration{(M_2\overrightarrow{\subst{x}{\widetilde{u}}},\theta)}{\emptyPMap}{\phi}{S}{\lambda}$,
 from the hypothesis of the theorem, there exists a complete trace $t' \equivT t$ such that $t' \in \Tr{(C_{P,2})}$.
 From Proposition~\ref{thm:traceeq-conf}, $t'^{\bot} \in \Tr{(C_O)}$, 
 so that $C_{P,2} | C_O \downarrow_t'$,
 and using Theorem~\ref{thm:semEq-iff-mergeIRed} (in the other direction), we get that $\mergeConf{C_{P,2}}{C_O} \Downarrow_t'$.
 Finally, using Lemma~\ref{thm:redint-iff-red}, we get that $(E[M_2\{\gamma\}\{\delta\}],S) \Downarrow$.
\end{proof}


%% file: fulla.tex
\section{Completeness}

\newcommand\myto{\!\rightsquigarrow}
\newcommand{\good}[2]{\mathrm{good}_{#1}(#2)}
\newcommand{\available}[1]{\mathrm{gtv}(#1)}

\newcommand\getval{\mathsf{getval}}
\newcommand\getv{\mathsf{val}}
\newcommand\names{\mathsf{names}}
\newcommand\fshvals{\mathsf{Fshvals}}
\newcommand\chkval{\mathsf{chkval}}
\newcommand\chkvals{\mathsf{Chkvals}}
\newcommand\setvals{\mathsf{Setstor}}
\newcommand\updvals{\mathsf{Updvals}}
\newcommand\newvals{\mathsf{Newvals}}
\newcommand\fshval{\mathsf{fshval}}
\newcommand\castPk{\mathsf{castPk}}
\newcommand\caster{\mathsf{caster}}
\newcommand\play{\mathsf{Play}}
\newcommand\pairer{\mathsf{pair}}
\newcommand\unpair{\mathsf{unpair}}
\newcommand\ctr{\mathsf{cnt}}
\newcommand\cde{\mathsf{cde}}

While sound, our model fails to be fully abstract as it overestimates the power of O:
the way cast relations ($\sf Cast$) are computed over-approximates the casts that can be implemented by the context in practice, as inhabitation constraints are not taken into account.  For instance, a cast from $\theta\to\theta_1$ to $\theta\to\theta_2$ does not yield one from $\theta_1$ to $\theta_2$ unless a value of type $\theta$ is available.
In this section we restrict our attention to a fragment of \SystemReF, called \SystemReF*, carved in such a way that the above problem cannot be manifested. We then prove our model fully abstract for terms in \SystemReF*.

\SystemReF* is defined by means of restricting the types allowed at the type interface of a term.
In particular, we pose the following restrictions affecting the types which can appear under a $\rm ref$ constructor. First, we do not allow any binders $\forall,\exists$ to appear in the scope of a  $\rm ref$ and, moreover, any type variable $\alpha$ inside a  
 $\refer\theta$ must be {\em reachably inhabited}: in order for a value of type $\refer\theta$ to be played in a trace, a value of type $\alpha$ must have been played before.

Both these restrictions are captured by the following type predicate $\good{\Upsilon}{\theta}$, which determines whether a type $\theta$ is in the defined fragment, assuming that the type variables in $\Upsilon$ are inhabited.\vspace{-.25mm}
\[\vspace{-.5mm}
\begin{array}{l@{\;}l@{\;\;}l} 
  \good{\Upsilon}{\refer\theta} & \defeq & \good{\Upsilon}{\theta}\land \support{\rm T}{\theta} \subseteq \Upsilon \land \theta\text{ is quantifier-free} \\
  \good{\Upsilon}{\theta\rightarrow\theta'} & \defeq & \good{\Upsilon}{\theta} \land \good{\Upsilon \cup \available{\theta}}{\theta'} \\
  \good{\Upsilon}{\forall\alpha.\theta} & \defeq & \good{\Upsilon}{\theta}  \\  
  \good{\Upsilon}{\theta\times\theta'} & \defeq & \good{\Upsilon}{\theta} \land \good{\Upsilon}{\theta'} \\
  \good{\Upsilon}{\exists\alpha.\theta} & \defeq & \good{\Upsilon\cup\{\alpha\}}{\theta}  \\
  \good{\Upsilon}{\theta} & \defeq & \true \quad \text{ otherwise}  
\end{array}
\]
Above, $\available{\theta}$ returns the type variables at the ground level of $\theta$:\vspace{-.75mm}
\[\vspace{-.75mm}
\begin{array}{l@{\;}l@{\;\,}ll@{\;}l@{\;\,}l} 
  \available{\alpha} & \defeq & \{\alpha\} &
  \available{\theta \times \theta'} & \defeq & \available{\theta} \cup \available{\theta'} \\
  \available{\exists \alpha.\theta} & \defeq & \available{\theta} \backslash \{\alpha\}  &
  \available{\refer\theta} & \defeq & \available{\theta}
\end{array}
\]
and  $ \available{\theta} \defeq \varnothing$ otherwise. We extend goodness to type interfaces by setting, given $\Sigma=\{l_1\!:\theta_1,\cdots,l_n\!:\theta_n\}$, 
$\Gamma=\{x_1\!:\!\theta_1',\cdots,x_m\!:\!\theta'_m\}$:\vspace{-.5mm}
\[
\good{}{\Delta;{\Sigma};{\Gamma}\vdash{\theta}} = \good{\emptyset}{(\refer\theta_1\times\cdots\times\refer\theta_n\times\theta_1'\times\cdots\times\theta_m')\to \theta}
\]

\begin{definition}
We let \SystemReF* contain all terms  $\typingTerm{\Sigma}{\Delta}{\Gamma}{M}{\theta}$ such that $\good{}{\Delta;{\Sigma};{\Gamma}\vdash{\theta}}$ holds.
\end{definition}

\begin{example}
The terms form Example~\ref{ex:quiz}\,(2) are not in \SystemReF*, as $\alpha'$ is not inhabited.
The two terms are then equivalent, because Opponent cannot cast $\alpha$ to $\Int$, lacking a value of 
type $\alpha'$ to do so. Our model, however, does not capture this equivalence.
\end{example}

Moreover, we call an initial  configuration 
$\inbrax{\typingTerm{\Sigma}{\Delta}{\Gamma}{M}{\theta}}$ \boldemph{good} 
just if its interface is, while a valid configuration $\inbrax{{\cal E},\gamma,\phi,S,\lambda}$ is good just if, 
taking $X_\lambda=\{\alpha\mid\nu(\lambda)\cap\Apol_\alpha\not=\emptyset\}$,
$\support{\rm T}{\phi}\subseteq X_\lambda$ and
$\good{X_\lambda\!\!}{\theta}$ 
 hold,
for all $\theta\in\codom{\phi}\cup\{\theta\mid\nu(\lambda)\cap\Afun_\theta\not=\emptyset\}$. We can then check that goodness is preserved under reduction. 

Working in this restricted fragment, we can always implement all possible casts anticipated from the cast closure construction of Section~\ref{sec:trace-sem}.
More specifically,
a \boldemph{cast-term} from  $\theta$ to $\theta'$ based
 on \emph{aliased pairs} $(\theta_1,\theta'_1),\ldots,(\theta_n,\theta'_n)$ 
and \emph{inhabited variables} $\alpha_1,\cdots,\alpha_m$  is a term $\cast{}{\theta}{\theta'}$
 such that:
 \begin{compactitem}[$\bullet$]
  \item $\typingTerm{}{\Delta}{\overrightarrow{x_i:\refer\theta_i},\overrightarrow{y_i:\refer\theta'_i},\overrightarrow{z_j:\alpha_j}}{\cast{}{\theta}{\theta'\!}}{\!\theta {\rightarrow} \theta'}$
  \item for any $\Sigma = \{\overrightarrow{l_i:\theta_i}\}$, $p_j\in\Apol_{\alpha_j}$, ${S}:{\Sigma}$ and $\typingTerm{\Sigma'}{\Delta}{}{v}{\theta}$, 
  $(\!(\cast{}{\theta}{\theta'}\overrightarrow{\subst{x_i,y_i}{l_i}}\overrightarrow{\subst{z_j}{p_j}}) v,S) \red^{\!\!*}\! (v',S \cdot S')$ with $v\cong v'$,
 \end{compactitem}\smallskip

\noindent
with $S'$ disjoint of $S$.
Recall now $\RCast{\phi}$ from Definition~\ref{def:Cast} 
and define its restriction
 $\RCastt{\phi}$, the  
closure of $\{(\theta,\theta') \sep \exists l.\, \refer\theta,\refer\theta' \in \phi(l)\}$ using all cast closure rules from Section~\ref{sec:trace-sem} apart from ($*$).

\begin{lemma}
\label{thm:castf}
Let $\phi$ be a valid typing function with
$\vec\alpha$ all free type variables in $\phi$.
Then, for all $(\theta,\theta')\in \RCastt{\phi}$ 
 there is a cast-term $\cast{}{\theta}{\theta'}$ based on pairs $\{(\theta'',\theta''') \sep \exists l.\, \theta'',\theta''' \in \phi(l) \}$ and $\vec\alpha$.
\end{lemma}

The ($*$) rule, though useful for soundness, has no clear way to be implemented with cast-terms, hence the reason for aiming at its exclusion.
The restriction we pose on \SystemReF*\ in that quantifiers cannot appear under a $\rm ref$ constructor renders the rule indeed redundant. 
Each $\phi$ produced in the model contains no types with quantifiers, so that the ($*$) rule can be eliminated (Lemma~\ref{lm:rcast-subst2}). 

The proof of full abstraction is based on a definability result: we show that every complete trace produced by a good P-configuration $C_P$ can be accepted by an appropriately designed O-configuration $C_O$. 
In addition, the given trace is all $C_O$ can accept up to nominal and trace equivalence.
The technique follows e.g.~\cite{laird2007fully}, albeit expanded to the polymorphic setting. Note that the absence of generic types~\cite{Longo93} in our language, because of type disclosure, 
rules out the option of reducing the problem to that for the monomorphic setting.

\begin{theorem}[Definability]
\label{thm:defin}
Let $C_P$ be a good configuration and $t$ a complete trace in $\Tr{(C_P)}$ with final store $S$.
There exists a valid configuration $C_O$ compatible with $C_P$
such that $\Tr{(C_O)}=\{\pi\star t' \sep (\forall a\!\in\!\nu(t)\backslash\nu(C_O).\,\pi(a)\!=\!a)\land \exists t''\! \equivT t\cdot(\faction{\langle\bar{()}\rangle}{S}{\emptyset}).\ t' \prefix t''^{\bot}\}$.
\end{theorem}

We present the main ingredients of the definability argument (cf.\ Appendix~\ref{app:defn}).
We argue by induction on the length of $t$.
Suppose $C_P=\inbrax{\EE_P,\gamma_P,\phi_P,S_P,\lambda}$, 
let $A_0$ be the set of all the names that appear in $t$ and $C_P$.
To determine the types behind the O-type-variables in $A_0$, we define a mapping $\delta$ by collecting all type constraints we can derive from the trace $t$ about O-type-variables,
{mapping to $\Int$ when no such constraints exist} . 
We number P-moves in $t$ in decreasing order, that is, the head move of $t$ has index $\|t\|=(t+1)/2$, and
let $\Theta_{\|t\|}$ recursively \ntt{include all function, reference and variable types that appear in $\phi_P\{\delta\}$ and $\lambda\{\delta\}$.}
At the $i$-th P-move of $t$, this set is updated to $\Theta_i=\{\theta_1^i,\theta_2^i,\cdots,\theta_{ts(i)}^i\}$ by including all the types disclosed in intermediate moves.
%

We use a counter $\ctr$ to determine the position we are in $t$ and
inductively construct $C_O=\inbrax{\EE_O,\gamma_O,\phi_O,S_O,\lambda^\bot}$ with the additional assumptions that:
\begin{compactitem}[$-$]
\item $\EE_O=(E_n,\eta_n\myto\eta_n',\Phi_n)::\cdots::(E_1,\eta_1\!\rightsquigarrow\eta_1',\Phi_1)$,
$n$ is determined from $t$ and
$E_i\equiv (\lambda z.\,{!r_i}(!\ctr)z)\bullet$\,,
for each $i$;
\item $\gamma_O$ obeys $\delta$ \ntt{(i.e.\ $\restrictTr{{\gamma_O}}{{\Atvar}}\subseteq\delta$)} and, moreover, assigns values to each function or pointer name belonging to O by referring to purpose-specific private references in $S_O$:
\begin{compactitem}[$\square$]
\item for each $f$ of arrow type, $\gamma_O(f)=\lambda z.\,{!q_f}(!\ctr)z$
\item for each $g$ of universal type, $\gamma_O(g)=\Lambda\alpha.\,!q_g'(!\ctr)\alpha$
\item for each pointer name $p$ of type $\beta$,
\begin{compactitem}[$\bullet$]
\item if $\gamma_O(\beta)$ an arrow type, $\gamma_O(p)=\lambda z.\,{!q_p}(!\ctr)z$
\item if $\gamma_O(\beta)$ a universal type, $\gamma_O(p)=\Lambda\alpha.\,!q_p'(!\ctr)\alpha$
\item if $\gamma_O(\beta)$ an existential type, $\gamma_O(p)=\inbrax{\alpha',v}$ and $v$ recursively follows the same discipline
\item if $\gamma_O(\beta)$ a product type, $\gamma_O(p)=\inbrax{v_1,v_2}$ and $v_1,v_2$ recursively follow the same discipline
\item if $\gamma_O(\beta)=\Int/\refer\theta$ and the value of $p$ gets disclosed in $t$, $\gamma_O(p)$ is the revealed value; otherwise, $\gamma_O(p)$ is a unique integer/location representing $p$
\item if $\gamma_O(\beta)=\alpha'$, $\gamma_O(p)$ is some polymorphic name respecting the type disclosures in $t$;
\end{compactitem}
\end{compactitem}
\item $\dom{S_O}$ contains $Q_F \uplus Q_F' \uplus Q_P \uplus Q_P' \uplus\{r_1,\cdots,r_n,l_1,\cdots,l_k\} 
  \uplus\{\ctr\} \uplus \{\ell_1,\cdots,\ell_{ts(\|t\|)}\}
  \uplus\{\getval_i\mid i\in [1,\|t\|]\}$;
\end{compactitem}
  where $Q_F$ contains a unique location $q_f$ for each function name $f$ in $\dom{\gamma_O}$, 
  $Q_F'$ contains the $q_g'$'s, $Q_P$ the $q_p$'s, and $Q_P'$ the $q_p'$'s. 

The main engine behind the construction is the use of references to record values played, continuations, functions, and generally all history of $t$ so that O can refer to it in order to: decide to accept each expected move by P, and play the corresponding expected move themselves. Looking at the domain of $S_O$,
the $l_i$'s are the shared locations between $C_P$ and $C_O$,
while 
$\ctr$ is an integer counter that counts the remaining P-moves in $t$.
We set $S_O(\ctr)=\|t\|$.
%
%
%
$L=\{\ell_1,\cdots,\ell_{ts(\|t\|)}\}$ is a set of private auxiliary locations which we shall 
use in order to cast between known types and types obtained by opening existential packages.

The role of the $\getval$'s is to  us to store all names that appear in the trace. 
For each $i$, $\getval_i$ is a location of type:\vspace{-2mm}
\begin{multline*}
\exists\vec{\alpha}.\left((\Int\to\theta_1^i)\times\cdots\times(\Int\to\theta_{ts(i)}^i)\right)\\[-1mm]
\times\left((\Unit\to{\refer\theta^i_1})\times\cdots\times(\Unit\to{\refer\theta^i_{{ts}(i)}})\right)\\[-6mm]
\end{multline*}
where $\vec{\alpha}$ is the sequence of all free type variables in $\Theta_i$.
Thus, the value of $\getval_i$ is an existential package whose first component contains enumerations of all values of type $\theta_j^i$, for each $i,j$. 
These is enough to represent all the available values at each point in the trace.
The second component inside the package stored in $\getval_i$ contains a single reference for each type and we shall {assign} 
to it a special role, namely of holding a private reference from the set $L$.

To see how the above work, let
$t=(\faction{{m}_1}{S_1}{\rho_1})\cdot(\faction{{m}_2}{S_2}{\rho_2})\cdot t'$
and suppose ${m}_1$ is a question $\questP{f}{v}$, introducing fresh type variables $\beta_1,\cdots,\beta_\iota$ (via values of existential type).
We encode acceptance of these first two moves in $q_f$, by setting $S_O(q_f)(\|t\|)$ to be:
\begin{align*}
 \
&\unpack{!\getval_{\|t\|}}{\inbrax{\vec{\alpha}',\inbrax{z',h}}}{}\\
&\letin{z=\castPk\inbrax{z',h}}{}
\\
&\lambda x_0.\, \unpack{N_1}{\inbrax{\beta_1,x_1}}{\cdots\;\unpack{N_{\iota}}{\inbrax{\beta_{\iota},x_{\iota}}}{}}\\
&\qquad\,\letin{\getv=\refer\inbrax{z,\lambda\_.\Omega,\cdots\,,\lambda\_.\Omega}}{}\\ 
&\qquad\,
\ctr\,{-}{-};
\fshvals;\, 
\chkvals;\, 
\newvals;\,\setvals;\,\play\tag{$*$}
\end{align*}
Since the type of
$!\getval_{\|t\|}$ is fully existentially quantified, when we (statically) unpack $!\getval_{\|t\|}$ and get $\vec\alpha',z'$, 
the $\vec\alpha'$ 
are distinct from the type variables $\vec\alpha$ in $\Theta_{\|t\|}$ and, consequently, each component $z_i':\Int\to\theta_i'$ of $z'$ is not of the expected type $\Int\to\theta_i$. 
However, when the unpack will actually happen this mismatch will be resolved.
For visible types (in the game-theoretic \emph{view} sense~\cite{view}), we need this mismatch to also be resolved statically, as we would like to be able to relate the values in $z'$ with $x_0$, any open variables, 
or the  return value of $!q_f$. 
Hence, we employ the $\castPk$ function which casts values of type $\theta_i'$ to $\theta_i$ in $z'$, using the locations in $L$ (each of type $\theta_i$) and their representations in $h$.

Each term $N_i$ is selected in such a way so that, using $\getv$ and $x_0,x_1,\cdots,x_{i-1}$, 
it captures the precise position within $({ m}_1,S_1,\rho_1)$ which introduces the type variable $\beta_i$. 
Note that, here and below, in order to access the values of $\rho_1$ we make use of the 
$\bf cast$ terms of Lemma~\ref{thm:castf}.
We then create the location $\getv$ to contain the old value stores $(z)$, extended with an empty store for each $\beta_i$ ($\lambda\_.\Omega$).
%
Also:
\begin{asparaitem}[$\Box$]
\item $\fshvals$ detects the positions inside $({ m}_1,S_1,\rho_1)$ that introduce fresh names
and updates $\getv$ by adding them as new values in their corresponding types.
This  yields an updated store $S_O'$. 
\item 
$\chkvals$ checks that $x_0$, the public part of $S_O'$ and the values revealed by type disclosure are the ones expected, that is, $v$, $S_1$ and $\rho_1$ respectively. 
{For these comparisons to be implemented, it suffices to focus on variable types only: the rest are either integers/references (can always be checked), or units/functions (no need to check them). 
Variable types belonging to P cannot be checked (P always plays fresh names for them), so we skip them.
Values of variable types $\alpha$ belonging to O will appear e.g.\ in $x_0$ with their instantiated types $\delta(\alpha)$. 
In this case, we are in position to distinguish between function names: these are functions provided by O as polymorphic values so O can pre-instrument so that when
calling them they each produce a unique observable effect. 
}
%
%
\item
$\newvals$ creates all the fresh locations of $({ m}_2,S_2)$ 
and stores them in the corresponding index of $\getv$. Moreover, for each name $f'$ of arrow type
in $({ m}_2,S_2)$,
$\newvals$ includes a code portion creating a reference $q_{f'}$
\ntt{to store a function which takes as an argument the value of the counter
specifying the current move, and returns a function
following the expected behaviour (and that stipulated by the store obtained for $t'$ by the inductive hypothesis).}
%
Similarly for names of universal types. 
Finally,
for each polymorphic O-name $p$ in $({ m}_2,S_2)$ of type $\alpha$,
$\newvals$ includes code creating $q_p$ and adding a function in $\getv$ according to the type $\gamma_O(\alpha)$ (e.g.\ if an arrow type then we add
$\lambda z.\,{!q_{p}}(!\ctr)z$, where $q_p$ encapsulates an effect which allows its recognition in the future).
\item $\setvals$
updates the store in such a way that all the values of $S_2$ are set, while $\play$ is defined by case analysis on ${m}_2$.
\end{asparaitem}
\smallskip

Using Definability, we can now prove the main theorem.

\begin{theorem}[Completeness]
Given \SystemReF*\ terms
$\typingTerm{\Sigma}{\Delta}{\Gamma}{M_1,M_2}{\theta}$, if
$M_1\cong M_2$ then $\sem{M_1}\sim\sem{M_2}$.
\end{theorem}

\cutout{
\begin{proof}
Let us take $\typingSubst{\delta}{\Delta}$, $\Sigma' \extend \Sigma$,
closed store $\typingHeap{S}{\Sigma'}$,
some $(\vec\alpha,\vec l,\vec u) \in \sem{\Delta,\Sigma,\Gamma}$, and write
 $C_{P,i}$ for $\configuration{M_i\overrightarrow{\subst{x}{\widetilde{u}}}}{\emptyPMap}{\phi}{S}{\lambda}$,
  so $\inbrax{\typingTerm{\Sigma}{\Delta}{\Gamma}{M_i}{\theta}} \xredint{\faction{\questO{?\,}{\vec\alpha,\vec l,\vec u)}}{S'\!\!}{\,\rho}} C_{P,i}$.
 Given $t_1\! \in\! \comp{\Tr{(C_{P,1})}}$,
 we build a trace $t_2 \in \comp{\Tr{(C_{P,2})}}$
 such that $t_1 \equivT t_2$.
 
 From Theorem~\ref{thm:defin}, there is a valid configuration $C_O$,
 compatible with $C_{P,1}$, such that $\Tr{(C_O)}$ is equal to $\{t' \sep \exists t'' \equivT t_1. t' \prefix t''^{\bot}\}$.
 Thus, $C_{P,1} | C_O \downarrow_t$, so that from Theorem~\ref{thm:semEq-iff-mergeIRed} and Lemma~\ref{thm:redint-iff-red} 
we get that $(E[M_1\{\gamma\}\{\delta\}],S) \Downarrow$.
 Using Theorem~\ref{thm:equiv-ciu-ctx}, one has that $\ciueqf{\Sigma}{\Delta}{\Gamma}{M_1}{M_2}{\theta}$,
 so that $(E[M_2\{\gamma\}\{\delta\}],S) \Downarrow$.
 
From Lemma~\ref{thm:redint-iff-red} there is $t_2$ such that $\mergeConf{C_{P,2}}{C_O} \Downarrow_{t_2}$.
 Using Theorem~\ref{thm:semEq-iff-mergeIRed}, we get $C_{P,2} | C_O \downarrow_{t_2}$,
 thus $t_2^\bot \in \Tr{(C_O)}$, so $t_2 \equivT t_1$.\!\!\!\!
\end{proof}
}


%% file: biblio.tex
\newcommand\qweqwe{}


%% file: app-syntax.tex
\begin{figure*}[t]\small
\begin{gather*}
  \inferrule*{(x : \theta) \in \Gamma}{\typingTerm{\Sigma}{\Delta}{\Gamma}{x}{\theta}}
\quad
  \inferrule*{(l:\theta) \in \Sigma}{\typingTerm{\Sigma}{\Delta}{\Gamma}{l}{\refer\theta}}
\quad
  \inferrule*{ }{\typingTerm{\Sigma}{\Delta}{\Gamma}{\nb{n}}{\Int}}
\quad
  \inferrule*{\typingTerm{\Sigma}{\Delta}{\Gamma}{M_1}{\Int} \quad 
              \typingTerm{\Sigma}{\Delta}{\Gamma}{M_2,M_3}{\theta}}
             {\typingTerm{\Sigma}{\Delta}{\Gamma}{\ifte{M_1}{M_2}{M_3}}{\theta}} 
\quad
\inferrule*{\typingTerm{\Sigma}{\Delta}{\Gamma}{M_1,M_2}{\Int}}
                {\typingTerm{\Sigma}{\Delta}{\Gamma}{M_1 \oplus M_2}{\Int}}
\\             
  \inferrule*{\typingTerm{\Sigma}{\Delta}{\Gamma}{M}{\theta}  \quad
              \typingTerm{\Sigma}{\Delta}{\Gamma}{N}{\theta'}} 
             {\typingTerm{\Sigma}{\Delta}{\Gamma}{\pair{M}{N}}{\theta \times \theta'}}
\quad
  \inferrule*{\typingTerm{\Sigma}{\Delta}{\Gamma}{M_1}{\refer\theta_1} \quad \typingTerm{\Sigma}{\Delta}{\Gamma}{M_2}{\refer\theta_2}}
             {\typingTerm{\Sigma}{\Delta}{\Gamma}{M_1 = M_2}{\Int}}
\quad  
  \inferrule*{\typingTerm{\Sigma}{\Delta}{\Gamma}{M}{\theta_1 \times \theta_2}} 
             {\typingTerm{\Sigma}{\Delta}{\Gamma}{\proj{i}(M)}{\theta_i}}  
\quad
\inferrule*{\typingTerm{\Sigma}{\Delta}{\Gamma,x:\theta}{M}{\theta'}}{\typingTerm{\Sigma}{\Delta}{\Gamma}{\lambda x.M}{\theta \rightarrow \theta'}}
\\
  \inferrule*
  {\typingTerm{\Sigma}{\Delta,\alpha}{\Gamma}{M}{\theta}}
              {\typingTerm{\Sigma}{\Delta}{\Gamma}{\Lambda\alpha.M}{\forall\alpha.\theta}}
\quad
  \inferrule*{\typingTerm{\Sigma}{\Delta}{\Gamma}{M }{\theta \rightarrow \theta'}\quad
              \typingTerm{\Sigma}{\Delta}{\Gamma}{N}{\theta}} 
             {\typingTerm{\Sigma}{\Delta}{\Gamma}{M N}{\theta'}}
\quad             
 \inferrule*{\typingTerm{\Sigma}{\Delta}{\Gamma}{M }{\forall\alpha.\theta}}
             {\typingTerm{\Sigma}{\Delta}{\Gamma}{M\theta'}{\theta\subst{\alpha}{\theta'}}}
\quad
\inferrule*{\typingTerm{\Sigma}{\Delta}{\Gamma}{M}{\theta\subst{\alpha}{\theta'}}}
             {\typingTerm{\Sigma}{\Delta}{\Gamma}{\pack{\pair{\theta'}{M}}}{\exists \alpha.\theta}}
\\
\inferrule*
          {\typingTerm{\Sigma}{\Delta}{\Gamma}{M}{\exists \alpha.\theta} \\ 
           \typingTerm{\Sigma}{\Delta,\alpha}{\Gamma,x:\theta}{N}{\theta'}}
          {\typingTerm{\Sigma}{\Delta}{\Gamma}{\unpack{M}{\pair{\alpha}{x}}{N}}{\theta'}}
\quad
\inferrule*{\typingTerm{\Sigma}{\Delta}{\Gamma}{M}{\theta}}{\typingTerm{\Sigma}{\Delta}{\Gamma}{\nuref M}{\refer\theta}}             
\quad             
  \inferrule*{\typingTerm{\Sigma}{\Delta}{\Gamma}{M}{\refer \theta}}{\typingTerm{\Sigma}{\Delta}{\Gamma}{!M}{\theta}}
\quad
  \inferrule*{\typingTerm{\Sigma}{\Delta}{\Gamma}{M}{\refer \theta} \quad
              \typingTerm{\Sigma}{\Delta}{\Gamma}{N}{\theta}}
             {\typingTerm{\Sigma}{\Delta}{\Gamma}{M := N}{\Unit}}
\\
\end{gather*}
\caption{Typing Rules of \SystemReF}
\end{figure*}


%% file: appendix-typecast.tex
\section{Properties of Cast Relations}
\label{sec:propcast}

\begin{definition}
 A cast relation $\kappa$ is said to be \boldemph{valid} if for all types $\theta$  there exists a $\theta'$ that is smaller than all types in $\kappa(\theta)$. 
\end{definition}

\begin{lemma}
\label{lm:validcastrel}
 If $\phi$ is a valid typing function, then $\RCast{\phi}$ is a valid cast relation.
\end{lemma}

\begin{lemma}
\label{lm:sym-clos}
 Let $\kappa$ be a casting relation such that for all $(\theta,\theta') \in \kappa, (\theta',\theta) \in \kappa$.
 Then for all $(\theta,\theta') \in \TDclos{\kappa}, (\theta',\theta) \in \TDclos{\kappa}$.
\end{lemma}

\begin{lemma}
\label{lm:sym-RCast}
 Taking $\phi$ a valid typing function, for all $(\theta,\theta') \in \RCast{\phi}, (\theta',\theta) \in \RCast{\phi}$. 
\end{lemma}

\begin{proof}
 Straightforward from Lemma~\ref{lm:sym-clos} and the fact that
 $\{(\theta,\theta') \sep \exists l \in \dom{\phi}, \refer\theta,\refer\theta' \in \phi(l)\}$ is symmetric.
\end{proof}

\begin{definition}
 Two types $\theta,\theta'$ are said two be skeleton-equivalent if:
 \begin{itemize}
  \item $\theta = \theta'$,
  \item $\theta = \refer\theta_0$, $\theta' = \refer\theta'_0$,
  \item or $\theta = \theta_1 \times \theta_2$, $\theta' = \theta'_1 \times \theta'_2$
  and $\theta_i,\theta'_i (i \in \{1,2\}$ are skeleton-equivalent,
  \item or $\theta = \theta_1 \rightarrow \theta_2$, $\theta' = \theta'_1 \rightarrow \theta'_2$
  and $\theta_i,\theta'_i (i \in \{1,2\}$ are skeleton-equivalent,  
  \item $\theta = Q\alpha\theta_0$, $\theta' = Q\alpha\theta'_0$ $(Q \in \{\forall,\exists\})$,
  and $\theta,\theta'$ are skeleton-equivalent.  
 \end{itemize}

\end{definition}

\begin{definition}
Given a set of types $X$, a type $\theta$ is its \boldemph{most general instantiation} if, 
for all $\theta' \in X$, $\theta \leq_\Phi \theta'$, and for all $\theta''$ such that for all $\theta' \in X, \theta'' \leq_\Phi \theta'$, we have $\theta'' \leq_\Phi \theta$. 
We write  $\theta$ as $\MGI{X}$.
\end{definition}

\begin{lemma}
\label{lm:min-rcast}
 Let $\kappa = \RCast{\phi}$ with $\phi$ a valid typing function, then
 for any type $\theta$, $\min(\kappa(\theta))$ is 
 formed by skeleton-equivalent types.
\end{lemma}

\begin{proof}
From Lemma~\ref{lm:validcastrel}, $\RCast{\phi}$ is a valid cast relation.
Thus, $\MGI{\kappa(\theta)}$ exists, let us write it $\theta_0$.
We prove that $\theta_0$ is skeleton equivalent to any type $\theta_1 \in \min(\kappa(\theta))$,
by induction on $\theta_1$:
\begin{itemize}
 \item If $\theta_1$ is equal to a type variable $\alpha$, to $\Int$ or $\Unit$, then since the order is total on type variables (which are the only types
 which can be above $\theta_0$), we have $\min(\kappa(\theta)) = \{\theta_0\}$.
 \item If $\theta_1$ is equal to $\theta^1_1 \times \theta^2_1$, then from the fact that $\theta_0 \leq \theta_1$,
 we get that $\theta_0 = \theta^1_0 \times \theta^2_0$, and the closure properties of $\kappa$
 gives us that $\MGI{\kappa(\theta^i_1)} = \theta^i_0$, so we conclude using the induction hypothesis.
 \item If $\theta_1$ is equal to $\refer\theta'_1$, then from the fact that $\theta_0 \leq \theta_1$,
 we get that $\theta_0 = \refer\theta'_0$. which are indeed skeleton equivalent with $\theta_1$.
\end{itemize}
\end{proof}

%% file: appendix-validconf-bis.tex
\section{Refined Type System}
\label{app-rts}

We prove in this section the safety property of our model (Lemma~\ref{thm:red-typedterm}),
using the standard progress-and-preservation technique.

\begin{lemma}
\label{lm:subst-typing}
 Taking $\typingTermTD{\phi}{\Delta}{\Gamma,x:\theta'}{M}{\theta}$
 and $\typingTermTD{\phi}{\Delta}{\cdot}{v}{\theta'}$,
 we have $\typingTermTD{\phi}{\Delta}{\Gamma}{M\subst{x}{v}}{\theta}$.
\end{lemma}

\begin{proof}
 By induction on the proof of $\typingTermTD{\phi}{\Delta}{\Gamma,x:\theta'}{M}{\theta}$.
\end{proof}

\begin{lemma}
\label{lm:rcast-goodprop}
 Let $\phi$ be a valid typing function, and $\kappa = \RCast{\phi}$.
 Suppose that $\alpha \notin \support{}{\phi}$,
 then 
\begin{itemize}
 \item $\kappa(\alpha) = \{\alpha\}$,
 \item for all $(\theta,\theta') \in \kappa$, 
 $\alpha \in \support{}{\theta}$ iff $\alpha \in \support{}{\theta'}$
 \item for all $(\refer\theta,\refer\theta') \in \kappa$, 
 if $\theta \neq \theta'$ then $\alpha \notin \support{}{\theta,\theta'}$.
\end{itemize}
\end{lemma}

\begin{lemma}
\label{lm:rcast-subst}
Let $\phi$ be a valid typing function with $\alpha\notin\nu(\phi)$, and let $\kappa = \RCast{\phi}$.
If $(\theta,\theta') \in \kappa$
 then $(\theta\subst{\alpha}{\theta_1},\theta'\subst{\alpha}{\theta_1}) \in \kappa$.
\end{lemma}

\begin{proof}
 By induction on $\theta$, using Lemma~\ref{lm:rcast-goodprop} when necessary:
 \begin{itemize}
  \item If $\theta = \alpha$, then $\theta' = \alpha$ since $\kappa(\alpha) = \{\alpha\}$.
  Thus, $\theta\subst{\alpha}{\theta_1}= \theta_1$ and $\theta'\subst{\alpha}{\theta_1} = \theta_1$.
  \item If $\theta = \alpha' \neq \alpha$, then $\alpha \notin \support{}{\theta'}$,
  so that $\theta\subst{\alpha}{\theta_1}= \theta'$ and $\theta'\subst{\alpha}{\theta_1} = \theta'$.
  \item If $\theta = \Int$, then $\alpha \notin \support{}{\theta'}$,
  so that $\theta\subst{\alpha}{\theta_1}= \theta'$ and $\theta'\subst{\alpha}{\theta_1} = \theta'$.
  \item If $\theta = \refer\theta_0$, then  $\alpha \notin \support{}{\theta_0}$, so that
  $\alpha \notin \support{}{\theta'}$,
  thus  $\theta\subst{\alpha}{\theta_1}= \theta'$ and $\theta'\subst{\alpha}{\theta_1} = \theta'$.
  \item If $\theta = \theta^1 \times \theta^2$, then there are two possibilities:
  \begin{itemize}
   \item $\theta' = \alpha' \neq \alpha$, then $\alpha' \notin \support{}{\theta}$,
   thus $\theta\subst{\alpha}{\theta_1}= \theta'$ and $\theta'\subst{\alpha}{\theta_1} = \theta'$.
   \item $\theta' = \theta'^1 \times \theta'^2$, then
   due to the closure property of $\kappa$, we get that $(\theta_i,\theta'_i) \in \kappa (i \in \{1,2\})$,
   so the induction hypothesis gives us that $(\theta^i\subst{\alpha}{\theta_1},\theta'^i\subst{\alpha}{\theta_1}) \in \kappa$,
   thus  $((\theta^1 \times \theta^2)\subst{\alpha}{\theta_1},(\theta'^1 \times \theta'^2)\subst{\alpha}{\theta_1}) \in \kappa$.
  \end{itemize}
  \item If $\theta = \forall \alpha'.\theta_0$, then there are two possibilities:
   \begin{itemize}
   \item $\theta' = \alpha'' \neq \alpha$, then $\alpha \notin \support{}{\theta}$,
   thus $\theta\subst{\alpha}{\theta_1}= \theta$ and $\theta'\subst{\alpha}{\theta_1} = \theta'$.
   \item $\theta' = \forall\alpha'.\theta'_0$, then 
   due to the closure property of  $\kappa$, we get that $(\theta_0, \theta'_0) \in \kappa$.
   Then, applying the induction hypothesis, we get that $(\theta_0\subst{\alpha}{\theta_1}, \theta'_0\subst{\alpha}{\theta_1}) \in \kappa$,
   thus $((\forall\alpha'.\theta_0)\subst{\alpha}{\theta_1}, (\forall\alpha'.\theta'_0)\subst{\alpha}{\theta_1}) \in \kappa$.
  \end{itemize}
 \end{itemize}
\end{proof}

\begin{lemma}
\label{lm:rcast-subst2}
Let $\phi$ be a valid typing function with $\alpha\notin\nu(\phi)$ and suppose all types in $\codom{\phi}$ are quantifier-free. 
If $(\theta,\theta') \in \RCast{\phi}$
 then $(\theta\subst{\alpha}{\theta_1},\theta'\subst{\alpha}{\theta_1}) \in \RCastt{\phi}$. Hence, in particular, $\RCast{\phi}=\RCastt{\phi}$.
\end{lemma}
\begin{proof}
The proof is by induction on $\theta$ and follows exactly that of the previous lemma, apart from the very last subcase.
So let $\theta = \forall \alpha'.\theta_0$, 
with $\theta' = \forall\alpha'.\theta'_0$. Then, since $\phi$ is quantifier-free, in order for $(\theta,\theta')$ to be in $\RCast{\phi}$, it must be that  $(\theta_0, \theta'_0) \in \RCast{\phi}$.
   Then, applying the induction hypothesis, we get that $(\theta_0\subst{\alpha}{\theta_1}, \theta'_0\subst{\alpha}{\theta_1}) \in \RCastt{\phi}$,
   thus $((\forall\alpha'.\theta_0)\subst{\alpha}{\theta_1}, (\forall\alpha'.\theta'_0)\subst{\alpha}{\theta_1}) \in \RCastt{\phi}$.
\end{proof}

\begin{lemma}
\label{lm:subst-TVar-typing}
 Taking $\typingTermTD{\phi}{\Delta,\alpha}{\Gamma}{M}{\theta_0}$
 with
 \begin{itemize}
  \item $\theta'_0 \in \RCast{\phi}(\theta_0)$,
  \item $\alpha \notin \support{}{\phi}$,
  \item $\support{}{M} \cap \ApolT{\alpha} = \varnothing$,
  \item for all $f \in \support{}{M} \cap \AfunT{\theta}$, $\alpha \notin \support{}{\theta}$,
 \end{itemize} 
 we have $\typingTermTD{\phi}{\Delta}{\Gamma\subst{\alpha}{\theta_1}}{M\subst{\alpha}{\theta_1}}{\theta'_0\subst{\alpha}{\theta_1}}$.
\end{lemma}

\begin{proof}
 By induction on the derivation of $\typingTermTD{\phi}{\Delta,\alpha}{\Gamma}{M}{\theta_0}$.
 \begin{itemize}
  \item If $M$ is a variable $x$ such that $(x,\theta_0) \in \Gamma$, then
  $(x,\theta_0\subst{\alpha}{\theta_1}) \in \Gamma\subst{\alpha}{\theta_1}$
  From $\theta'_0 \in \RCast{\phi}(\theta_0)$,, using Lemma~\ref{lm:rcast-subst}, we get that $\theta'_0\subst{\alpha}{\theta_1} \in \RCast{\phi}(\theta_0\subst{\alpha}{\theta_1})$.
  and we conclude using the subtyping rule.
  \item If $M$ is a location $l$ such that $\theta_0 = \refer\theta'_0$ and $(l,\theta'_0) \in \phi$,
  then $\alpha \notin \support{}{\theta'_0}$, so that $\theta_0\subst{\alpha}{\theta_1} = \theta_0$,
  and we conclude using the fact that $\theta'_0 \in \RCast{\phi}(\theta_0)$.
  \item If $M$ is a polymorphic type $p \in \ApolT{\alpha'}$ with $\theta_0 = \alpha'$, then $\alpha' \neq \alpha$
  so that $\theta_0\subst{\alpha}{\theta_1} = \theta_0$
  and we conclude again using the fact that $\theta'_0 \in \RCast{\phi}(\theta_0)$.
  \item If $M$ is a functional name $f \in \AfunT{\theta_0}$, then $\alpha \notin \support{}{\theta_0}$, 
  so that $\theta_0\subst{\alpha}{\theta_1} = \theta_0$
  and we conclude again using the fact that $\theta'_0 \in \RCast{\phi}(\theta_0)$.
  \item If $M$ is an arbitrary term such that $\typingTermTD{\phi}{\Delta,\alpha}{\Gamma}{M}{\theta''_0}$
  with $(\theta''_0,\theta_0) \in \RCast{\phi}$, then
  we get that $\theta'_0 \in \RCast{\phi}(\theta''_0)$ by transitivity, so
  the induction hypothesis gives us that 
  $\typingTermTD{\phi}{\Delta}{\Gamma\subst{\alpha}{\theta_1}}{M\subst{\alpha}{\theta_1}}{\theta'_0\subst{\alpha}{\theta_1}}$.
  \item If $M$ is a pair $\pair{M_1}{M_2}$ such that $\theta_0 = \theta_{0,1} \times \theta_{0,2}$
  and $\typingTermTD{\phi}{\Delta,\alpha}{\Gamma}{M_i}{\theta_{0,i}}$ for $i \in \{1,2\}$,
  then the induction hypothesis gives us that 
  $\typingTermTD{\phi}{\Delta}{\Gamma\subst{\alpha}{\theta_1}}{M_i\subst{\alpha}{\theta_1}}{\theta'_{0,i}\subst{\alpha}{\theta_1}}$
  so that $\typingTermTD{\phi}{\Delta}{\Gamma\subst{\alpha}{\theta_1}}{\pair{M_1}{M_2}\subst{\alpha}{\theta_1}}{\theta'_0\subst{\alpha}{\theta_1}}$.
  \item If $M$ is an application $M' \ N$ such that 
  $\typingTermTD{\phi}{\Delta,\alpha}{\Gamma}{M'}{\theta'_0 \rightarrow \theta_0}$ and
  $\typingTermTD{\phi}{\Delta,\alpha}{\Gamma}{N}{\theta'_0}$,
  then the induction hypothesis gives us that 
  $\typingTermTD{\phi}{\Delta}{\Gamma\subst{\alpha}{\theta_1}}{M'\subst{\alpha}{\theta_1}}{(\theta'_0 \rightarrow \theta_0)\subst{\alpha}{\theta_1}}$
  and $\typingTermTD{\phi}{\Delta}{\Gamma\subst{\alpha}{\theta_1}}{N\subst{\alpha}{\theta_1}}{\theta'_0\subst{\alpha}{\theta_1}}$
  so that $\typingTermTD{\phi}{\Delta}{\Gamma\subst{\alpha}{\theta_1}}{(M' \ N)\subst{\alpha}{\theta_1}}{\theta_0\subst{\alpha}{\theta_1}}$.
  \item If $M$ is a type application $M' \ \theta'_1$ such that $\theta_0 = \theta'_0\subst{\alpha'}{\theta'_1}$ and
  $\typingTermTD{\phi}{\Delta,\alpha}{\Gamma}{M'}{\forall \alpha'.\theta'_0}$,
  then the induction hypothesis gives us that 
  $\typingTermTD{\phi}{\Delta}{\Gamma\subst{\alpha}{\theta_1}}{M'\subst{\alpha}{\theta_1}}{(\forall \alpha'.\theta'_0)\subst{\alpha}{\theta_1}}$
  so that $\typingTermTD{\phi}{\Delta}{\Gamma\subst{\alpha}{\theta_1}}{(M' \ \theta'_1)\subst{\alpha}{\theta_1}}{\theta_0\subst{\alpha}{\theta_1}}$.
  \item If $M$ is a $\lambda$-abstraction $\lambda x.N$ such that $\theta_0 = \theta'_0 \rightarrow \theta''_0$ with
  $\typingTermTD{\phi}{\Delta,\alpha}{\Gamma,x:\theta'_0}{N}{\theta''_0}$
  then the induction hypothesis gives us that 
  $\typingTermTD{\phi}{\Delta}{\Gamma\subst{\alpha}{\theta_1},x:\theta'_0\subst{\alpha}{\theta_1}}{N\subst{\alpha}{\theta_1}}{\theta''_0\subst{\alpha}{\theta_1}}$
  so that $\typingTermTD{\phi}{\Delta}{\Gamma\subst{\alpha}{\theta_1}}{(\lambda x.N)\subst{\alpha}{\theta_1}}{(\theta'_0 \rightarrow \theta''_0)\subst{\alpha}{\theta_1}}$.
  \item If $M$ is a $\Lambda$-abstraction $\Lambda \alpha'.N$ (with $\alpha' \neq \alpha$) such that $\theta_0 = \forall\alpha'\theta'_0$ with
  $\typingTermTD{\phi}{\Delta,\alpha,\alpha'}{\Gamma}{N}{\theta'_0}$
  then the induction hypothesis gives us that 
  $\typingTermTD{\phi}{\Delta,\alpha'}{\Gamma\subst{\alpha}{\theta_1}}{N\subst{\alpha}{\theta_1}}{\theta'_0\subst{\alpha}{\theta_1}}$
  so that $\typingTermTD{\phi}{\Delta}{\Gamma\subst{\alpha}{\theta_1}}{(\Lambda \alpha.N)\subst{\alpha}{\theta_1}}{(\forall\alpha.\theta'_0)\subst{\alpha}{\theta_1}}$.
   \item If $M$ is a pair $\pair{\theta'_1}{N}$ such that $\theta_0 = \exists\alpha'.\theta'_{0}$ (with $\alpha' \neq \alpha$)
  and $\typingTermTD{\phi}{\Delta,\alpha}{\Gamma}{N}{\theta'_0\subst{\alpha'}{\theta'_1}}$,
  then the induction hypothesis gives us that 
  $\typingTermTD{\phi}{\Delta}{\Gamma\subst{\alpha}{\theta_1}}{N\subst{\alpha}{\theta_1}}{(\theta'_0\subst{\alpha'}{\theta'_1})\subst{\alpha}{\theta_1}}$
  so that $\typingTermTD{\phi}{\Delta}{\Gamma\subst{\alpha}{\theta_1}}{N\subst{\alpha}{\theta_1}}{(\theta'_0\subst{\alpha}{\theta_1})\subst{\alpha'}{\theta'_1\subst{\alpha}{\theta_1}}}$.
  We then conclude that $\typingTermTD{\phi}{\Delta}{\Gamma\subst{\alpha}{\theta_1}}{\pair{\theta'_1}{N}\subst{\alpha}{\theta_1}}{(\exists\alpha'.\theta'_0)\subst{\alpha}{\theta_1}}$. 
  \item If $M$ is an unpack $\unpack{M_1}{\pair{\alpha'}{x}}{M_2}$ (with $\alpha' \neq \alpha$) such that $\typingTermTD{\phi}{\Delta,\alpha}{\Gamma}{M_1}{\exists\alpha'.\theta'_1}$
  and $\typingTermTD{\phi}{\Delta,\alpha,\alpha'}{\Gamma,x:\theta'_1}{M_2}{\theta_0}$,
  then the induction hypothesis gives us that
  $\typingTermTD{\phi}{\Delta,\alpha'}{\Gamma\subst{\alpha}{\theta_1},x:\theta'_1\subst{\alpha}{\theta_1}}{M_2\subst{\alpha}{\theta_1}}{\theta_0\subst{\alpha}{\theta_1}}$
  and $\typingTermTD{\phi}{\Delta}{\Gamma\subst{\alpha}{\theta_1}}{M_1\subst{\alpha}{\theta_1}}{(\exists\alpha'.\theta'_1)\subst{\alpha}{\theta_1}}$.
  Thus, $\typingTermTD{\phi}{\Delta}{\Gamma\subst{\alpha}{\theta_1}}{\unpack{M_1\subst{\alpha}{\theta_1}}{\pair{\alpha'}{x}}{M_2\subst{\alpha}{\theta_1}}}{\theta_0\subst{\alpha}{\theta_1}}$.
 \end{itemize}

\end{proof}

\begin{lemma}
\label{lm:decomp-typing}
 Taking $\typingTermTD{\phi}{\Delta}{\cdot}{E[M]}{\theta}$,
 there exists a type $\theta'$ such that 
 $\typingTermTD{\phi}{\Delta}{\cdot}{M}{\theta'}$
 and $\typingCtxTD{\phi}{\Delta}{E}{\theta'}{\theta}$
\end{lemma}

\begin{lemma}[Subject Reduction]
 Taking $\typingTermTD{\phi}{\Delta}{\cdot}{M}{\theta}$ and $\typingHeapTD{S}{\phi}$,
 if $(M,S) \red (M',S')$ then there exists $\phi'$ a typing function with $\dom{\phi} \cap \dom{\phi'} = \varnothing$ such that
 $\typingTermTD{\phi \cup \phi'}{\Delta}{\cdot}{M'}{\theta}$ and $\typingHeapTD{S'}{\phi \cup \phi'}$.
\end{lemma}

\begin{proof}
 We write $M$ as $E[M_0]$ such that $(M_0,S) \red (M'_0,S')$ and $M' = E[M'_0]$.
 From Lemma~\ref{lm:decomp-typing}, there exists a type $\theta_0$ such that
 $\typingTermTD{\phi}{\Delta}{\cdot}{M_0}{\theta_0}$
 and $\typingCtxTD{\phi}{\Delta}{E}{\theta_0}{\theta}$.
 We then reason by induction on the shape of $M_0$ to prove that 
 $\typingTermTD{\phi}{\Delta}{\cdot}{M'_0}{\theta_0}$
 and $\typingHeapTD{S'}{\phi \cup \phi'}$, writing $\kappa$ for $\RCast{\phi}$:
 \begin{itemize}
  \item If $M_0$ is a redex $(\lambda x.N) v$, then $M'_0 = N\subst{x}{v}$ and $S'=S$.
    We then have $\typingTermTD{\phi}{\Delta}{\cdot}{\lambda x.N}{\theta_1 \rightarrow \theta_0}$
    and $\typingTermTD{\phi}{\Delta}{\cdot}{v}{\theta_1}$.
    There exists $\theta'_0,\theta'_1$ such that $(\theta_1 \rightarrow \theta_0,\theta'_1 \rightarrow \theta'_0) \in \kappa$
    and $\typingTermTD{\phi}{\Delta}{x:\theta'_1}{N}{\theta'_0}$.
    Then, from the closure properties $\kappa$, we get that $(\theta_1,\theta'_1) \in \kappa$ and
    $(\theta_0,\theta'_0) \in \kappa$, so using the subtyping rules,
    we get that
    $\typingTermTD{\phi}{\Delta}{x:\theta'_1}{N}{\theta_0}$
    and  $\typingTermTD{\phi}{\Delta}{\cdot}{v}{\theta'_1}$
    and we conclude using Lemma~\ref{lm:subst-typing}.
  \item If $M_0$ is a type redex $(\Lambda \alpha.N) \theta_1$, then $M'_0 = N\subst{\alpha}{\theta_1}$ and $S'=S$.
    We then have $\typingTermTD{\phi}{\Delta}{\cdot}{\Lambda \alpha.N}{\forall\alpha.\theta'_0}$
    with $\theta_0 = \theta'_0\subst{\alpha}{\theta_1}$.
    There exists $\theta''_0$ such that $(\forall\alpha.\theta'_0,\forall\alpha.\theta''_0) \in \kappa$
    and $\typingTermTD{\phi}{\Delta,\alpha}{\cdot}{N}{\theta''_0}$.
    Using Lemma~\ref{lm:subst-TVar-typing}, we get that 
     $\typingTermTD{\phi}{\Delta}{\cdot}{N\subst{\alpha}{\theta_1}}{\theta''_0\subst{\alpha}{\theta_1}}$.
    Then, from the closure properties $\kappa$, we get that
    $(\theta''_0\subst{\alpha}{\theta_1},\theta'_0\subst{\alpha}{\theta_1}) \in \kappa$, so using the subtyping rules,
    we get that
    $\typingTermTD{\phi}{\Delta}{\cdot}{N\subst{\alpha}{\theta_1}}{\theta'_0\subst{\alpha}{\theta_1}}$.
  \item if $M_0$ is the term $\unpack{(\pack{\pair{\theta_1}{v}})}{\pair{\alpha}{x}}{N}$, then $M'_0 = N\subst{\alpha}{\theta_1}\subst{x}{v}$ and
  $S' = S$.
  We have $\typingTermTD{\phi}{\Delta,\alpha}{x:\theta'_0}{N}{\theta_0}$, with $\alpha \notin \support{}{\theta_0}$, and
  $\typingTermTD{\phi}{\Delta}{\cdot}{\pack{\pair{\theta_1}{v}}}{\exists\alpha.\theta'_0}$.
  There exists $\theta''_0$ such that $(\exists\alpha.\theta''_0,\exists\alpha.\theta'_0) \in \kappa$
  and $\typingTermTD{\phi}{\Delta}{\cdot}{v}{\theta''_0\subst{\alpha}{\theta_1}}$.
  Thus, using the closure properties of $\kappa$, we get that $(\theta''_0\subst{\alpha}{\theta_1},\theta'_0\subst{\alpha}{\theta_1}) \in \kappa$
  and using the subtyping rules,
  $\typingTermTD{\phi}{\Delta}{\cdot}{v}{\theta'_0\subst{\alpha}{\theta_1}}$.
  
  Using Lemma~\ref{lm:subst-TVar-typing}, we get that $\typingTermTD{\phi}{\Delta}{x:\theta'_0\subst{\alpha}{\theta_1}}{N\subst{\alpha}{\theta_1}}{\theta_0}$
  and, from Lemma~\ref{lm:subst-typing}, we finally get that $\typingTermTD{\phi}{\Delta}{\cdot}{N\subst{\alpha}{\theta_1}\subst{x}{v}}{\theta_0}$
  \item If $M_0$ is the term $!l$ then $M'_0$ is the value $S(l)$ and $S'=S$.
    There exist $\theta'_0$ such that $(\refer\theta'_0,\refer\theta_0) \in \kappa$ and
    $(l,\theta'_0) \in \phi$.
    From $\typingHeapTD{S}{\phi}$ we get 
    the existence of $\theta''_0$ such that $(\refer\theta''_0,\refer\theta'_0) \in \kappa$
    that $\typingTermTD{\phi}{\Delta}{\cdot}{S(l)}{\theta''_0}$.
    
    By transitivity and the definition of $\phi$, $(\theta''_0,\theta_0) \in \kappa$, thus 
    $\typingTermTD{\phi}{\Delta}{\cdot}{S(l)}{\theta_0}$.
    \item If $M_0$ is the term $l:=v$ then $M'_0$ is the value $\unit$ and $S'=S[l \mapsto v]$.
    Then there exists $\theta_0$ such that 
    $\typingTermTD{\phi}{\Delta}{v}{\theta_0}$.
    Moreover, there exist $\theta'_0$ such that $(\refer\theta'_0,\refer\theta_0) \in \kappa$ and
    $(l,\theta'_0) \in \phi$. Thus, $(l,\theta_0) \in \phi$, so that
     $\typingHeapTD{S'}{\phi}$
  \item If $M_0$ is equal to the term $\refer \ v$, then $M'_0$ is equal to the location $l$ 
    and $S'=S \cdot [l \mapsto v]$. Thus, there exists a type $\theta_1$ such that 
    $\typingTermTD{\phi}{\Delta}{v}{\theta_1}$
    and $(\refer\theta_1,\refer \theta_0) \in \kappa$, so that
    $(\theta_1,\theta_0) \in \kappa$.
    Defining $\phi'$ as $\{(l,\theta_1)\}$, we get that
    $\typingHeapTD{S'}{\phi \cup \phi'}$
    and $\typingTermTD{\phi \cup \phi'}{\Delta}{\cdot}{l}{\refer\theta_0}$.
\end{itemize}
\end{proof}

\begin{lemma}
\label{lm:irred-val}
 Taking $\typingTermTD{\phi}{\Delta}{\cdot}{M}{\theta}$ and $\typingHeapTD{S}{\phi}$,
 such that for all $p \in \support{}{M,S} \cap \ApolT{\alpha}$, $\min{(\RCast{\phi}(\alpha))} = \{\alpha\}$.
 Suppose that $(M,S)$ is irreducible, then either $M$ is a value or a callback.
\end{lemma}

\begin{proof}
\begin{itemize}
 \item  Suppose that $M$ is equal to $E[v_1 + v_2]$ with $\typingTermTD{\phi}{\Delta}{\cdot}{v_1,v_2}{\Int}$.
 If one the $v_i$ is equal to a polymorphic name $p \in \ApolT{\alpha}$, then 
 $\Int \in \RCast{\phi}(\alpha)$, which is impossible since $\min{(\RCast{\phi}(\alpha))} = \alpha$ and $\Int \TRel{\Phi} \alpha$. 
 Thus, the $v_i$ are equal to some integers, which is absurd since $(M,S)$ is irreducible.
 \item Suppose that $M$ is equal to $E[p \ v]$ with $p \in \ApolT{\alpha}$, then 
 $\theta_1 \rightarrow \theta_2 \in \RCast{\phi}(\alpha)$, which is impossible since $\min{(\RCast{\phi}(\alpha))} = \alpha$ and $\theta_1 \rightarrow \theta_2 \TRel{\Phi} \alpha$.
 \item Suppose that $M$ is equal to $E[!p]$ with $p \in \ApolT{\alpha}$, then 
 $\refer\theta \in \RCast{\phi}(\alpha)$, , which is impossible since $\min{(\RCast{\phi}(\alpha))} = \alpha$ and $\refer\theta \TRel{\Phi} \alpha$.
\end{itemize}
\end{proof}

\section{Preservation of Valid Configurations}

We now prove the preservation of the validity of a configuration by the interaction reduction (Lemma~\ref{thm:validconf}).

\begin{lemma}
\label{lm:typing-subst-term}
Suppose that $\typingTermTD{\phi}{\Delta}{\cdot}{M}{\theta}$, then taking 
\begin{itemize}
 \item $\phi' \supseteq \phi$ valid;
 \item $\xi$ a set of polymorphic names;
 \item $(\rho,\phi_\rho) \in \sem{\xi}_{\kappa,\kappa'}$ where $\kappa = \RCast{\phi}$ and $\kappa' = \RCast{\phi'}$,
 with $\phi_\rho \subseteq \phi'$
\end{itemize}
we have $\typingTermTD{\phi'}{\Delta}{\cdot}{M\{\rho\}}{\theta}$.
\end{lemma}

\begin{proof}
  By induction on the structure of the typing judgment $\typingTermTD{\phi}{\Delta}{}{M}{\theta}$.
 \begin{itemize}
  \item If $M$ is a polymorphic name $p \in \ApolT{\alpha} \cap {\xi}$ such that $\theta=\alpha$,
   then $M\{\rho\} = \rho(p)$ which is of type $\theta' \in \min{(\kappa(\alpha))}$.
   We conclude using the subtyping rule, using the crucial fact that $\phi_\rho \subseteq \phi'$ if $\rho(p)$ is a location.
  \item If $M$ is a polymorphic name $p \in \ApolT{\alpha} \backslash \dom{\xi}$ such that $\theta=\alpha$,
   then $M\{\rho\} = p$ which is of type $\alpha$.
  \item If $M$ is a functional name $f \in \AfunT{\theta}$, 
  then $M\{\rho\} = f$ which is of type $\theta$.
  \item If $M$ is a location $l \in \dom{\phi}$ such that $\theta=\refer\theta'$ with $(l,\theta') \in \phi$,
   then from the fact that  $\phi' \supseteq \phi$, $(l,\theta') \in \phi'$.
   \item If $\typingTermTD{\phi}{\Delta}{}{M}{\theta'}$
   with $\theta' \in \RCast{\phi}(\theta)$, then the induction hypothesis gives us that
   $\typingTermTD{\phi'}{\Delta}{}{M\{\rho\}}{\theta'}$, and
   we conclude using the fact that $\theta' \in \RCast{\phi'}(\theta)$, since $\phi' \supseteq \phi$.
   \item The other cases are straightforward by induction.
 \end{itemize} 
\end{proof}

\begin{proof}[Proof of Lemma~\ref{thm:validconf}]
The hard part is to prove that the evaluation stack and the store of $C'$ are well typed.
If $C'$ is a Player configuration, then there exists a configuration
$C''$ such that $C \xredint{} C'' \xredint{\faction{m}{S}{\rho}} C'$.
From Lemma~\ref{thm:red-typedterm}, we get that $C''$ is a valid configuration.
Then, it is straightforward that the evaluation stack and the store of $C'$ are well-typed.
If $C'$ is an Opponent configuration, then $C \xredint{\faction{m}{S}{\rho}} C'$,
and we conclude using Lemma~\ref{lm:typing-subst-term}.
\end{proof}


%% file: app-validir.tex
\section{Well-definedness of the interaction reduction}

In this section, we show that the interaction reduction can always reduces
valid configurations
until configurations with empty stacks are reached,
and that valid Player configurations behave deterministically.
These properties will be used throughout the proof of soundness.

\begin{lemma}\label{lem:unique}
Given $\gamma$ an environment, $\phi$ a typing function and $\kappa$ a cast relation, suppose
that for all $p \in \dom{\gamma} \cap \ApolT{\alpha}$, $\typingTermTD{\phi}{\Delta}{\cdot}{\gamma(p)}{\gamma(\alpha)}$
  and $\gamma(\alpha)$ is lower than all the types in $\RCast{\phi}$.
Then there exist a unique (up to permutation of fresh names) triple $(\rho,\gamma_\rho,\phi_\rho)$ and $\kappa'$ such that
 \begin{itemize}
  \item$ \phi\cup\phi_\rho$ is valid and $\gamma \cdot \gamma_\rho$ is well-defined
  \item $(\rho,\gamma_\rho,\phi_\rho) \in \APEnv{\gamma}{\kappa,\kappa'}$,
  \item $\kappa' = \RCast{\phi\cup\phi_\rho}$ and $\kappa \subseteq \kappa'$.
 \end{itemize}
\end{lemma}

\begin{proof}
 Let us first define $(\rho^0,\gamma^0_\rho,\phi^0_\rho)$ and $\kappa'^0$ as:
 \begin{itemize}
  \item $\kappa'^0 = \RCast{\phi}$
  \item $(\rho^0,\gamma^0_\rho,\phi^0_\rho) \in \APEnv{\gamma}{\kappa,\kappa'^0}$
 \end{itemize}
We firt prove that $\APEnv{\gamma}{\kappa,\kappa'^0}$ is not empty.
Taking $p \in \dom{\gamma} \cap \ApolT{\alpha}$, such that $X_p = \min{\kappa'(\alpha)}\backslash\min{\kappa(\alpha)}$
is not empty,
from Lemma~\ref{lm:min-rcast},
the elements of $X_p$ are skeleton-equivalent. 
Taking $\theta_0 \in X_p$, we build inductively over $\theta_0$ a triple $(v,\gamma_v,\phi_v)$:
\begin{itemize}
 \item If $\theta_0$ is equal to $\Int,\Unit$ or a type variable, then $X_p = \{\theta_0\}$,
 in which case we simply take $(v,\gamma_v,\phi_v) \in \AVal{\gamma(p)}{\theta_0}$.
 \item If $\theta_0$ is equal to some $\refer\theta'_0$, then from the definition of skeleton equivalence,
 we get that $X_p = \{\refer\theta_1,\ldots,\refer\theta_n\}$,
 so that $\gamma(p)$ is a location from the well-typedness of $\gamma(p)$ and the fact that $\gamma(\alpha)$ is lower than all the types in $\RCast{\phi}$. 
 We then take $v= \gamma(p)$, $\gamma_v = \emptyseq$ and $\phi_v = \{(l,\theta_i) \sep i \in \{1,\ldots,n\}\theta \in \}$,
which is indeed valid because of the validity of $\phi$.
\item If $\theta_0$ is equal to some $\theta^1_0 \rightarrow \theta^2_0$, then from the definition of skeleton equivalence,
 we get that $X_p = \{\theta^1_0 \rightarrow \theta^2_0,\ldots,\theta^1_n \rightarrow \theta^2_n\}$,
 so that $\gamma(p)$ is a functional value from the well-typedness of $\gamma(p)$ and the fact that $\gamma(\alpha)$ is lower than all the types in $\RCast{\phi}$. 
 We then take $v= f$, $\gamma_v = [f \mapsto v]$ and $\phi_v = \emptyseq$.
\end{itemize}

For all $i \in \N$, we then define $\kappa'^i = \RCast{\phi\cup \bigcup_{j <i} \phi^j_\rho}$,
and we pick  $(\rho^i,\gamma^i_\rho,\phi^i_\rho) \in \APEnv{\gamma}{\kappa,\kappa'^i}$.
Using the same proof as before, we can show that $\APEnv{\gamma}{\kappa,\kappa'^i}$ is not empty.

Moreover, for all $i \in \N^*$, $\kappa'^i = \kappa'^0$. This is due to the fact that for all $l \in \dom{\phi^i}\backslash\dom{\phi^{i-1}}, \theta,\theta' \in \phi^i(l)$
there exists an $l' \in \dom{\phi^{i-1}}$ such that $\refer\theta,\refer\theta' \in \phi^i(l)$.

This means that we can simply choose $\kappa'$ as $\kappa'^0$, and $(\rho,\gamma_\rho,\phi_\rho)$ as $(\rho^0,\gamma^0_\rho,\phi^0_\rho)$.
\end{proof}

\begin{lemma}
\label{lm:avald-rcast}
  Taking $\kappa$ a cast relation, and $\theta \in \kappa(\theta')$, then $\AValD{u}{\theta}{\RCast{\phi}} = \AValD{u}{\theta'}{\RCast{\phi}}$.
\end{lemma}

\begin{proof}
 This simply comes from the fact that $\min(\kappa(\theta)) = \min(\kappa(\theta'))$.
\end{proof}

\begin{lemma}
Let us write $C = \configuration{(u,\theta)::\EE}{\gamma}{\phi}{S}{\lambda}$ for a valid Player configuration.
Then the set $\{(v,\gamma_v,\phi_v) \sep (v,\gamma_v,\phi_v) \in \AValD{u}{\theta}{\RCast{\phi}}, \dom{\gamma} \cap \dom{\gamma_v} = \varnothing,
 \phi \cup \phi_v \text{ valid}, \land \forall p \in \support{}{v} \exists \alpha. p \in \ApolT{\alpha} \land 
 (\lambda(\alpha) = P \Leftrightarrow p \notin \support{}{\lambda}) 
 \}$
 is a singleton modulo permutation of names.
\end{lemma}

\begin{proof}
From the validity hypothesis, we get that $\typingTermTD{\phi}{\Delta}{\cdot}{u}{\theta\{\gamma_{|\Apol}\}}$
and for all $p \in \support{}{v} \cap \ApolT{\alpha}$ $\min{(\RCast{\phi}(\alpha))} = \{\alpha\}$.
Using Lemma~\ref{lm:min-rcast}, we get that $\min(\RCast{\phi}(\theta))$ is formed by skeleton-equivalent types.
Thus, let us take $\theta_0 \in \min(\RCast{\phi}(\theta))$, using Lemma~\ref{lm:avald-rcast},
we can prove the wanted property directly for $\theta_0$ (rather than $\theta$), by induction on $\theta_0$:
 \begin{itemize}
  \item If $\theta_0$ is a type variable $\alpha$, then $\min(\RCast{\phi}(\theta)) = \{\theta_0\}$, and 
  \begin{itemize}
   \item either $\alpha \in \dom{\gamma}$ (i.e. $\lambda(\alpha) = P$), in which case we can only take $v$ to be a polymorphic name $p \notin \dom{\lambda}$,
   with $\gamma_v = [p \mapsto v]$ and $\phi_v = \emptyseq$.
   \item or $\alpha \notin \dom{\gamma}$ (i.e. $\lambda(\alpha) = O$), in which can only take $v = u$, wich must be a polymorphic name
   since $\typingTermTD{\phi}{\Delta}{\cdot}{u}{\alpha}$ and $\min(\RCast{\phi}(\alpha)) = \{\alpha\}$ (because $\min(\RCast{\phi}(\theta)) = \alpha$).
   We then take $\gamma_v = \emptyseq$ and $\phi_v = \emptyseq$.
 \end{itemize}
  \item If $\theta_0$ is equal to $\Int$, then $\min(\RCast{\phi}(\theta)) = \{\theta_0\}$, and from Lemma~\ref{lm:irred-val} $u$ is an integer. 
     We then take $v=u$, $\gamma_v = \emptyseq$ and $\phi_v = \emptyseq$.
  \item If $\theta_0$ is equal to a functional type, then 
  we simply take $v$ to be a functional name $f \notin \dom{\lambda}$,
   with $\gamma_v = [f \mapsto v]$ and $\phi_v = \emptyseq$.
  \item If $\theta_0$ is equal to $\refer\theta'_0$ then from Lemma~\ref{lm:irred-val} $u$ is a location.
  From the definition of skeleton equivalence, we get that $\min(\RCast{\phi}(\theta)) = \{\refer \theta_1,\ldots,\refer\theta_n\}$
  We then take $v=u$, $\gamma_v = \emptyseq$ and $\phi_v = \{(u,\theta_i) \sep i \in \{1,\ldots,n\}\}$.
  \item If $\theta_0$ is equal to $\theta^1_0\times\theta^2_0$ then from Lemma~\ref{lm:irred-val} $u$ is a pair $\pair{u_1}{u_2}$.
  From the definition of skeleton equivalence, we get that $\min(\RCast{\phi}(\theta)) = \{\theta^1_0\times\theta^2_0,\ldots,\theta^1_n\times\theta^2_n\}$.
  We conclude easily using the induction hypothesis.
 \end{itemize}

\end{proof}

%% file: app-equiv.tex
\section{Properties of equivalence of traces}


\begin{lemma}
 Suppose that
 $(m_1,S_1,\rho_1) \cdot t_1 \equivT (m_1,S_1,\rho_1) \cdot t_1$
 with $m_1,m_2$ two Opponent moves, then
 $(m_1,S_1,\rho_1)$ is equal to $(m_2,S_2,\rho_2)$. 
\end{lemma}



\begin{lemma}
 Let $C = \configuration{\EE}{\gamma}{\phi}{S}{\lambda}$ an Opponent configuration 
 such that $C \xredint{\faction{m_i}{S_i}{\rho_i}} C_i$,
 for $(\faction{m_1}{S_1}{\rho_1}),(\faction{m_2}{S_2}{\rho_2})$ two Opponent moves
 which differs only w.r.t two polymorphic names $p_1,p_2 \in \dom{\gamma}$.
 Suppose that $\gamma_1(p_1)$ differs from $\gamma_2(p_2)$ only w.r.t
 two polymorphic names $p'_1,p'_2$ such that both $\lambda(p'_i) = O$.
 Then for all trace $t_1 \in \Tr{(C_1)}$, there exists a trace $t_2 \in \Tr{(C_2)}$
 which differs from $t_1$ only by some occurrences of $p'_2$ which are
 equal to $p'_1$ in $t_1$. 
\end{lemma}

\begin{proof}
 This is due to the fact that no operations (not even equality testing) can be performed
 on polymorphic names.
\end{proof}


Using this two lemma we can then prove Lemma~\ref{thm:traceeq-conf}.

%% file: app-compconf.tex
\section{Composite configurations}\label{c:app-comp}

 \begin{figure*}
\begin{tabular}{@{\ }l@{\;\;\ }l}
\textsc{(PInt)} &
$\configurationF{(M,\theta)::\EE_P}{\EE_O}{\gamma_P,\gamma_O}{\phi}{S_P,S_O}
    \xredint{\phantom{\faction{\ansP{v}}{S'}{\rho}}} 
    \configurationF{(M',\theta)::\EE_P}{\EE_O}{\gamma_P,\gamma_O}{\phi}{S'_P,S_O}$ \\
& given $(M,S_P) \red (M',S'_P)$  with $(S'_P \backslash S_P) \cap S_O = \varnothing$\\

\textsc{(PA)}&
$\configurationF{(u,\theta)::\EE_P}{(E,\theta \rightsquigarrow \theta')::\EE_O}{\gamma_P,\gamma_O}{\phi}{S_P,S_O}
    \xredint{\faction{\ansP{v}}{S'}{\rho}}
  \configurationF{\EE_P}{(\widetilde{E}[\widetilde{v}],\theta')::\widetilde{\EE_O}}{\gamma'_P,\gamma_O}{\phi\cup\phi'}{S_P,\widetilde{S_O}[\widetilde{S'}]}$ \\
& given $(v,\gamma_{v,P},\rho_v) \in \AVal{u}{\theta}$ \\

\textsc{(PQ)}&
$\configurationF{(E[f\, u],\theta)::\EE_P}{\EE_O}{\gamma_P,\gamma_O}{\phi}{S_P,S_O}
    \xredint{\faction{\questP{f}{v}}{S'}{\rho}}
    \configurationF{(E,\theta' \rightsquigarrow \theta)::\EE_P}{(\widetilde{u'} \ \widetilde{v},\theta')::\widetilde{\EE_O}}{\gamma'_P,\widetilde{\gamma_O}}
     {\phi\cup\phi'}{S_P,\widetilde{S_O}[\widetilde{S'}]}$\\
& given $f \in \dom{\gamma_P} \cap \AfunT{\theta_f}$, $(v,\gamma_v,\phi_v) \in \AVal{u}{\argT{\theta_f}}$, $\theta' = \retT{\theta_f}{v}$ and $\gamma_P(f) = u'$ \\

\textsc{(OInt)}& 
$\configurationF{\EE_P}{(M,\theta)::\EE_O}{\gamma_P,\gamma_O}{\phi}{S_P,S_O}
    \xredint{\phantom{\faction{\ansP{v}}{S'}{\rho}}} 
    \configurationF{\EE_P}{(M',\theta)::\EE_O}{\gamma'_P,\widetilde{\gamma_O}}{\phi}{S_P,\widetilde{S_O}[\widetilde{S'}]}$\\
& given $(M,S_O) \red (M',S'_O)$ with $(S'_O \backslash S_O) \cap S_P = \varnothing$\\

\textsc{(OA)}&
$\configurationF{(E,\theta' \rightsquigarrow \theta):: \EE_P}{(u,\theta')::\EE_O}{\gamma_P,\gamma_O}{\phi}{S_P,S_O}
    \xredint{\faction{\ansO{v}}{S'}{\rho}}
    \configurationF{(\widetilde{E}[\widetilde{v}],\theta)::\widetilde{\EE_P}}{\EE_O}{\widetilde{\gamma_P},\gamma'_O}{\phi\cup\phi'}{\widetilde{S_P}[\widetilde{S'}],S_O}$ \\
& given $(v,\gamma_{v},\phi_{v}) \in \AVal{u}{\theta}$ \\

\textsc{(OQ)}&
$\configurationF{\EE_P}{(E[f u],\theta)::\EE_O}{\gamma_P,\gamma_O}{\phi}{S_P,S_O}
     \xredint{\faction{\questO{f}{v}}{S'}{\rho}}
     \configurationF{(\widetilde{u'}\, \widetilde{v},\eta)::\widetilde{\EE_P}}{(E,\theta' \rightsquigarrow \theta)::\EE_O}
     {\widetilde{\gamma_P},\gamma'_O}{\phi\cup\phi'}{\widetilde{S_P}[\widetilde{S'}],S_O}$\\
& given $f \in \dom{\gamma_O} \cap \AfunT{\theta_f}$, $(v,\gamma_{v},\phi_v) \in \AVal{u}{\argT{\theta_f}}{}$, $\theta' = \retT{\theta_f}{v}$ and $\gamma_O(f) = u'$
\\[2mm]
%
\hline\\[-3mm]
X1 & 
$\phi' = \phi_v \cup \phi_S \cup \phi_{\rho}$ and $\gamma'_X = \gamma_X \cdot \gamma_v \cdot \gamma_S \cdot \gamma_{\rho}$
\\
X2 & for all $f \in \support{\rm F}{S',v,\rho}, f \notin \dom{\gamma_P} \cup \dom{\gamma_O}$
 and for all $\alpha \in \support{\rm T}{S',v,\rho}, \alpha \notin \dom{\gamma_P} \cup \dom{\gamma_O}$  
\\
X3 & $\dom{\phi\cup\phi'} = S_X^*(\support{\rm L}{v,\rho} \cup \dom{\phi})$ 
   and $(S',\gamma_S,\phi_S) \in \AStore{S_{X|\dom{\phi'}}}{\phi'}$;
\\
X4 & for all $p \in \support{\Apol}{S',v,\rho}$ with $p \in \ApolT{\alpha}$:
if $\alpha \in \dom{\gamma_X}$ then $p \notin \dom{\gamma_P} \cup \dom{\gamma_O}$;
if $\alpha \in \dom{\gamma_{X^\bot}}$ then  $p \in \dom{\gamma_{X^\bot}}$
\\
X* &  $(\rho,\gamma_\rho,\phi_\rho) \in \APEnv{\gamma'}{\kappa,\kappa'}$ where
 $\kappa = \RCast{\phi}$, $\kappa = \RCast{\phi \cup \phi'}$ with $\phi \cup \phi'$ valid 
\end{tabular}
\label{fig:cir}
\caption{Composite Interaction Reduction. Rules {\sc(XQ),(XA)} satisfy conditions X1-X4 and X*, for ${\rm X} \in \{{\rm O},{\rm P}\}$, and $\widetilde{Z}=Z\{\rho^*\}\{\gamma_X\}$.}
\end{figure*}

\begin{definition}
A pair of stacks $(\EE_P,\EE_O)$  are \emph{type-compatible} if:
\begin{itemize}
\item  either both are empty; 
\item or $\EE_P$ is active and $\EE_O$ is passive, and
 $\EE_P = (M,\theta'_{n+1})::(E_{P,n},\theta_n \rightsquigarrow\theta'_n)::\ldots::(E_{P,1},\theta_1 \rightsquigarrow\theta'_1)$ and
 $\EE_O = (E_O,\theta'_{n+1} \rightsquigarrow\theta_n)::
   (E_{O,n},\theta'_n \rightsquigarrow \theta_{n-1})::\ldots::
   (E_{O,1},\theta'_1 \rightsquigarrow \theta_0)$;
\item
 or $\EE_P$ is passive and $\EE_O$ is active, and
 $\EE_P = (E_{P,n},\theta_{n+1} \rightsquigarrow \theta'_n)::\ldots::(E_{P,1},\theta_2 \rightsquigarrow \theta'_1)$ and
 $\EE_O = (M,\theta_{n+1})::(E_{O,n},\theta'_n \rightsquigarrow \theta_n)::\ldots::(E_{O,1},\theta'_1 \rightsquigarrow \theta_1)$, 
 (so if $n=0$, then $\EE_P = \emptyStack$ and $\EE_O = (M,\theta_1)$).
\end{itemize}
\end{definition}

\begin{definition}
A composite configuration $C = \configurationF{\EE_P}{\EE_O}{\gamma_P,\gamma_O}{\phi}{S_P,S_O}$ is said to be \emph{valid} when:
 \begin{itemize}
  \item $(\gamma_P,\gamma_O)$ are compatible;
  \item $\dom{S_P} \cap \dom{S_O} = \dom{\phi}$ and for all $l \in \dom{\phi}$, either $S_P(l) = S_O(l)\{\gamma_P\}$ or $S_O(l) = S_P(l)\{\gamma_O\}$;
  \item $(\EE_P,\EE_O)$  are \emph{type-compatible} and both are well-typed for $\phi$ and $\gamma_{P|\Atvar}$ (resp. $\gamma_{O|\Atvar}$).  
 \end{itemize}
\end{definition}

Similarly to Lemma~\ref{thm:validconf}, we can show the following.

\begin{lemma}
\label{lm:validCompConf-red}
 Taking $C$ a valid composite configuration such that $C \xRedint{\faction{m}{S}{\rho}} C'$, then $C'$ is a valid composite configuration.
\end{lemma}
 

It is then possible to decompose a composite configuration into two standard configuration.
\begin{lemma}
\label{lm:decomp-compConf}
 Taking $C$ a valid composite configuration, there exist two compatible configurations $C_P,C_O$ such that
 $C = \mergeConf{C_P}{C_O}$.
\end{lemma}
It is straightforward to prove a correspondence between the composite interaction reduction
of the merge of two compatible configuration $C_P$ and $C_O$, and the (individual) interaction reduction of
$C_P$ and $C_O$, as soon as the resulting configurations are compatible.
\begin{lemma}
\label{lm:mergIRed-iff-IRed}
 Taking $C_P,C_O$ and $C'_P,C'_O$ two pairs of compatible configuration,
 we have  $\mergeConf{C_P}{C_P} \xRedint{\faction{m}{S}{\rho}} \mergeConf{C'_P}{C'_O}$
 iff 
  \begin{itemize}
  \item $C_P \xRedint{\faction{m}{S}{\rho}} C'_P$,
  \item $C_O \xRedint{\faction{m^\bot}{S}{\rho}} C'_O$.
 \end{itemize}
\end{lemma}


\begin{lemma}
\label{lm:semEq-to-IRed}
Let $C_P,C_O$ two compatible configurations, such that
$C_P | C_O \downarrow_t$ with $t =  (\faction{m}{S}{\rho}) \cdot t'$.
Then there exist two compatible configurations $C'_P,C'_O$ such that:
\begin{itemize}
 \item $C_P \xRedint{\faction{m}{S}{\rho}} C'_P$,
 \item $C_O \xRedint{\faction{m^\bot}{S}{\rho}} C'_O$,  
 \item $C'_P|C'_O \downarrow_{t'}$.
\end{itemize}
\end{lemma}
\begin{proof}
The difficulty here is to build $C'_P,C'_O$ which are compatible.
 Suppose that $C_P$ is a Player configuration, so that $C_O$ is an Opponent one and $m$ is a Player move.
 Let us write $C_P = \configuration{\EE_P}{\gamma_P}{\phi_P}{S_P}{\lambda_P}$ and 
 $C_O = \configuration{\EE_O}{\gamma_O}{\phi_O}{S_O}{\lambda_O}$.
 Then there exists a configuration $C'_P$ such that 
 $C_P \xRedint{\faction{m}{S}{\rho}} C'_P$ and $t' \in \Tr{(C'_P)}$.
 Let us write $C'_P = \configuration{\EE'_P}{\gamma'_P}{\phi'}{S'_P}{\lambda'_P}$.
 We can choose $S'_P$ such that $(\dom{S'_P} \backslash \dom{S_P}) \cap \dom{S_O} = \varnothing$.

%
 
 There exists a configuration $C'_O$ such that 
 $C_O \xRedint{\faction{m^\bot}{S}{\rho}} C'_O$ and $t'^\bot \in \Tr{(C'_O)}$.
 Let us write $C'_O = \configuration{\EE'_O}{\gamma'_O}{\phi'}{S'_O}{\lambda'_O}$.
 
 Moreover, we have $S'_O = S_O[S]$, so that $(\dom{S'_O} \backslash \dom{S_O}) \subseteq S$.
 Thus, $\dom{S'_P} \cap \dom{S'_O} = (\dom{S_P} \cap \dom{S_O}) \cup \dom{S}$. 
 From the compatibility of $C_P,C_O$, we have that $\dom{S_P} \cap \dom{S_O} = \dom{\lambda_P} \cap \Loc$.
 and $(\dom{\lambda'_P} \backslash \dom{\lambda_P}) \cap \Loc = S$, thus 
 $\dom{S'_P} \cap \dom{S'_O} = \dom{\lambda'_P} \cap \Loc$.
 
 All the other conditions for $C'_O$ and $C'_P$ to be compatible are straightforward to check.
\end{proof}

\begin{proof}[Proof of Theorem~\ref{thm:semEq-iff-mergeIRed}]
 \textbf{left to right}
 We reason by induction on the length of $t$.
 
 If $t$ is empty, then the evaluation stack of $C_P$ is necessarily empty,
 while the one of $C_O$ has a single element whose term is equal to some term $M$.
 Then, from $C_O \xredint{\faction{\ansP{\unit}}{S}{\emptyPMap}} C'_O$,
 we get that $(M,S_2) \red^* (\unit,S'_2)$, where $S_2$ is the store of $C_O$,
 so that $\mergeConf{C_P}{C_O} \downarrow_t$.
 
 Suppose that $t = (\faction{m}{S'}{\rho}) \cdot t'$, from Lemma~\ref{lm:semEq-to-IRed}
 we get the existence of a pair of compatible configurations $C'_P,C'_O$ such that 
 \begin{itemize}
  \item $C_P \xRedint{\faction{m}{S'}{\rho}} C'_P$,
  \item $C_O \xRedint{\faction{m^\bot}{S'}{\rho}} C'_O$,
  \item $C'_P|C'_O \downarrow_{t'}$
 \end{itemize}
 Then from Lemma~\ref{lm:mergIRed-iff-IRed} we get that 
 $\mergeConf{C_P}{C_O} \xRedint{\faction{m}{S'}{\rho}} \mergeConf{C'_P}{C'_O}$.

 Applying the induction hypothesis on $C'_P|C'_O \downarrow_{t'}$,
 we get that $\mergeConf{C'_P}{C'_O} \Downarrow_{t'}$,
 so that $\mergeConf{C_P}{C_O} \Downarrow_t$.
 
 \textbf{From right to left}
 We reason by induction on the length of $t$.
 If $t$ is empty, then the evaluation stack of $C_P$ is necessarily empty,
 while the one of $C_O$ is equal to a single term $M$, 
 such that $C_O$ internally reduces to a configuration where the evaluation stack
 is formed by a unique term equal to $\unit$.
 This means that there exists $C'_O$ and $S,\rho$ such that $C_O \xRedint{\faction{\ansP{\unit}}{S}{\rho}} C'_O$,
 so $C_P|C_O \downarrow_t \unit$.
  
 Suppose now that $t = (\faction{m}{S'}{\rho}) \cdot t'$ such that $\mergeConf{C_P}{C_O} \xredint{a} C'$,
 from Lemma~\ref{lm:validCompConf-red} we get that $C'$ is compatible, so that from Lemma~\ref{lm:decomp-compConf} 
 there exists two compatible configuration $C'_P,C'_O$ such that 
 $C' = \mergeConf{C_P}{C_O}$.
 Then from  Lemma~\ref{lm:mergIRed-iff-IRed} we get that
 \begin{itemize}
  \item $C_P \xRedint{\faction{m}{S'}{\rho}} C'_P$,
  \item $C_O \xRedint{\faction{m^\bot}{S'}{\rho}} C'_O$.
 \end{itemize}
 Then applying the induction hypothesis on $\mergeConf{C'_P}{C'_O}$,
 we get that $C'_P|C'_O \downarrow_{t'}$, so that
 $C_P|C_O \downarrow_t$.
\end{proof}


%% file: appendix-soundness.tex
\section{Soundness Proof}
\label{app:sound}

\begin{lemma}
 Let $(v,\gamma,\phi) \in \AVal{u}{\theta}$, then $u = v\{\gamma\}$.
\end{lemma}

\begin{lemma}
 Let $(S',\gamma,\phi') \in \AStore{S}{\phi}$, then $S = S'\{\gamma\}$.
\end{lemma}

\begin{lemma}
 Let $(\rho,\gamma',\phi) \in \APEnv{\gamma}{\kappa}$, then
 for all term $M$, $M\{\rho\}\{\gamma'^{*}\} = M\{\gamma^{*}\}$.
\end{lemma}

\begin{proof}[Proof of Lemma~\ref{thm:redint-iff-red}]
\textbf{From right to left.} 
We reason by induction on the length of the trace reducing $C$ to the composite configuration with two empty evaluation stacks.
 Let us suppose that $\EE_P = (M,\theta)::\EE'_P$
and $\EE_O = (E,\theta \rightsquigarrow \theta_1)::\EE'_O$.
If $\configurationF{\EE_P}{\EE_O}{\gamma_P,\gamma_O}{\phi_P}{S}$ internally reduces to
 $\configurationF{(M',\theta)::\EE'_P}{\EE_O}{\gamma_P,\gamma_O}{\phi_P}{S'}$
 then $(M,S) \red (M',S')$ with $(M',S')$ irreducible, 
so that $(E[M]\{\gamma^*\},S\{\gamma^*\}) \red (E[M']\{\gamma^{*}\},S'\{\gamma^{*}\})$.
Using Theorem~\ref{thm:red-typedterm}, we get that:
 \begin{itemize}
  \item Either $M' = E'[f \ u]$ with $f \in \dom{\gamma_P} \cap \AfunT{\theta'}$ where $\theta'$ is a functional type. 
  Then taking $(v,\gamma_v,\phi_v) \in \AVal{u}{\argT{\theta'}}{}$
  we have $C \xredint{\faction{\questP{f}{v}}{S}{\rho}} C'$ where
  $C' = \configurationF{(E',\eta \rightsquigarrow \theta)::\EE'_P}
    {(u'\ \widetilde{v},\eta)::\widetilde{\EE_O}}{\gamma'_P,\widetilde{\gamma_O}}{\phi}{S_P,\widetilde{S_O}[S']}$
  with $\gamma_P(f) = u'$.
 Then,  \[\begin{array}{lll}
          (\mergeStack{((E'[f\ u],\theta)::\EE_P)}{\EE_O})\{\gamma^*\} & = &  (\mergeStack{\EE_P\{\gamma^*\}}{\EE_O\{\gamma^*\}})[E'[f\ u]\{\gamma^*\}]\\
          & = &  (\mergeStack{\EE_P\{\gamma^*\}}{\widetilde{\EE_O}\{\gamma^*\}})[E'[u' \ \widetilde{v}]\{\gamma^*\}]\\
          & = &  (\mergeStack{(E',\eta \rightsquigarrow \theta)::\EE_P}{(u' \ \widetilde{v},\eta)::\widetilde{\EE_O}})\{\gamma^*\}
         \end{array}\]
and we conclude applying the induction hypothesis to $C'$.
 
 \item Otherwise $M'$ is equal to a value $u$.
 Then, taking $(v,\gamma_v,\phi_v) \in \AVal{u}{\theta}$
  we have $C \xredint{\faction{\ansP{f}{v}}{S}{\rho}} C'$
  with $C' = \configurationF{\EE'_P}{(E[\widetilde{v}],\theta_1)::\widetilde{\EE'_O}}{\gamma'_P,\widetilde{\gamma_O}}{\phi}{S_P,\widetilde{S_O}[S']}$.
 Then,  \[\begin{array}{lll}
          (\mergeStack{(u,\theta)::\EE'_P}{(E,\theta \rightsquigarrow \theta_1)::\EE'_O})\{\gamma^*\} 
          & = &  (\mergeStack{\EE'_P\{\gamma^*\}}{\EE'_O\{\gamma^*\}})[E[\widetilde{v}]\{\gamma^*\}]\\
          & = &  (\mergeStack{\EE'_P}{(E[\widetilde{v}],\theta_1)::\widetilde{\EE'_O}})\{\gamma^*\}
         \end{array}\]  
and we conclude applying the induction hypothesis to $C'$. 
 \end{itemize}
 The same reasoning applies when $\EE_O$ is the active evaluation stack and $\EE_P$ is the passive one.

 \textbf{From left to right.}  
We reason by induction on the length of the reduction $(\mergeStack{\EE_P}{\EE_O}\{\gamma^{*}\},S\{\gamma^{*}\}) \red^* (u,S')$.
If $\mergeStack{\EE_P}{\EE_O}\{\gamma^{*}\}$ is equal to a value $u$,
then $\EE_P = \emptyStack$ and $\EE_O$ is equal to a value $u'$ such that $u'\{\gamma^{*}\} = u$.
Thus, $C \xredint{\ansO{v},S,\rho} \configurationF{\emptyStack,\emptyStack}{\gamma_P,\gamma'_O}{\phi'}{S'_P,S'_O}{}$.

Let us now suppose that $\EE_P$ is an active stack $(M,\theta)::\EE'_P$.
If $C$ internally reduce to $\configurationF{(M',\theta)::\EE'_P}{\EE_O}{\gamma_P,\gamma_O}{\phi}{S'_P,S_O}$,
then $(\mergeStack{\EE_P}{\EE_O}\{\gamma^{*}\},S\{\gamma^{*}\}) \red^* (\mergeStack{(M,\theta)::\EE'_P}{\EE_O}\{\gamma^{*}\},S'_P[S_O\{\gamma^{*}\}])$
 and we can apply the induction hypothesis.
Otherwise, $(M,S_P)$ is irreducible, so that:
\begin{itemize}
 \item Either $M$ is equal to a callback $E'[f \ u]$ with $f \in \dom{\gamma_P} \cap \AfunT{\theta'}$ where $\theta'$ is a functional type.
 we have $C \xredint{\faction{\questP{f}{v}}{S}{\rho}} C'$ where
  $C' =  \configurationF{(E',\eta \rightsquigarrow \theta)::\EE'_P}{(u'\ \widetilde{v},\eta)::\widetilde{\EE_O}}
    {\gamma'_P,\widetilde{\gamma_O}}{\phi}{S_P,\widetilde{S_O}[S']}$
  with $\gamma_P(f) = u'$.
\end{itemize}
\end{proof}

%% file: app-fulla.tex
\section{Definability}\label{app:defn}

In this section we prove that our model is fully abstract.
To do so, we rely crucially on the existence of \emph{cast terms} 
ion result for a fragment of \SystemReF\ defined by means of restricting the types allowed in the interface of a term.

\renewcommand{\TConstr} [1] {\mathsf{TC}(#1)}
\newcommand\bff{}

In this section we show that any finite trace $t$ produced by a P-configuration $C_P$ has a matching O-configuration that can actually produce $t$ when combined with $C_P$.

Given a pair $(\rho,S)$,
 we define $\TConstr{\rho,S}$ to be the set of all typing constraints 
on type variables imposed by this pair:
\[
\TConstr{\rho,S} =
\{(\alpha,\theta) \sep  \exists (p,v) \in \rho.\, p \in \ApolT{\alpha}\land \exists\phi, \phi_v \subseteq \phi.\, S:\phi\land (v,\phi_v) \in \sem{\theta}\}
\]
 We extend this definition to a trace $t$, writing $\TConstr{t}$ for the union of all  $\TConstr{\rho,S}$, for all the full moves $(\_,S,\rho)$ of $t$. Thus, for any type variable $\alpha$, $\TConstr{t}(a)$ is the set of types that $\alpha$ has been revealed to be partially instantiated to.
We are going to use this predicate to infer the most general types for Opponent type variables that are played in a trace, using the notion of 
\emph{most general instance} introduced in Appendix~\ref{sec:propcast}, 
where the ordering $\leq$ is taken with respect a given ordering of type variables (in our case, the order of introduction in the trace $t$ that we examine).

\medskip\noindent
\textbf{Theorem~\ref{thm:defin}}~
Let $C_P$ be a good configuration and $t$ a complete trace in $\Tr{(C_P)}$ with final store $S$.
There exists a valid configuration $C_O$ that is compatible with $C_P$
such that $\Tr{(C_O)}=\{\pi\star t' \sep (\forall a\in\nu(t)\backslash\nu(C_O).\,\pi(a)=a)\land (\exists t'' \equivT t\cdot(\faction{\langle\bar{()}\rangle}{S}{\emptyset}).\ t' \prefix t''^{\bot})\}$.

\begin{proof}
Suppose $C_P=\inbrax{\EE_P,\gamma_P,\phi_P,S_P,\lambda}$, 
let $A_0$ be the set of all the names that appear in $t$ and $C_P$, and let 
us write $\|t\|$ for the number of P-moves in $t$, i.e.\ $\|t\|=(|t|+1)/2$.
To determine the types behind the O-type variables in $A_0$, we define a mapping $\delta_0$ using the type constraints predicate defined above. For each such $\alpha$:
\begin{itemize}
\item if $(\alpha,\_) \notin \TConstr{t}$ then we let  $\delta_0(\alpha)=\Int$,
\item otherwise, we take $\delta_0(\alpha)=\MGI{\{\theta \sep (\alpha,\theta) \in \TConstr{t}\}}$.
\end{itemize}
We then take $\delta=\delta_0^*$ to be the recursive closure of $\delta_0$ (as the codomain of $\delta_0$ may itself contain O-type variables).
We number P-moves in $t$ in decreasing order, that is, the head move of $t$ has index $\|t\|$, and
let $\Theta_{\|t\|}$ recursively \ntt{include all function, reference and variable types that appear in $\phi_P\{\delta\}$ and $\lambda\{\delta\}$.}
%
%
At the $i$-th P-move of $t$, this set is updated to $\Theta_i=\{\theta_1^i,\theta_2^i,\cdots,\theta_{ts(i)}^i\}$ by including all the types disclosed in intermediate moves.
The set $\hat\Theta_{\|t\|}\subseteq\Theta_{\|t\|}$ on the other hand is the restriction to \emph{visible types}, in the game-theoretic sense 
(that is, the types that are in scope with respect to the function that P calls/returns at the head move in $t$).
%
%
We shall use a counter $\ctr$ to determine the position we are in $t$ and
construct $C_O=\inbrax{\EE_O,\gamma_O,\phi_O,S_O,\lambda^\bot}$ by induction on the length of $t$, with the additional assumptions that:
\begin{compactitem}[$-$]
\item $\EE_O=(E_n,\eta_n\myto\eta_n',\Phi_n)::\cdots::(E_1,\eta_1\!\rightsquigarrow\eta_1',\Phi_1)$,
$n$ is determined from $t$ and
$E_i\equiv (\lambda z.\,{!r_i}(!\ctr)z)\bullet$\,,
for each $i$; 
\item $\gamma_O$ obeys $\delta$ \ntt{(i.e.\ $\restrictTr{{\gamma_O}}{{\Atvar}}\subseteq\delta$)} and, moreover, assigns values to each function or pointer name belonging to O by following this discipline:
\begin{compactitem}
\item for each $f$ of arrow type, $\gamma_O(f)=\lambda z.\,{!q_f}(!\ctr)z$
\item for each $g$ of universal type, $\gamma_O(g)=\Lambda\alpha.\,!q_g'(!\ctr)\alpha$
\item for each pointer name $p$ of type $\beta$,
\begin{compactitem}
\item if $\gamma_O(\beta)$ an arrow type, $\gamma_O(p)=\lambda z.\,{!q_p}(!\ctr)z$
\item if $\gamma_O(\beta)$ a universal type, $\gamma_O(p)=\Lambda\alpha.\,!q_p'(!\ctr)\alpha$
\item if $\gamma_O(\beta)$ an existential type, $\gamma_O(p)=\inbrax{\alpha',v}$ and $v$ recursively follows the same discipline
\item if $\gamma_O(\beta)$ a product type, $\gamma_O(p)=\inbrax{v_1,v_2}$ and $v_1,v_2$ recursively follow the same discipline
\item if $\gamma_O(\beta)=\Int/\refer\theta$ and the value of $p$ gets disclosed in $t$, $\gamma_O(p)$ is the revealed value; otherwise, $\gamma_O(p)$ is a unique integer/location representing $p$
\item if $\gamma_O(\beta)=\alpha'$, $\gamma_O(p)$ is some polymorphic name respecting the type disclosures in $t$;
\end{compactitem}
\end{compactitem}
\item $\dom{S_O}= Q_F \uplus Q_F' \uplus Q_P \uplus Q_P' \uplus\{r_1,\cdots,r_n,l_1,\cdots,l_k\} 
  \uplus\{\ctr,\cde\} \uplus R\uplus \{\ell_1,\cdots,\ell_{ts(\|t\|)}\}
  \uplus\{\getval_i\mid i\in [1,\|t\|]\}$, \ 
  where $Q_F$ contains a unique location $q_f$ for each function name $f$ in $\dom{\gamma_O}$, 
  $Q_F'$ contains the $q_g'$'s, $Q_P$ the $q_p$'s, and $Q_P'$ is for the $q_p'$'s. The role of $\cde:\Int$ is special.
\end{compactitem}
We set $S_O(\ctr)=\|t\|$.
The $l_i$'s are the locations in $\lambda$, 
and such that $\dom{S_O}\cap A_0\subseteq\{l_1,\cdots,l_k\}$,
$R$ is a set of $r_i$'s that are no longer in use,
while 
$\ctr$ is an integer counter that counts the remaining P-moves in $t$.
%
%
$L=\{\ell_1,\cdots,\ell_{ts(\|t\|)}\}$ is a set of private auxiliary locations which we shall 
use in order to cast between known types and types obtained by opening existential packages.

We next explain the role of the $\getval$ functions that allow us to keep in store all the names that appear in the trace. 
For each $i$, $\getval_i$ is a location of type:
\[
\exists\vec{\alpha}.\left((\Int\to\theta_1^i)\times\cdots\times(\Int\to\theta_{ts(i)}^i)\right)
\times\left((\Unit\to{\refer\theta^i_1})\times\cdots\times(\Unit\to{\refer\theta^i_{{ts}(i)}})\right)
\]
where $\vec{\alpha}$ is the sequence of all free type variables in $\Theta_i$.
Thus, the value of $\getval_i$ is an existential package whose first component contains enumerations of all values of type $\theta_j^i$, for each $i,j$, 
whose names are public (either in $\lambda$ or in some move). 
These represent the available values at each point in the trace and
what we are actually storing are those types whose values are names.
We represent an enumeration of values of type $\theta$ as a function from natural numbers to  $\theta$. 
The second component inside the package stored in $\getval_i$ contains a single reference for each type and we shall {assign} 
to it a special role, namely of holding a private reference from the set $L$.

We now proceed with the proof.
First we treat the inductive case, so let
$t=(\faction{{m}_1}{S_1}{\rho_1})\cdot(\faction{{m}_2}{S_2}{\rho_2})\cdot t'$
and let $C_P'$ be the configuration reached by the interaction reduction after the first two moves, so that $C_P'$ produces $t'$. By applying the induction hypothesis to $C_P'$ 
we obtain an O-configuration $C_O'=\inbrax{\EE',\gamma',\phi',S',\lambda'}$ which produces $t'^\bot\cdot(\faction{\ansP{()}}{S}{\emptyset})$ 
up to the required closures (and satisfies our additional assumptions on its shape). 
We next do a case analysis on ${m}_1$.
For economy, let us write $\Theta_{\|t\|}$ simply as $\Theta=\{\theta_1,\cdots,\theta_{ts}\}$.
\smallskip

\noindent$\blacksquare$~Case of ${m}_1=\questP{f}{v}$. 
Let us suppose the move introduces fresh type variables $\beta_1,\cdots,\beta_{\iota}$, via respective values of existential type.
For all $i<\|t\|$ and all $f',g'$, we set $S_O(q_{f'})(i)=S'(q_{f'})(i)$ and $S_O(q_{g'}')(i)=S'(q_{g'}')(i)$. Moreover,
we set $S_O(q_{g'}')(\|t\|) = \Omega$ for all $g'$, 
$S_O(r_{n})(\|t\|)$, and 
$S_O(q_{f'})(\|t\|) = \Omega$ for all $f'\not=f$, and:
\begin{align*}
S_O(q_f)(\|t\|) = \
&\unpack{!\getval_{\|t\|}}{\inbrax{\vec{\alpha}',\inbrax{z',h}}}{}\\
&\quad\letin{z=\castPk\inbrax{z',h}}{}
\\
&\qquad\lambda x_0.\, \unpack{N_1}{\inbrax{\beta_1,x_1}}{\cdots\;\unpack{N_{\iota}}{\inbrax{\beta_{\iota},x_{\iota}}}{}}\\
&\quad\qquad\qquad\,\letin{\getv=\refer\inbrax{z,\lambda\_.\Omega,\cdots\,,\lambda\_.\Omega}}{}\\ 
&\qquad\qquad\qquad\,
\ctr\,{-}{-};
\fshvals;\, 
\chkvals;\, 
\newvals;\,\setvals;\,\play\tag{$*$}
\end{align*}
Note that, since the type of
$!\getval_{\|t\|}$ is fully existentially quantified, when we (statically) unpack $!\getval_{\|t\|}$ and get $\vec\alpha',z'$, 
the $\vec\alpha'$ 
are distinct from the type variables $\vec\alpha$ in $\Theta$ and, consequently, each component $z_i':\Int\to\theta_i'$ of $z'$ is not of the expected type $\Int\to\theta_i$. 
However, $\vec\alpha'$ is in fact representing $\vec\alpha$, and each $\theta_i'$ represents $\theta_i$, 
and when the unpack will actually happen, $\vec\alpha'$ is going to be elementwise substituted for $\vec\alpha$ and this mismatch will be resolved.
For visible types, though, we need this mismatch to also be resolved statically, as we would like to be able to relate the values in $z'$ with $x_0$, any open variables, 
or the expected return type of $!q_f$. 
Hence, we introduce a $\castPk$ function which dully casts values of type $\theta_i'$ to $\theta_i$ in $z'$, using the locations in $L$ and their representations in $h$:
\begin{align*}
\castPk &\equiv \lambda\inbrax{z_1',\cdots,z_{ts}',h_1,\cdots,h_{ts}}.\,
\letin{\,\overrightarrow{z_i={\sf Z}_i}\,}{\,\inbrax{z_1,\cdots,z_{ts}}}\\
{\sf Z}_i &\equiv\begin{cases}
z_i' & \text{ if }\theta_i\notin\hat\Theta_{\|t\|}\\
\lambda x.\,h_i:=\lambda\_.z_i'x;\,{!\ell_{i}\unit} & \text{ if }\theta_i\in\hat\Theta_{\|t\|}
\end{cases}
\end{align*}
where, for each $i$, $\ell_{i}$ is the location of type $\Unit\to\theta_i$ that we keep in $L$.

Each term $N_i$ is selected in such a way so that, using $\getv$ and $x_0,x_1,\cdots,x_{i-1}$, 
it captures the precise position within $({\bff m}_1,S_1,\rho_1)$ which introduces the type variable $\beta_i$. 
As a result, $\getv$ and $x_0,\cdots,x_\iota$ contain all the values that have been played so far, including the ones just played in the last move. 
Note that $\getv$ extends $z$ with default value functions ($\lambda\_.\Omega$) for the types newly introduced (because of fresh $\{\beta_1,\cdots,\beta_\iota\}$).
\ntt{
Moreover, the inspection of the values in $\rho_1$ requires to utilise an additional class of cast functions, which implement the casts stipulated by the cast relations of Section~\ref{sec:trace-sem} and are defined recursively by use of the $\bf cast$ terms (Lemma~\ref{thm:castf}). These casts are available due to the inhabitation properties of \lang*, and will also be implicitly used in the sequel for accessing $\rho_1$. }

We next analyse the macros of line ($*$) above. 
\begin{asparaitem}[$\Box$]
\item $\fshvals$ detects the positions inside $({\bff m}_1,S_1)$ that introduce fresh names
and updates $\getv$ by adding them as new values in their corresponding types.
Note here that this procedure may require to use functions like $\castPk$ above. 
For instance, if some $\refer\beta$ is a visible type with $\refer(\beta\times\alpha)$ not visible (due to $\alpha$ not being visible), 
then values of type $\refer(\beta\times\alpha)$ would be essentially accessed via $z'$, in which they would have type $\refer(\beta'\times\alpha)$. 
If a location $l:\refer(\beta\times\alpha)$ is played in $({\bff m_1},S_1)$ then, 
in order to retrieve the new polymorphic name of type $\beta$  introduced in it, 
call it $p$, we would need to reach $l$ via $\getv$ (i.e.\ via $z'$), dereference it and project to obtain $p:\beta'$, 
and then cast the latter to $\beta$ using a $\castPk$ function.
These updates result into an updated store $S_O'$. 

Similar uses of $\castPk$ are implicitly involved in the functions presented below without special mention.
\cutout{\item $\updvals$ updates $\getv$ by
any new values that materialise because of type disclosures occurring due to $({\bff m_1},S_1)$.
\nt{say something about casting?}}
\item 
$\chkvals$ checks that $x_0$, the public part of $S_O'$ and the disclosed values are the ones expected, that is, $v$, $S_1$ and $\rho_1$ respectively. 
\ntt{For these comparisons to be implemented, the main task is to check all values of  variable types in e.g.\ $x_0$: the rest are either integers/references (can always be checked), or units/functions (no need to check them). 
Variable types belonging to P cannot be checked (and $P$ will always play a fresh polymorphic name for them), so we can skip them.
On the other hand, values of variable types $\alpha$ belonging to O will appear in $x_0,S_O'$ with their instantiated types $\delta(\alpha)$. In this case, while we are still not able to distinguish between P-polymorphic names (and this is taken care of by the $\cal F$ closure applied in the definition of trace equivalence), we are now in position to distinguish between function names: these are functions provided by O as polymorphic values so O can pre-instrument them in a way that, calling them with a default argument they each produce a unique observable effect. This effect will be an assignment to the variable $\cde$. In order to call these special functions of type, say, $\theta_1\to\theta_2$ we need to have $\theta_1$ inhabited with (at least) a default value. More than that, in order for these functions to return, we need to be able to inhabit $\theta_2$ as well. Such inhabitation conditions are guaranteed by our working in the restricted type fragment of \SystemReF*.}

%
\item
$\newvals$ creates all the fresh locations of $({\bff m}_2,S_2)$ 
and stores them in the corresponding index of $\getv$. Moreover, for each arrow functional name $f'$ 
of $({\bff m}_2,S_2)$,
$\newvals$ includes the code:
\begin{align*}
&\letin{q_{f'}=\refer(\lambda c.\,\mathsf{case}\ c\ \mathsf{of}\
[
\;i\mapsto S'(q_{f'})(i) \ \text{ if }0<i<\|t\|
\;\mid\,\; \_\mapsto\Omega\;])\ }{}\\
&\quad\letin{f'=\lambda z.\,{!q_{f'}}(!\ctr)z\, }{\,
\getv[\text{add } f']}
\end{align*}
where by $\getv[\text{add } f']$ we update $\getv$ to include $f'$ in its corresponding index. 
For each polymorphic O-name $p$ of $({\bff m}_2,S_2)$, with type $\alpha$ and $\gamma_O(\alpha)=\theta_{i_0}=\theta_{i_1}\to\theta_{i_2}$, 
$\newvals$ includes the code:
\begin{align*}
&\letin{q_{p}=\refer(\lambda c.\,\mathsf{case}\ c\ \mathsf{of}\ [\\
&\qquad\qquad\qquad\qquad\;\;\;\,
i\mapsto \lambda\_.\,\cde:=\hat p;\pi_{i_2}\!({!\getv})1\ \text{ if }0<i<\|t\|\land\text{$p$ abstract at $i$}\\
&\qquad\qquad\qquad\qquad\ \mid
i\mapsto \pi_{i_0}\!({!\getv})\hat f'\ \text{ if }0<i<\|t\|\land\text{$p$ known as $f'$ at step $i$}\\
&\qquad\qquad\qquad\qquad\
\mid \_\mapsto\Omega \; ])\ 
}{}\\
&\quad\letin{p=\lambda z.\,{!q_{p}}(!\ctr)z\, }{\,
\getv[\text{add } p]}
\end{align*}
where $\hat p$ is an integer representing $p$,
and $\hat f'$ is the index of $f'$ in the $i_0$-component of $\getval_i$.
We say $p$ is \emph{known as} $f'$ at step $i$ if there has been a disclosure of $p\mapsto f'$ in some $\rho$ component in $t$ before the $i$th P-move (otherwise, $p$ is \emph{abstract}).
The code in $q_p$  allows us to recognise whether a function $v$ played by P is actually $p$: we apply $v$ to $\pi_{i_1}\!(!\getv)1$ and then check if the value of $\cde$ has been set to $\hat p$. 
We work similarly for each universal function name and each polymorphic name mapped to a universal value.

Finally, $\newvals$ 
creates the new members of $L$ and
updates $\getval_{\|t\|-1}$ with its implementation:
\begin{align*}
&\letin{\ell_{ts+1}=\refer(\lambda\_.\Omega)}{\,\cdots\,\letin{\ell_{ts'}=\refer(\lambda\_.\Omega)}{}}\;
\getval_{\|t\|-1}:=\pack{\inbrax{\vec{\alpha}\vec\beta,\inbrax{\getv,\inbrax{{\sf L_1},\cdots,{\sf L}_{ts'}}}}}
\end{align*}
where each ${\sf L}_i$ is either $\ell_i$, if the latter is visible, or the corresponding $h_i$ otherwise.
\item $\setvals$
updates the store in such a way that all the values of $S_2$ are set. 
\item $\play$ is defined further below by case analysis on ${m}_2$.
\end{asparaitem}
\noindent$\blacksquare$~Case of ${\bff m}_1=\questP{g}{\alpha}$. 
%
For all $f,g'$ and $i<\|t\|$, we set $S_O(q_{f})(i)=S'(q_{f})(i)$ and $S_O(q_{g'}')(i)=S'(q_{g'}')(i)$. Moreover,
$S_O(q_{f})(\|t\|) = \Omega$ for all $f$, 
and also $S_O(r_{n})(\|t\|)=\Omega$, and 
$S_O(q_{g'}')(\|t\|) = \Omega$ for all $g'\not=g$, and:
\begin{align*}
S_O(q'_g)(\|t\|) = \;\,
&\unpack{!\getval_{\|t\|}}{\inbrax{\vec{\alpha}',\inbrax{z',h}}}{\\
&\quad\letin{z=\castPk\inbrax{z',h}}{}}
\\
&\qquad\Lambda\beta_1.\
\unpack{N_2}{\inbrax{\beta_2,x_2}}{\cdots\;\unpack{N_{\iota}}{\inbrax{\beta_{\iota},x_{\iota}}}{}}\\%
&\qquad\qquad\quad\letin{\getv=\refer\inbrax{z,\lambda\_.\Omega,\cdots\,,\lambda\_.\Omega}}{}\\
&\qquad\qquad\qquad\ctr\,{-}{-};\,\fshvals;\, \updvals;\, \chkvals;\, 
\newvals;\,\setvals;\,\play
\end{align*}
where the given operators are defined as in the previous case.

\noindent$\blacksquare$~Case of ${\bff m}_1=\ansP{v}$. 
In this case we stipulate that,
for all $f,g$ and $i<\|t\|$, $S_O(q_{f})(i)=S'(q_{f})(i)$ and $S_O(q_{g}')(i)=S'(q_{g}')(i)$, 
and $S_O(q_{f})(\|t\|)=S_O(q_{g}')(\|t\|)=\Omega$. The following continuation is enacted on $r_n$ and $\|t\|$:
\begin{align*}
S_O(r_n)(\|t\|) = \;\,
&\unpack{!\getval_{\|t\|}}{\inbrax{\vec{\alpha}',\inbrax{z',h}}}{\\
&\quad\letin{z=\castPk\inbrax{z',h}}{}}
\\
&\qquad\lambda x_0.\,\unpack{N_1}{\inbrax{\beta_1,x_1}}{\cdots\;\unpack{N_{\iota}}{\inbrax{\beta_{\iota},x_{\iota}}}{}}\\
&\qquad\quad\letin{\getv=\refer\inbrax{z,\lambda\_.\Omega,\cdots\,,\lambda\_.\Omega}}{}\\
&\qquad\qquad\ctr\,{-}{-};\,\fshvals;\, \updvals;\, \chkvals;\, 
\newvals;\,\setvals;\,\play
\end{align*}
where the given operators are defined as in the first case above. 

\medskip
\noindent
The code for $\play$ depends on the move ${\bff m}_2$, on which we also do a case analysis.

\noindent$\blacksquare$~Case of ${\bff m}_2=\questO{f}{v}$. Suppose $f=\pi_k(S_O'(!\getv))(j)$ and let the last stack component of ${\cal E}'$ be $(E_{n'},\theta_{m'}\myto\theta_{m''},\Phi_{n'})$. Set:
\begin{align*}
\play\equiv\ &
\letin{f=\pi_k(!\getv)j}{}\\
&\quad\letin{r=\refer (\lambda c.\,\mathsf{case}\ c\ \mathsf{of}\
[\,
 i\mapsto S'(r_{n'})(i) \ \text{ if }0<i<\|t\|\;\,\mid\;\, \_\mapsto\Omega\,])}{}\\
&\qquad{(\lambda z.\,{!r}(!\ctr)z)}
(fM)
\end{align*}
where $M\equiv v\{\gamma'\}$.

\noindent$\blacksquare$~Case of ${\bff m}_2=\questO{g}{\alpha}$. Suppose $g=\pi_k(S_O'(\getv))(j)$ and $\gamma'(\alpha)=\theta$. Set:
\begin{align*}
\play\equiv\ &
\letin{g=\pi_k(!\getv)j}{}\\
&\quad\letin{r=\refer (\lambda z.\,\mathsf{case}\ z\ \mathsf{of}\
[\,
 i\mapsto S'(r_{n'})(i) \ \text{ if }0<i<\|t\|\;\,\mid\;\, \_\mapsto\Omega\,])}{}\\
&\qquad{(\lambda z.\,{!r}(!\ctr)z)}
(g\,\theta)
\end{align*}
\noindent$\blacksquare$~Case of ${\bff m}_2=\ansO{v}$. Set
$\play\equiv v\{\gamma'\}$. 

\medskip\noindent
Finally, for the base case, if $t=(\faction{\ansP{v}}{S}{\rho})$ then we work as in the answer case above.
\end{proof}

\begin{proof}[Proof of Lemma~\ref{thm:castf}]
 By induction on the proof that $\theta' \in \kappa(\theta)$.
 \begin{itemize}
  \item If there exists an $l \in \dom{\phi}$ such that $\theta,\theta' \in \phi(l)$, then we define
  $\cast{}{\theta}{\theta'}$ as the function $\mathtt{\lambda z:\theta. let \ a = !x in x:=z;let \ b = !y \ in \ x:=a; b}$. 
  \item If there exists a type $\theta''$ such that $\theta'' \in \kappa(\theta)$ and $\theta' \in \kappa(\theta'')$,
  then we define $\cast{}{\theta}{\theta'}$ as $\lambda z:\theta.\cast{}{\theta''}{\theta'}(\cast{}{\theta}{\theta''}z)$.
  \item If $\refer\theta' \in \kappa(\refer\theta)$, then we define $\cast{}{\theta}{\theta'}$
  as $\mathtt{\lambda z:\theta. let \ a = \refer z \ in \ !(\cast{}{\theta}{\theta'} z)}$.
  \item If there exist four types $\theta_1,\theta'_1,\theta_2,\theta'_2$ such that $\theta = \theta_1 \times \theta_2$ and $\theta' = \theta'_1 \times \theta'_2$
  with $\theta'_1 \in \kappa(\theta_1)$ and $\theta'_2 \in \kappa(\theta_2)$, 
  then we define $\cast{}{\theta}{\theta'}$ as $\lambda z:\theta.\pair{\cast{}{\theta_1}{\theta'_1} \proj{1} z}{\cast{}{\theta_2}{\theta'_2} \proj{2} z}$.
  \item If there exist two types $\theta_2,\theta'_2$ s.t. $(\theta'\times\theta'_2) \in \kappa(\theta\times\theta_2)$,
  then we define $\cast{}{\theta}{\theta'}$ as $\lambda z:\theta.\proj{1}\cast{}{\theta\times\theta_2}{\theta'\times\theta'_2}$.
  \item If there exist four types $\theta_1,\theta'_1,\theta_2,\theta'_2$ such that $\theta = \theta_1 \rightarrow \theta_2$ and $\theta' = \theta'_1 \rightarrow \theta'_2$
  with $\theta_1 \in \kappa(\theta'_1)$ and $\theta'_2 \in \kappa(\theta_2)$ 
  then we define $\cast{}{\theta}{\theta'}$ as $\lambda z:\theta.\lambda z_1:\theta'_1. \cast{}{\theta_2}{\theta'_2} (z (\cast{}{\theta'_1}{\theta_1} z_1))$.
  \item If there exist two types $\theta_2,\theta'_2$ such that $(\theta\rightarrow\theta'_2) \in \kappa(\theta'\rightarrow\theta_2)$,
  then we define $\cast{}{\theta}{\theta'}$ as \[\mathtt{\lambda z:\theta.let \ y = \refer (\lambda \_.\Omega_{\theta'}) \ in \
  let \ \_ \ =  (\cast{}{(\theta\rightarrow\theta_2)}{(\theta'\rightarrow\theta'_2)} \lambda x:\theta'.y:=x;a) z \ in \ (!y)()}\]
  where $\mathtt{a}$ is a value of type $\theta_2$, which exists because of the hypothesis.
  \item If there exist two types $\theta_1,\theta'_1$ such that $(\theta'_1\rightarrow\theta') \in \kappa(\theta_1\rightarrow\theta)$,
  then we define $\cast{}{\theta}{\theta'}$ as $\mathtt{\lambda z:\theta.(\cast{}{(\theta_1\rightarrow\theta)}{(\theta'_1\rightarrow\theta')} \lambda x:\theta_1.z) a}$
  where $\mathtt{a}$ is a value of type $\theta'_1$, which exists because of the hypothesis.
\cutout{  \item If there exists two types $\theta_1,\theta'_1$ such that $\theta = \theta_1\subst{\alpha}{\alpha'}$, $\theta' = \theta'_1\subst{\alpha}{\alpha'}$
  with $\exists\alpha.\theta'_1 \in \kappa(\exists\alpha.\theta_1)$,
  then we define $\cast{}{\theta}{\theta'}$ as $\lambda z:\theta. \unpack{(\cast{}{\exists\alpha.\theta_1}{\exists\alpha.\theta'_1}\pack{\pair{\alpha}{z}})}{\pair{\alpha'}{x}}{x}$.
  \gj{Problems here !!}
}  
  \item If there exist two types $\theta_1,\theta'_1$ such that $\theta = \exists\alpha.\theta_1$, $\theta' = \exists\alpha.\theta'_1$
  with $\theta'_1 \in \kappa(\theta_1)$ and $\kappa(\alpha) = \{\alpha\}$,
  then we define $\cast{}{\theta}{\theta'}$ as $\lambda z:\theta. \unpack{z}{\pair{\alpha}{x}}{\pack{\pair{\alpha}{\cast{}{\theta_1}{\theta'_1}x}}}$.

    \item If there exist two types $\theta_1,\theta'_1$ such that $\theta = \forall\alpha.\theta_1$, $\theta' = \forall\alpha.\theta'_1$
  with $\theta'_1 \in \kappa(\theta_1)$ and $\kappa(\alpha) = \{\alpha\}$,
  then we define $\cast{}{\theta}{\theta'}$ as $\lambda z:\theta. \Lambda \alpha.\cast{}{\theta_1}{\theta'_1} (z \alpha)$.
  
 \end{itemize}

\end{proof}